\documentclass[english,prl,superscriptaddress,reprint,twocolumn]{revtex4-1}
\usepackage{standalone}
\usepackage{nccfoots}
\usepackage{comment}

\usepackage{graphics}
\usepackage{float}
\graphicspath{{figures/}}
\usepackage{braket}
\usepackage{color}
\usepackage{nicefrac}
\usepackage{bbold}
\usepackage{mathtools}

\usepackage[colorlinks=true,
   linkcolor=blue,
   citecolor=blue,
   urlcolor =blue]{hyperref}

\newcommand{\citeSM}{\cite[{\tiny SM}\kern-0.3em][]{SM}}

\newcommand{\be}{\begin{equation}}
\newcommand{\ee}{\end{equation}}

\newcommand*{\ketbra}[2]{\ensuremath{\ket{#1}\bra{#2}}}

\newcommand{\tr}{\mathrm{Tr}}
\newcommand{\du}{\! \text{d}u \, }

\newcommand{\ps}[1]{\hat{s}_{#1}}

\newcommand*{\trAE}[2][]{\ensuremath{\textnormal{Tr}_{#1}\left[ #2 \right]}}
\newcommand*{\expect}[2][]{\ensuremath{\mathbb{E}_{#1}\left[ #2 \right]}}
\newcommand{\bbe}[2][]{\mathbb{E}_{#1}\left[ #2 \right]}

\expandafter\let\csname equation*\endcsname\relax

\expandafter\let\csname endequation*\endcsname\relax

\usepackage{amsmath,amsfonts,amssymb,amsthm, bbm,braket}
\usepackage{tikz}
\usepackage{xcolor}

\DeclareMathOperator{\var}{\mathrm{Var}}

\definecolor{Red}{HTML}{E53E30}  % "Pantone 179"
\definecolor{Green}{HTML}{00AD69}  % "Pantone 3405"
\definecolor{Blue}{HTML}{2171b5}
\definecolor{Purple}{HTML}{652F6C}  % "Pantone 520"

\begin{document}
\title{Importance sampling of randomized measurements for probing entanglement}
%and for quantum verification}

\author{Aniket Rath}
\affiliation{Univ.  Grenoble Alpes, CNRS, LPMMC, 38000 Grenoble, France}

\author{Rick van Bijnen}
\affiliation{Center for Quantum Physics, University of Innsbruck, Innsbruck A-6020, Austria}	
\affiliation{Institute for Quantum Optics and Quantum Information of the Austrian Academy of Sciences,  Innsbruck A-6020, Austria}

\author{Andreas Elben}
\affiliation{Center for Quantum Physics, University of Innsbruck, Innsbruck A-6020, Austria}	
\affiliation{Institute for Quantum Optics and Quantum Information of the Austrian Academy of Sciences,  Innsbruck A-6020, Austria}

\author{Peter Zoller}
\affiliation{Center for Quantum Physics, University of Innsbruck, Innsbruck A-6020, Austria}	
\affiliation{Institute for Quantum Optics and Quantum Information of the Austrian Academy of Sciences,  Innsbruck A-6020, Austria}

\author{Beno\^it Vermersch}
\affiliation{Univ.  Grenoble Alpes, CNRS, LPMMC, 38000 Grenoble, France}
\affiliation{Center for Quantum Physics, University of Innsbruck, Innsbruck A-6020, Austria}	
\affiliation{Institute for Quantum Optics and Quantum Information of the Austrian Academy of Sciences,  Innsbruck A-6020, Austria}

% \begin{abstract}
% We show how to drastically reduce the number of measurements in randomized measurement protocols using importance sampling. We obtain drastic reduction of statistical errors using classical techniques of machine-learning and tensor networks, which only use partial information  on the quantum state.
% \textcolor{orange}{Our numerical results indicate  that we can double the subsystem sizes for which entanglement can be measured.
% In particular, we show an exponential reduction of the required number of measurements to estimate the purity of product states and GHZ states.}
% \end{abstract}

\begin{abstract}
We show that combining randomized measurement protocols with importance sampling allows for characterizing entanglement in significantly larger quantum systems and in a more efficient way than in previous work. A drastic reduction of statistical errors is obtained using classical techniques of machine-learning and tensor networks using partial information on the quantum state. In current experimental settings of engineered many-body quantum systems this significantly increases the \mbox{(sub-)}system sizes for which entanglement can be measured.
In particular, we show an exponential reduction of the required number of measurements to estimate the purity of product states and GHZ states.
\end{abstract}
\maketitle

Measuring the properties of many-body  states, and in particular quantifying entanglement for increasing system sizes is a key challenge in assessing and utilizing the power of large-scale quantum computers~\cite{Arute2019} and simulators~\cite{Ebadi2021,Scholl2021}.  The recent development of \textit{randomized measurements} provides us with a general toolbox to measure in a state-agnostic way physical quantities associated with entanglement~\cite{VanEnk2012,Cong2016,Elben2018,Vermersch2018,Elben2018a,Knips2019,Ketterer2019,Huang2020,Elben2020,Zhou2020,Ketterer2020,Ketterer2020a,Vitale2021,Satoya2021,Rath2021b}, scrambling~\cite{Vermersch2019,Qi2019,Garcia2021}, topological order~\cite{Elben2019, Cian2020}, and in cross-device quantum verification~\cite{Elben2020a}. Randomized measurements are particularly well suited to current experimental settings, requiring only (random) single qubit rotations and site-resolved measurements.  Moreover, estimations are made directly from the measured data, with low number of measurements compared to tomography~\cite{Gross2010}. These protocols have enabled in recent experimental work the measurement of (entanglement) R\'enyi entropies~\cite{Brydges2019,Vitale2021}, negativities~\cite{Elben2020}, state-fidelities~\cite{Elben2020a}, and scrambling~\cite{Joshi2020}. 

While these experiments have been performed in the regime of subsystems with ten particles, the ongoing development of quantum systems involving hundreds of qubits~\cite{Arute2019,Ebadi2021,Scholl2021} raises the challenge to scale these protocols to significantly larger (sub-)system sizes. The current bottleneck is the required number of measurements to overcome statistical errors: For instance, the number of randomized measurements to estimate the purity with a given accuracy is of the order of $2^{aN}$ for a (sub-)system of $N$ qubits, with $a \approx 1$~\cite{Elben2018,Huang2020}. In this letter, we show that importance sampling will allow us to push randomized protocols to study significantly larger (sub-)system sizes.
%\ar{\sout{in some cases effectively doubling the number of particles which can be studied}}
In particular, our scaling analysis for product states and GHZ states shows that the required number of measurements $2^{a'N}$ has a reduced exponent $a'<a$ compared to our previous `uniform' sampling approach. 
We also observe below significant  reductions of statistical errors when estimating with importance sampling the purity of random states, and highly entangled states created by a quantum quench.

\begin{figure}[t]
    \includegraphics[width = 0.5\textwidth]{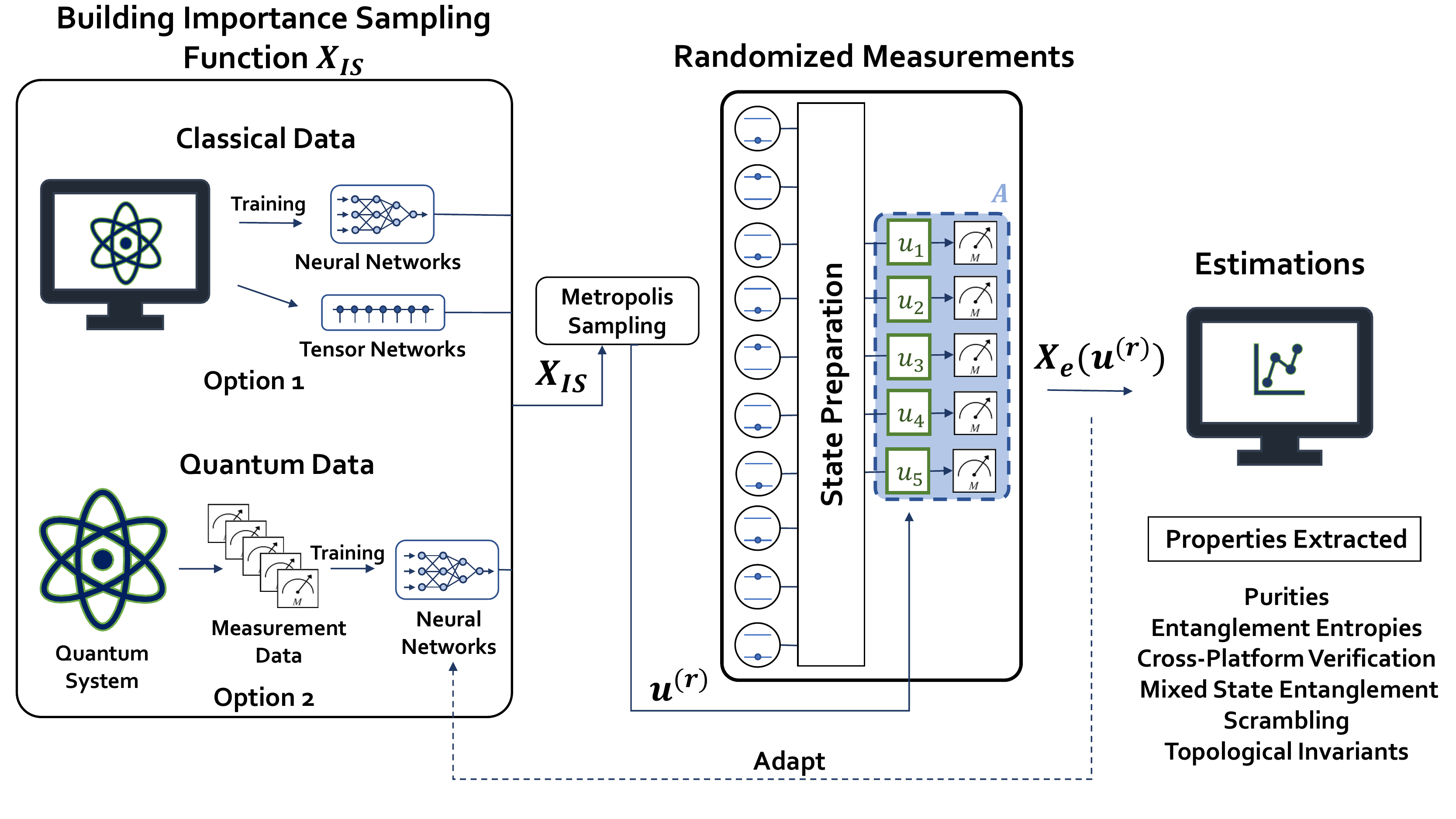}
    \caption{{\it Randomized measurement protocol with importance sampling}. In the first phase, we construct a classical function $X_\mathrm{IS}(u)$. In the second phase, unitaries are sampled from the appropriate classical representation. In the last phase, measurements are performed in the quantum system, and are analyzed to construct different properties accessible by randomized measurements. The measurement data obtained during the experiment could also be considered as additional samples to obtain improved classical function for future experiments.}
    \label{fig:setup}
\end{figure}

While our approach can be realized in any randomized measurement protocol, we consider for concreteness the situation of probing entanglement for a bipartite quantum system $A$ and $B$. Our aim is to measure the purities $p_2=\mathrm{Tr}(\rho^2)$, and second R\'enyi entropies $S_2=-\log(p_2)$ of a subsystem $A$ of $N$ qubits described by a reduced density matrix $\rho$. The values of $p_2$ and $S_2$ can be used to quantify entanglement~\cite{Horodecki1996}, but also to unravel universal aspects of many-body quantum matter~\cite{Eisert2010}.
Using the protocol presented in this Letter, the number of measurements to access the purity in existing setups can be exponentially reduced, allowing for instance to probe topological order on large-scale surface codes~\cite{satzinger2021realizing}, or to verify in a state-agnostic way large quantum circuits~\cite{Elben2020a,zhu2021crossplatform}.

The purity $p_2$ can be written as an integral \mbox{$p_2=\int X(u) du$} of the quantity~\cite{Elben2018,Elben2018a,Brydges2019}
\begin{equation}
X(u) = 2^N\sum_{s,s'} (-2)^{-D[s,s']} P_u(s) P_u(s') , \label{eq:RM}
\end{equation}
with the integration performed over all local unitary transformations $u=u_1\otimes \dots \otimes u_N$, with respect to the Haar measure $du  = \Pi_i\, du_i$ (see Supplemental Material (SM)~\cite{SM}\nocite{Diaconis2015,Brydges2019,NRecipes2007}). Here, $P_u(s)=\bra{s}u\rho u^\dag \ket{s}$ are the probabilities of measuring a particular bitstring $s$ in the computational basis after rotation $u$ (c.f.~Fig~\ref{fig:setup}), and $D$ is the Hamming distance. In practice, the purity can be evaluated using a Monte Carlo integration $p_2  \approx (N_u)^{-1} \sum_r X(u^{(r)})$, obtained by uniformly sampling a finite number of local transformations $u^{(r)}=u_1^{(r)}\otimes \dots \otimes u_N^{(r)} $ ($r = 1, \dots, N_u$).

Statistical errors in the estimation of the purity are due to both shot noise (the finite number of measurement samples $N_M$ used to estimate the probabilities $P_u(s)$), and to the finite number of transformations $N_u$. 
The challenge to overcome statistical errors is related to the fact that the function $X(u)$ takes values in an exponentially large interval $[2^{-N}, 2^N]$ (see SM~\cite{SM}).
Here, we propose to sample unitaries from a  distribution $p_\mathrm{IS}$ that prioritizes the  `important' regions of $X$ giving larger contributions to the total integral (\ref{eq:RM}), and we write the purity as
\begin{equation}
    p_2 = \int \left( \frac{X(u)}{p_\mathrm{IS}(u)}\right) p_\mathrm{IS}(u) du.\label{eq:IS}
\end{equation}
The gain in estimating the purity via Monte Carlo integration with importance sampling becomes apparent when quantifying the  statistical error $\mathcal{E}$ in measuring $p_2$ for $N_M\to \infty$ with a finite number of unitaries $N_u$, being of the order of $\mathrm{std}_\mathrm{IS}(X/p_\mathrm{IS})/\sqrt{N_u}$, 
when compared with uniform sampling $\mathrm{std}(X)/\sqrt{N_u}$~\cite{NRecipes2007}. Here $\mathrm{std}$ and $\mathrm{std}_\mathrm{IS}$ are the standard deviations according to the Haar measure $du$, and the distribution $p_\mathrm{IS}(u)du$, respectively.

Our protocol is summarized in Fig.~\ref{fig:setup}. 
(i) {\it Building $X_\mathrm{IS}$}:
We first construct on a classical computer an approximation $X_\mathrm{IS}(u)$ of the function $X(u)$. This function can be built based on partial information on the quantum state (classical data). We can also form  $X_\mathrm{IS}(u)$ from measurements performed on a quantum system (quantum data). This can be data from prior experiments under study, but could also be data from another experiment, potentially a more noisy quantum device or platform running the same quantum task.
% we can also think of  using another experimental platform producing approximately the same state with higher repetition rates (e.g., a GHZ state created in a superconducting qubit experiment being used to prepare the measurement of  the purity of this GHZ state in a trapped ion experiment).
(ii) {\it Sampling}: We define a probability distribution $p_\mathrm{IS}(u) = |X_\mathrm{IS}(u)|/\int |X_\mathrm{IS}(u)| du$~\footnote{In the examples below, we have $X_\mathrm{IS}\approx X(u)>0$.}, and sample a set of $N_u$ random unitaries via the Metropolis algorithm~\cite{NRecipes2007}.
(iii) {\it Measurements}: For each $u^{(r)}$, we collect $m=1,\dots,N_M$ bitstrings $s_m^{(r)}$
from randomized measurements performed on the quantum device.
(iv) {\it Estimation}: 
As the bistrings $s_m^{(r)}$ are distributed according to the probabilities $P_{u^{(r)}}(s)$, we use Eq.~\eqref{eq:RM}, and construct an unbiased estimation of $X(u^{(r)})$
\begin{equation}
    X_e(u^{(r)})=\frac{2^N}{N_M(N_M-1)}
    \sum_{m \neq m'} (-2) ^{-D[s_m^{(r)},s^{(r)}_{m'}]},
    \label{eq:ShotNoise}
\end{equation}
which only differs from $X(u^{(r)})$ due to shot noise. 
Averaging $(X_e(u^{(r)})/p_\mathrm{IS}(u^{(r)}))$ over the unitaries $u^{(r)}$, \mbox{$r=1,\dots,N_u$}, we obtain an estimation of the purity $[p_2]_\mathrm{IS}$.
%Note that, in the spirit of Bayesian optimization~\cite{Movckus1975}, the steps (i)-(iv) can be repeated multiple times, where the measurement step (iv) is used to update the prior distribution  $X_\mathrm{IS}(u)$ (see also Fig.~\ref{fig:setup}).\arcom{Could drop this sentence as we discuss this idea in the outlook?}

Importance sampling reduces the total required number of measurements $N_uN_M$ associated with a given statistical error $\mathcal{E}$. 
 When sampling unitaries $u$ according to $p_\mathrm{IS}$, we first reduce the required number of unitaries $N_u$ to achieve $\mathcal{E}$ in the limit $N_M\to \infty$, as discussed above.
In addition, the number of shots $N_M$ required to satisfy an error threshold is also less compared to uniform sampling. 
The intuition behind this result is that the unitaries $u$ sampled according to $p_\mathrm{IS}$ are preferentially chosen in the vicinity of the maximum of $X$, where the effect of shot noise is minimal. 
For instance, with a product state, the maximum value of $X(u)$ is obtained when the distribution is peaked as $P_u(s)=\delta_{s,s_0}$ (see SM~\cite{SM}), i.e., when one shot only $N_M=1$ is sufficient to obtain convergence $X_e(u)=X(u)$.  When estimating the purity by averaging $X_e(u)$ over $p_\mathrm{IS}$, we indeed numerically observe, for product and GHZ states, an exponential reduction of the required value of $N_M$.

Task (i) of our protocol is the crucial part governing the efficiency of our protocol. If the quantum state can be represented classically up to unknown decoherence effects, such as for the product state, or a GHZ state, we can build a quasi-exact representation  $X_\mathrm{IS}(u)$ of $X(u)$.
Our protocol is also relevant when only approximations $X_\mathrm{IS}(u)$ of $X(u)$ are available, for instance if we have only access to a mean-field or a variational wavefunction. 
In particular, we show below that tensor networks~\cite{Schollwock2011}, which, with limited bond dimension, cannot faithfully represent a highly entangled state, are indeed useful to access the purity with reduced number of measurements compared to uniform sampling.
Similarly, when building $X_\mathrm{IS}(u)$ from quantum data, we can use recent tomographic techniques~\cite{Gross2010,Cramer2010,Torlai2018,Torlai2019,Kokail2021}, even in situations when they do not accurately represent the quantum state. 

The rest of this letter presents a detailed recipe to build the approximation $X_\mathrm{IS}(u)$ from limited information on the state, as well as performance tests and scaling analyses of statistical errors with various quantum states. 

{\it Building the sampler $X_\mathrm{IS}$---}To construct $X_\mathrm{IS}$, we assume we have access to a finite number $N_\mathrm{samples}$ of random measurements $X_a(u^{(k)})$, $k=1,\dots, N_\mathrm{samples}$. These measurements can be obtained from  {\it classical data}, i.e., from a representation of the state on a classical computer. 
 $X_a(u^{(k)})$ is only an approximation of the true measurement $X(u^{(k)})$. This can be due to unknown decoherence effects, but also to fundamental reasons that limit our ability to represent classically a quantum state. For instance, we can consider that $X_a$ is generated by a mean-field, variational tensor-network methods~\cite{Schollwock2011} (e.g.,  matrix-product-states (MPS) - two-dimensional projected-entangled pair states (PEPS)) with limited bond dimension, or machine-learning  representations~\cite{Carleo2017}.
Alternatively, we can also have prior access to the experimental system realizing the quantum state and measure $X_a(u^{(k)})=X_e(u^{(k)})$
via Eq.~\eqref{eq:ShotNoise} based on {\it quantum data}, c.f. Fig.~\ref{fig:setup}. Note that step (i) leads to a result that can be saved classically, i.e. this step does not need to be repeated every time we want to probe a given quantum system.

As detailed in SM~\cite{SM}, we can parametrize single qubit random unitaries $u_i=R_y(\theta_i)R_z(\varphi_i)$ in terms of two rotations.
The function $X(u)$ we would like to approximate is thus a multivariate function of $2N$ variables $\theta_i,\varphi_i$, $i=1,\dots,N$. 
In order to construct $X_\mathrm{IS}(u)$ as an object that can be used for sampling, we rely on machine-learning (ML) techniques of nonlinear multivariate regression. We use existing highly optimized algorithms to fit our samples by a neural network representing our target multivariate function $X_\mathrm{IS}(u)$. For each sample $k$, the $2N$ angles $\theta_i^{(k)},\varphi_i^{(k)}$ parametrizing $u^{(k)}$ are used as inputs of the neural network, while the value of the measured function $X_a(u^{(k)})$ is the output of the network. This provides a `training' procedure, which results in a fitted neural network $X_\mathrm{IS}(u)$, which we can finally save and use for the next step of sampling unitaries (ii) of the protocol. 
Note that, when a theory representation $X_a$ is available, one could define $X_\mathrm{IS}=X_a$, i.e., avoid fitting with ML and sample directly from $X_a$. While this approach is probably the most obvious for small systems, using ML offers in the large scale scenario the possibility of converting the result of a very costly classical computation into a neural network $X_\mathrm{IS}(u)$. This neural network can be seen as a `compressed object' and can be saved and shared classically on-demand (multiple times and/or for multiple users) to realize the sampling task (ii). 

\begin{figure}[t]
\begin{minipage}[b]{0.48\linewidth}
\centering
\includegraphics[width=\textwidth]{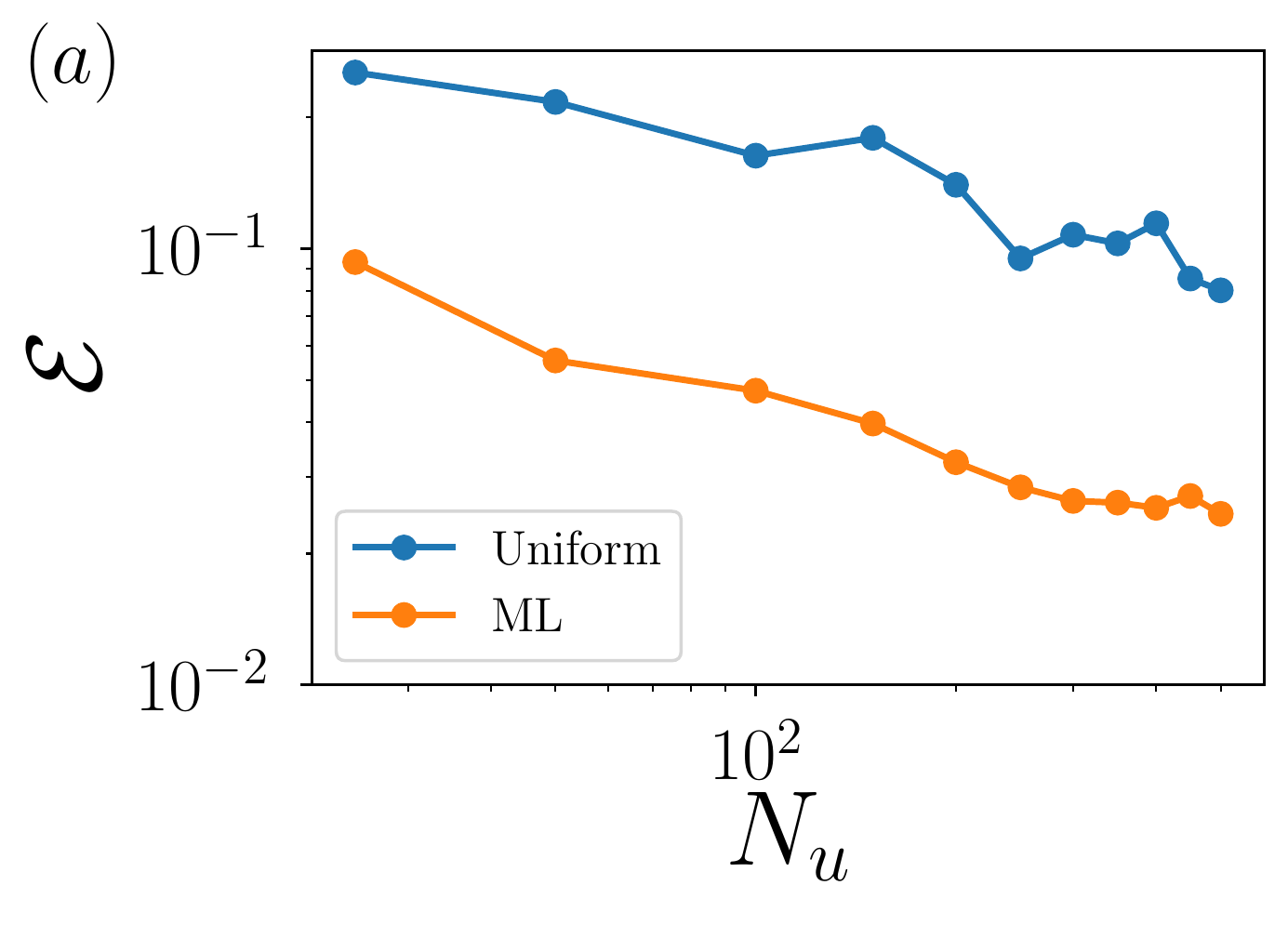}
\end{minipage}
\hskip -0.75ex
\begin{minipage}[b]{0.48\linewidth}
\centering
\includegraphics[width=\textwidth]{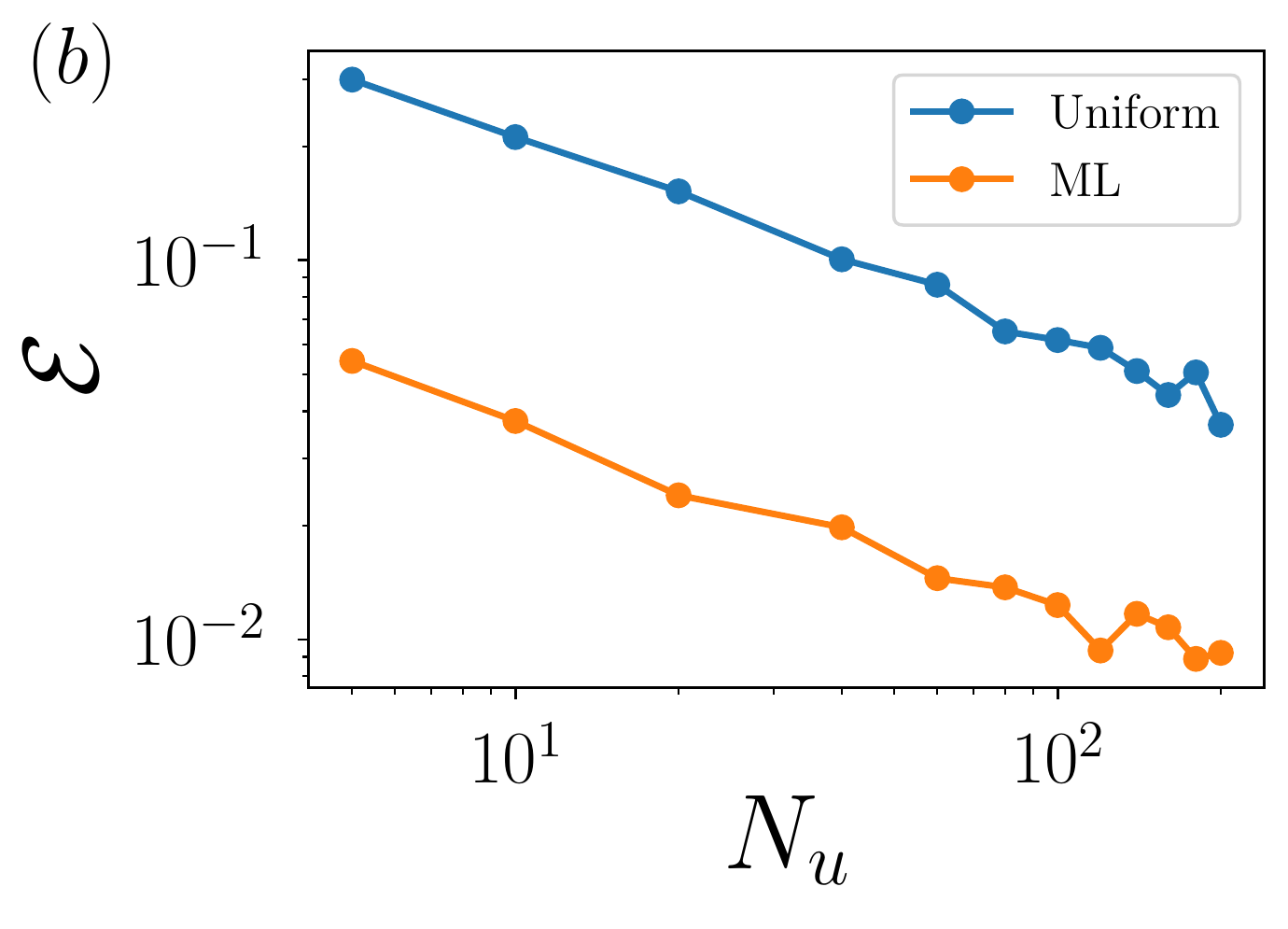}
\end{minipage}
%\vskip -1ex
%\label{fig:ML}
\begin{minipage}[b]{0.48\linewidth}
\centering
\includegraphics[width=\textwidth]{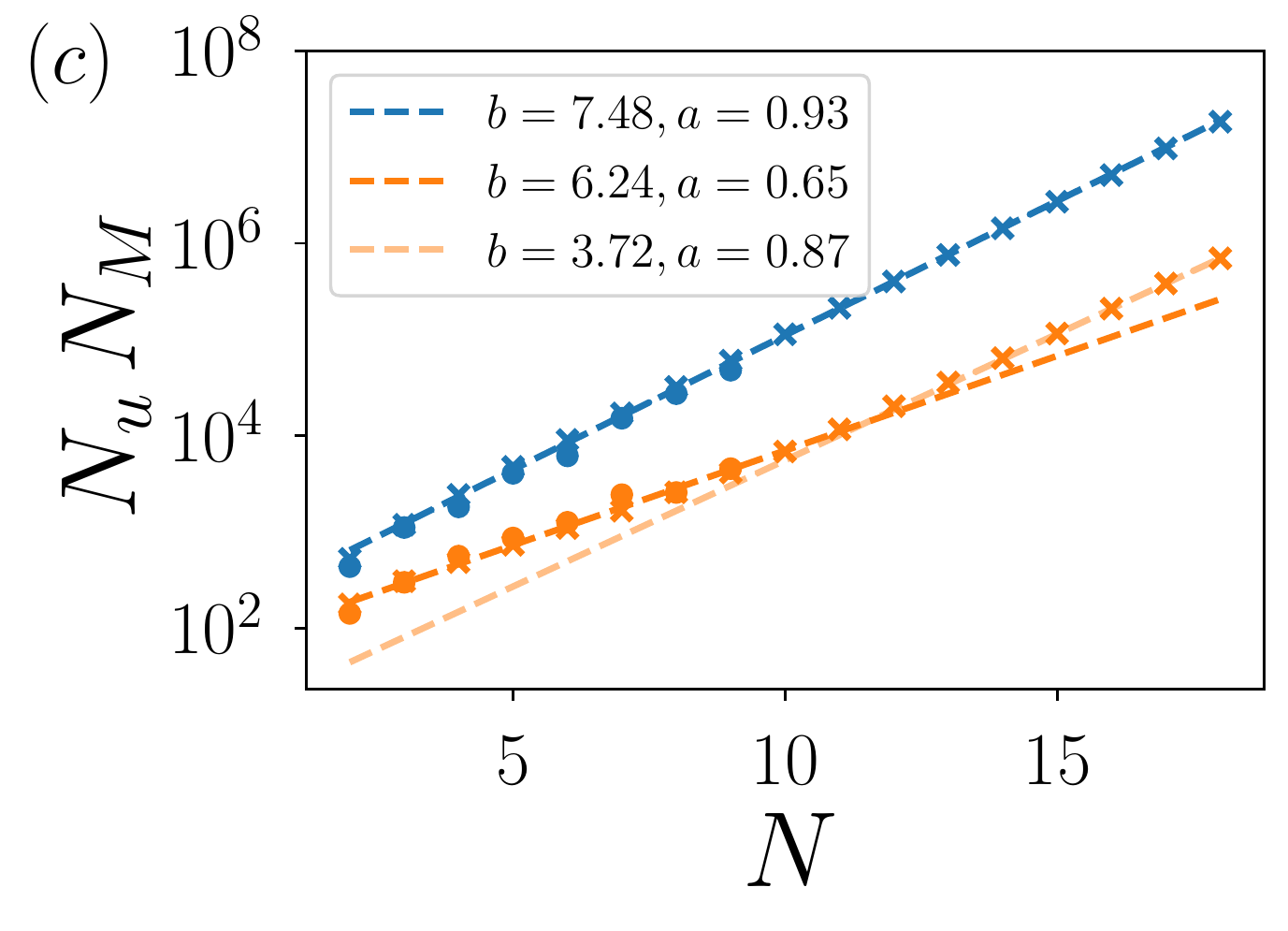}
\end{minipage}
\hskip -0.75ex
\begin{minipage}[b]{0.48\linewidth}
\centering
\includegraphics[width=\textwidth]{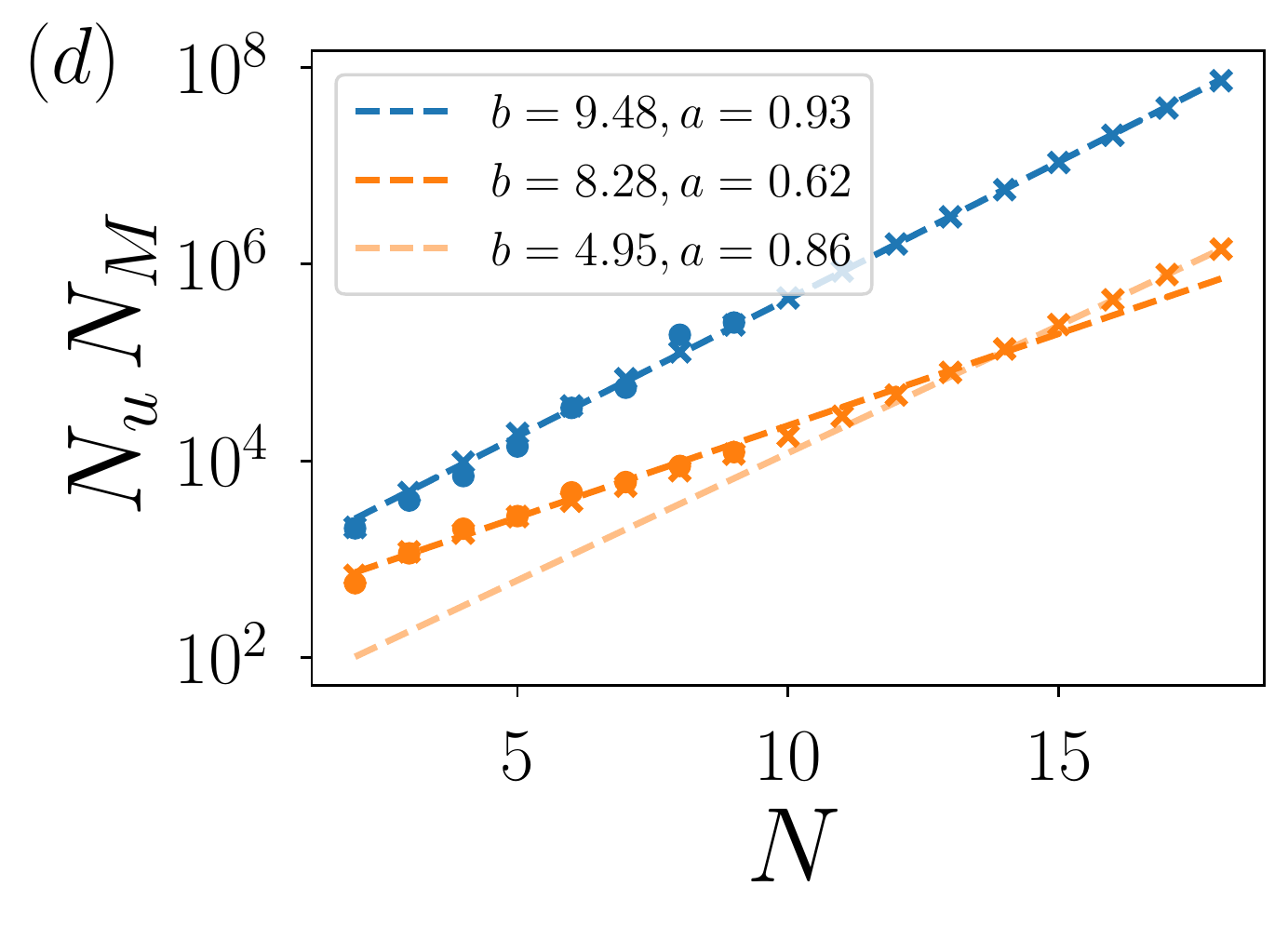}
\end{minipage}
%\vskip -1ex
\begin{minipage}[b]{0.48\linewidth}
\centering
\includegraphics[width=\textwidth]{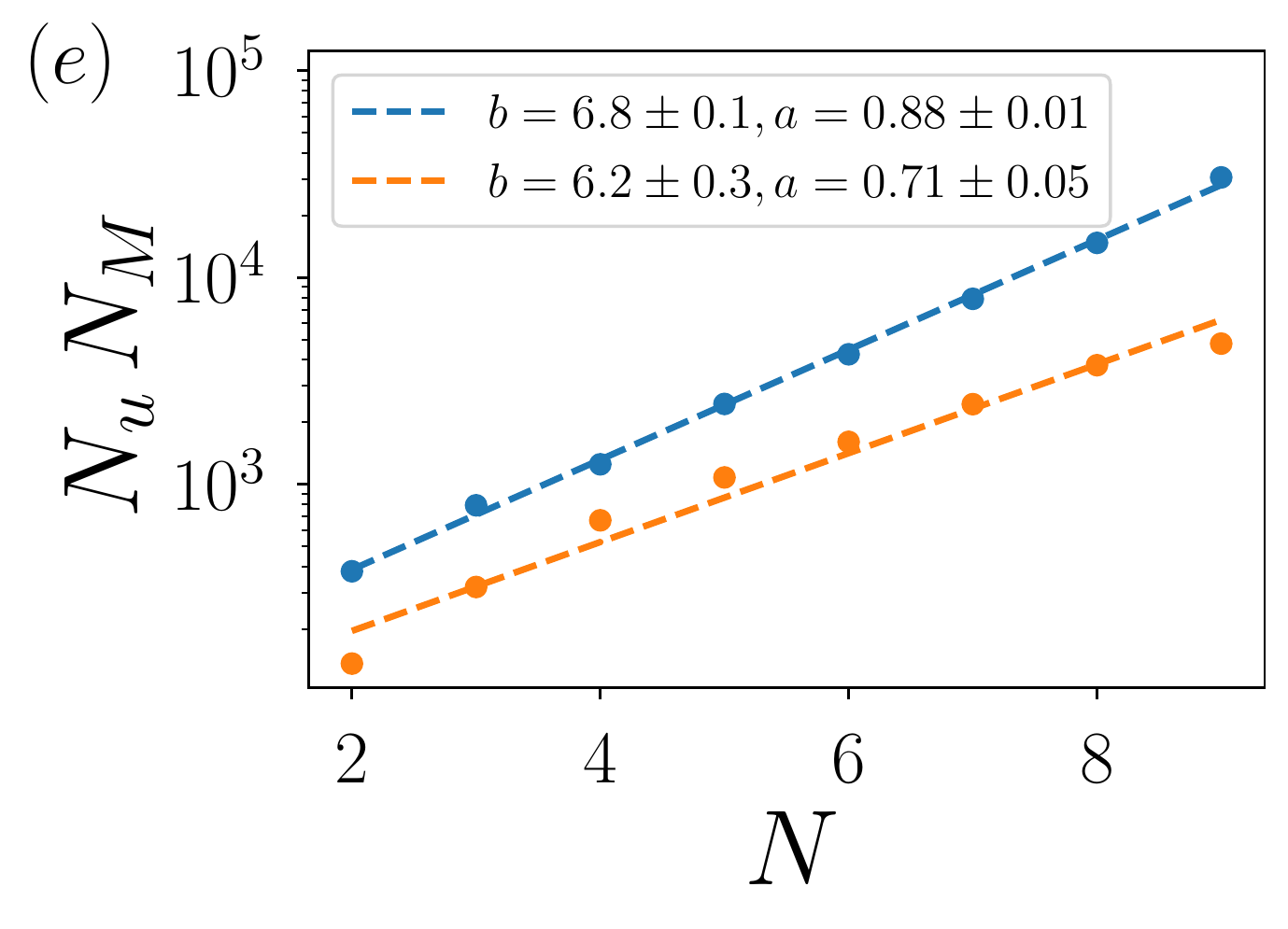}
\end{minipage}
\hskip -0.75ex
\begin{minipage}[b]{0.48\linewidth}
\centering
\includegraphics[width=\textwidth]{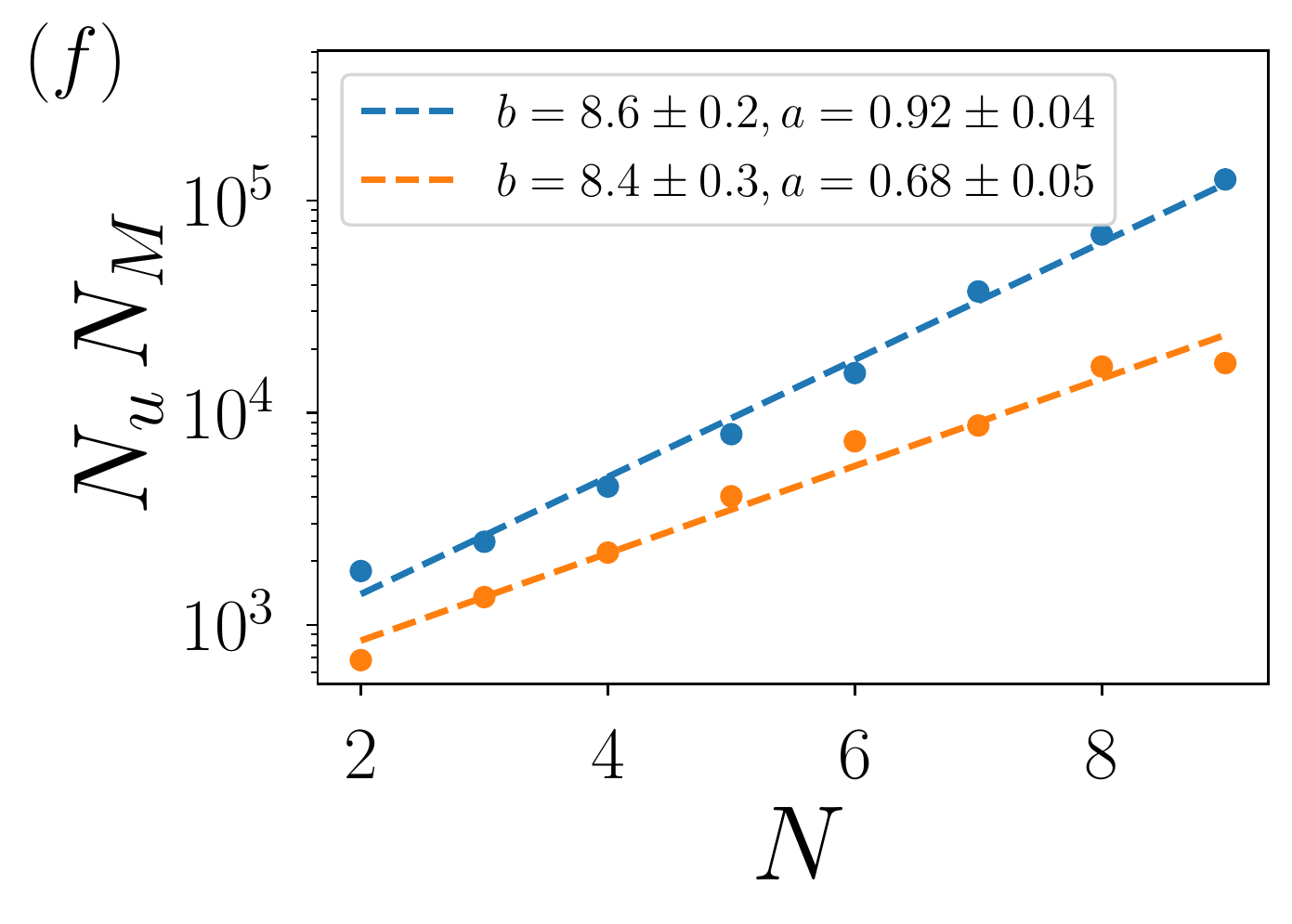}
\end{minipage}
\caption{{\it Statistical error scalings for product and GHZ states} Average statistical error $\mathcal{E}$ of the estimated purity for (a) 10-qubit product state and (b) 5-qubit GHZ state in function of $N_u$ with $N_M = 1000$ for a uniform  sampling (Uniform) and importance sampling from a machine learning model (ML).
(c-d) Scaling of the required total number of measurements $N_uN_M$ as a function of $N$ for uniform and importance sampling for a product state, to obtain a statistical error of $\mathcal{E}=0.1$ (c), and $\mathcal{E}=0.05$ (d), respectively.
%The value of $N_u/N_M$ was adjusted numerically to obtain minimized statistical errors.
We represent with cross the analytical prediction (c.f. SM) and with circles the numerical simulations. Panels (e-f) show the numerical simulations for the GHZ states for $\mathcal{E}=0.1$ and $\mathcal{E}=0.05$ with corresponding exponential fits of the type $2^{b+aN}$.% \AEc{Just to avoid searching in the text, could one add the fitted formula $2^{b+aN}$ here, too? And say which data corresponds to which sampler  in panels c-e ? }}
% use the same color convention as in Fig 2a, legend too small, numerical symbols not really visible. Add letters for panels. Add plots for GHZ+caption. I would go only to $N=15$, since we focus on numerics
\label{fig:ML}} 
\end{figure}
{\it Performance tests---}
We now benchmark our protocol.
For all states that we analyzed, product states, GHZ states, random states and other highly entangled states, we observe a drastic reduction of statistical errors with importance sampling.

We begin by considering product states $\rho=\ket{\psi}\bra{\psi}$, with $\ket{\psi}=\ket{0}^{\otimes N}$.
We consider having access to classical data with samples of randomized measurements that are not affected by shot-noise. %($N_M\to \infty$ in Eq.~\eqref{eq:RM})\AEc{there is no $N_M$ in Eq. 1}. 
The details of the training procedure are presented in the SM~\cite{SM}. For such product state, the training of a neural network $X_\mathrm{IS}(u)$ is straightforward, and we achieve a fit of $X(u)$ using three layers of neurons, with mean absolute error below five percents, see SM~\cite{SM}.
To assess the performance of importance sampling,  we will compare the average statistical error $\mathcal{E}$ in estimating the purity, with the one obtained with uniform sampling ($X_\mathrm{IS}=1$).
We compute $\mathcal{E} = \overline{|p_2 - {p_2}_{e}|}$ by numerically simulating our protocol, with $\overline{\phantom{a}}$ an average over simulated experiments.
The results are shown in Fig.~\ref{fig:ML}(a). With uniform and importance sampling, the error decays as $1/\sqrt{N_u}$, with a prefactor that is approximately $5$ times smaller for importance sampling.
We consider GHZ states \mbox{$\ket{\psi}=(\ket{0}^{\otimes N}+\ket{1}^{\otimes N})/\sqrt{2}$}
%In this case, the training of a deep neural network becomes essential to `learn' the non-trivial correlations present in the measurement function $X(u)$, see SM.
in Fig.~\ref{fig:ML}(b). Here, importance sampling provides a significant advantage over uniform sampling, meaning that the neural network succeeded in learning how to sample correlated random unitaries that are adapted to probe a GHZ state.

We can also extract from numerical simulations the total number of measurements $N_u N_M$, minimized over possible choices of $N_u,N_M$, that is required to achieve a statistical error $\mathcal{E}$. Here, to ensure that we extract scaling relations that are independent of the choice of the neural network ansatz, with importance sampling, we sample directly from the ideal theory state $X_\mathrm{IS}(u)=X(u)$.
 In this case, for a fixed number of measurements, the statistical error is minimized for $N_u=1$. We present in the SM~\cite{SM} additional numerical simulations, using optimized neural networks for $N_u=200,500$ which support the same conclusions.
For the product state, we observe in Fig.~\ref{fig:ML}(c-d) that the required $N_uN_M$ grows as $ 2^{b+aN}$ (see also Ref.~\cite{Brydges2019}) with $a\approx 0.93$ for uniform sampling, and $a\approx 0.65$ for importance sampling.
Our numerical results for the GHZ states [panels (e)-(f)] show similar results, with favorable scaling exponents for importance samplings, in particular at high accuracy $\mathcal{E}=0.05$ [panel (f)].
As the exponent $a$ is reduced compared to uniform sampling, 
we see that importance sampling offers an exponential reduction of the required number of measurements.  In addition, in all panels (c-f), the prefactor $2^b$ obtained for importance sampling is smaller than the one for uniform sampling.

For pure product states, we can compare our numerical results with analytical calculations, which are presented in the SM~\cite{SM}, and extend them to the large $N$ limit.
Our analytical study shows the existence of two regimes: For $N \lesssim N_c $, smaller than a certain value $N_c\propto \log(1/\mathcal{E})$, we find a strongly favorable scaling exponent of $a=0.37$ for importance sampling.
For large $N \gtrsim N_c$, the exponent increases towards  $a\approx 0.88$ which is however still smaller than in the case of uniform sampling, $a\approx 0.92$.
In particular, we note that the favorable scaling regime, $N<N_c\propto \log(1/\mathcal{E})$, grows with the inverse error threshold $\mathcal{E}$, in agreement with the results shown in Fig.~\ref{fig:ML}(c-d).
The advantage of importance sampling is thus most pronounced at high accuracy (small $\mathcal{E}$), enabling estimation of the purities with exponentially less measurements compared to uniform sapling.

\begin{figure}[t]
\begin{minipage}[b]{0.5\linewidth}
\centering
\includegraphics[width=\textwidth]{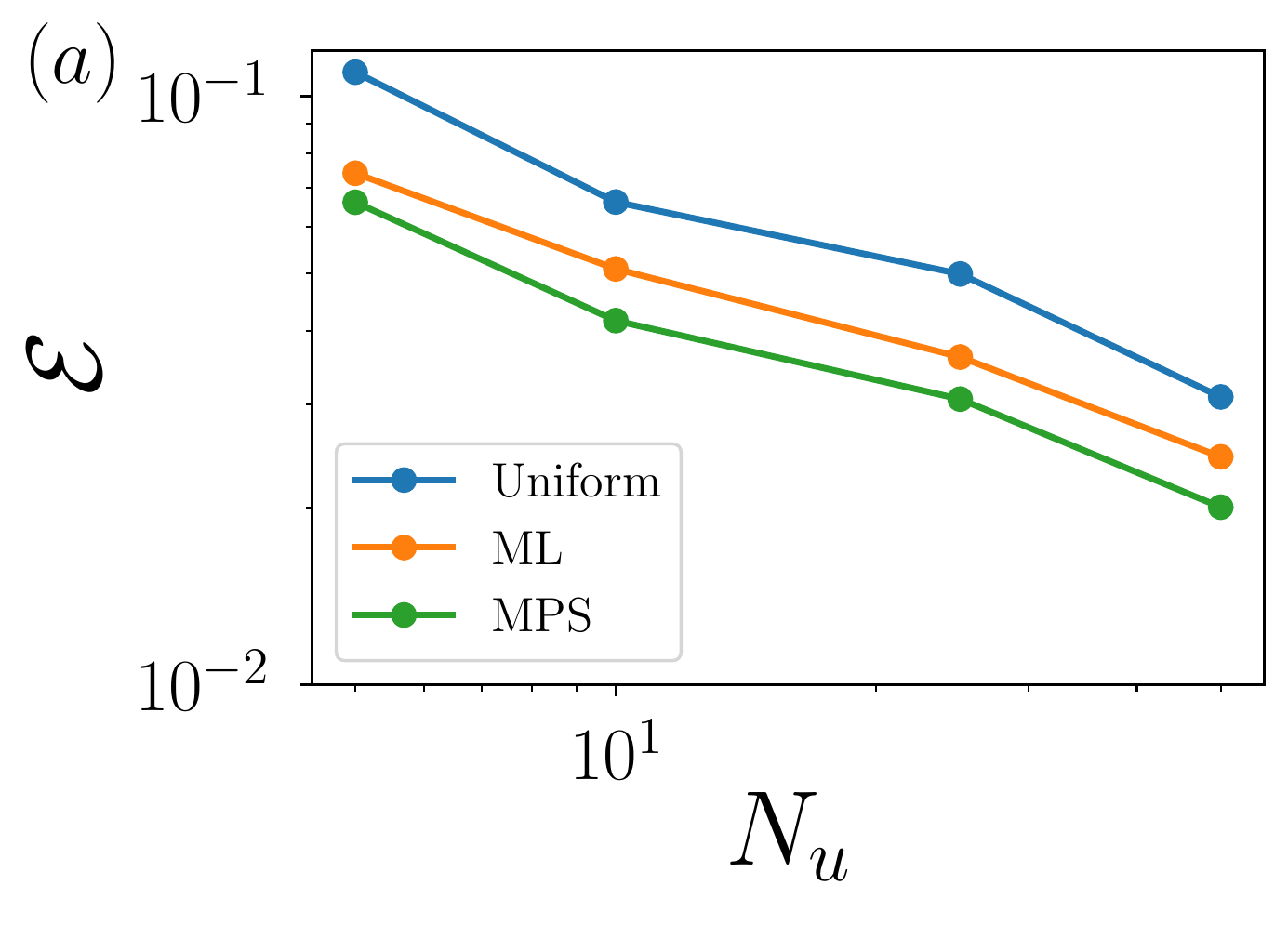}
%\caption{default}
%\label{fig:figure1}
\end{minipage}
\hskip -1ex
\begin{minipage}[b]{0.5\linewidth}
\centering
\includegraphics[width=\textwidth]{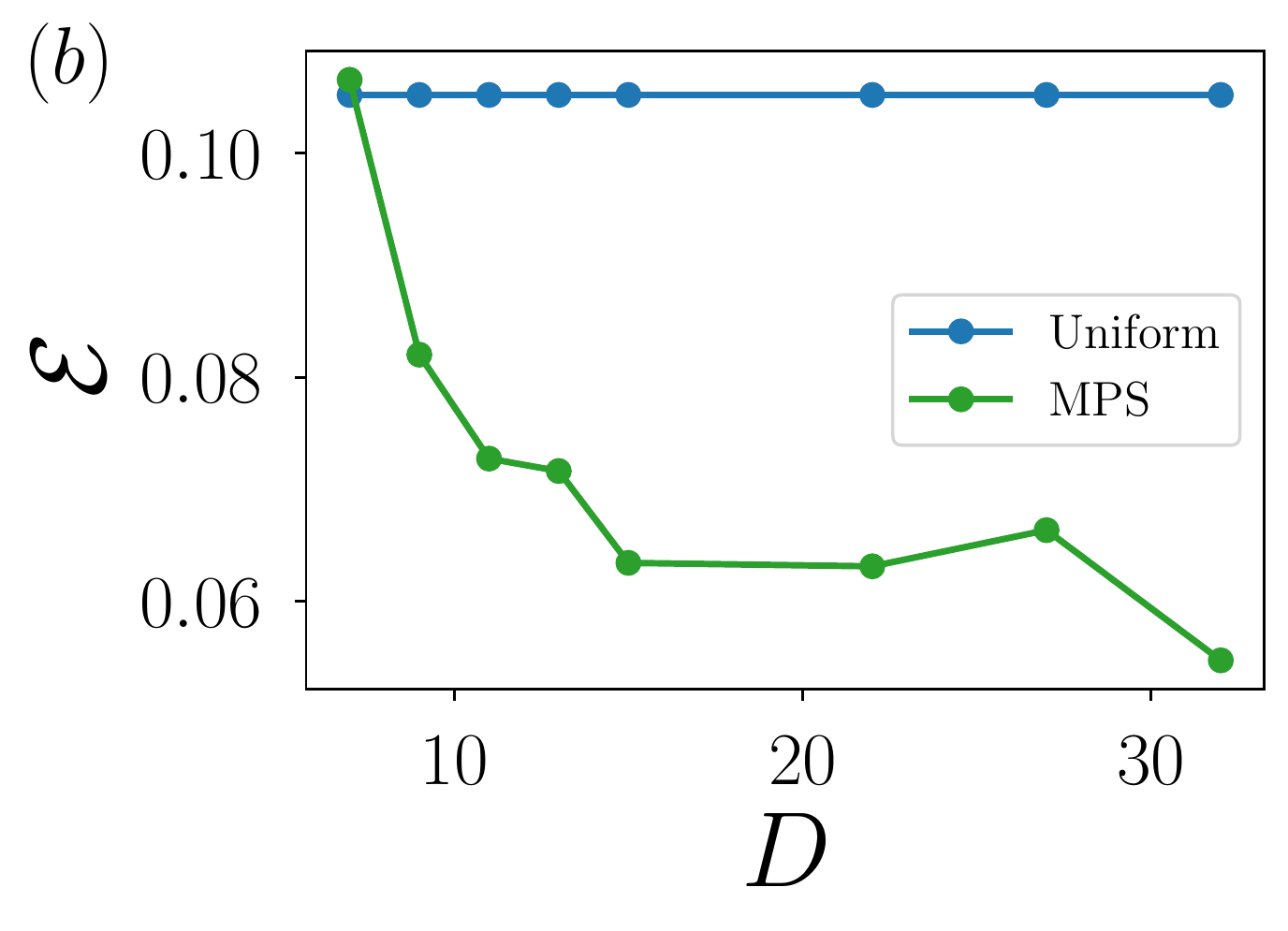}
%\caption{default}
%\label{fig:figure2}
\end{minipage}
\caption{{\it Purity estimation of a highly entangled 10 qubit state with ML and MPS samplers.} Panel (a) shows the average statistical error $\mathcal{E}$ of the estimated purity in function of $N_u$ with $N_M = 7500$ for a uniform sampling and importance sampling done from a neural network and a MPS representation of the corresponding state respectively. Panel (b) illustrates the scaling of the error $\mathcal{E}$ w.r.t different bond dimensions $D$ used for the MPS representation of the state for $N_u = 5$ and $N_M = 7500$.\label{fig:ML_MPS}}
\end{figure}
%\begin{figure}[t]
   % \includegraphics[width = 0.85\textwidth]{figures/fig_1.jpg}
    %\caption{}
    %\label{fig:ML_MPS}
%\end{figure}
We have demonstrated that importance sampling provides an exponential reduction of the measurement budget for two specific states, product and GHZ states, which are `well-conditioned states'
, and whose fidelity can be efficiently estimated via direct fidelity estimation~\cite{Flammia2011,DaSilva2011}.
However, importance sampling is not useful only for these states. First, we show in the SM~\cite{SM} a scaling analysis for pure random states that show a significant gain in using importance sampling compared to uniform sampling, which is here however constant with $N$. %Here, we} 
Second, we can also use our protocol to probe  mixed, and highly entangled states, which are created via a quantum quench~\cite{Brydges2019}. Here, we consider a state modelling a  trapped-ion $10-$qubit experiment described in Ref.~\cite{Brydges2019}, which corresponds to the dynamics of a long-range $XY$ Hamiltonian~\footnote{The precise master equation is given in Ref.~\cite{Brydges2019}, the propagation time is $t=5$ ms}. This highly entangled state is characterized by a purity of $p_2\approx 0.62$, and a half-system purity of $p_2'\approx 0.16$, in agreement with the experimentally measured values~\cite{Brydges2019}. In order to mimic a situation when the decoherence parameters are unknown, we train our neural network on an ideal pure state, i.e., modelling the system without errors, and use it to estimate the purity of the mixed state $\rho$. The results are shown in Fig.~\ref{fig:ML_MPS}. While we see a clear improvement w.r.t uniform sampling, here importance sampling does not achieve the level of performance seen for GHZ states and product states. This is due to an imperfect training of the used convolutional neural network (CNN).
While the training can obviously be improved by changing the structure of the neural network, we propose now a `physics-motivated' complementary approach based on tensor networks, and which offers for this particular state an improvement over ML and provides a simplified approach to build $X_\mathrm{IS}$.

{\it Importance sampling from Matrix-Product-States--}
We illustrate how we can use approximate theory representation for importance sampling. Here, we consider 
Matrix-Product-States (MPS), which have been introduced to solve numerically condensed-matter problems~\cite{Schollwock2011}.
%, these objects can  be also used for machine learning~\cite{Stoudenmire2016,Glasser2018,Han2018}.
With $N$ qubits, MPS are wavefunctions of the form
\begin{equation}
    \ket{\psi_D}= \sum_{\substack{s_1,\dots,s_N\\ \ell_1,\dots, \ell_{N-1}}} [A_1]^{(\ell_1)}_{s_1}[A_2]^{(\ell_1,\ell_2)}_{s_2}\dots [A_N]^{(\ell_{N-1})}_{s_N} \ket{\mathbf{s}},
\end{equation}
with $\ket{\mathbf{s}}=\ket{s_1}\otimes \dots \ket{s_N}$, and 
where each `bond' index $\ell_i$ can take at most $D$ different values. A schematic of the sequence of $2,3$ leg tensor $A_i$ representing the MPS is shown in Fig.~\ref{fig:setup}. The bond dimension $D$ is the key parameter of a MPS, setting the maximum entanglement entropy $\propto \log(D)$ that can be captured by such state~\cite{Schollwock2011}. MPS are in particular relevant for approximating  ground states of a many-body Hamiltonians~\cite{Eisert2010,Schollwock2011}. The MPS framework thus appears as a `physically-inspired' approach to build an importance sampling function  $X_\mathrm{IS}(u)$, which complements the ML approach  (c.f., Fig.~\ref{fig:setup}.). The training of $X_\mathrm{IS}(u)$ here is straightforward: (i.1) Form via a MPS algorithm  an approximation $\ket{\psi_D}\bra{\psi_D}$ of the quantum state $\rho$~\cite{Schollwock2011}. (i.2) Build the function $X_\mathrm{IS}(u)$ with Eq.~\eqref{eq:RM}, by realizing projective measurements $(s_u^k)_{D}$ on the MPS. While this step can be realized efficiently~\cite{Han2018}, here we simply use the probabilities $[P_u(s)]_D$ to build $X_\mathrm{IS}(u)$.

As shown in Fig.~\ref{fig:ML_MPS}a), importance sampling with a MPS with $D=15$ already provides a reduction of statistical errors compared to our best neural network model, while the fidelity $\bra{\psi_D}\rho\ket{\psi_D}=0.7$ shows that this MPS is indeed only an approximation of $\rho$.
Here, $\ket{\psi_D}$ was formed by an algorithm that approximates the dynamics of a system with long-range interactions~\cite{Zaletel2015}, see also Ref.~\cite{Brydges2019}. When using MPS importance sampling, an interesting trade-off appears in terms of required classical versus quantum hardware to measure entanglement: MPS with increasing bond dimensions require more  classical resources, but are more performant for importance sampling. This is shown in  Fig.~\ref{fig:ML_MPS}b), where the statistical error is represented as a function of $D$. {As shown in the SM~\cite{SM}, we can draw the same conclusions when considering subsystems of $5$ and $10$ qubits being part of a $10$ and $20$ qubit system, respectively.}

{\it Conclusion---}
Importance sampling boosts the power of randomized measurements protocols, allowing for measuring more efficiently purities and second R\'enyi entropies. Our approach is immediately applicable in all randomized measurement protocols, e.g.  to measure scrambling~\cite{Vermersch2019}, topological invariants~\cite{Elben2019,Cian2020} , and fidelities~\cite{Flammia2011,DaSilva2011,Elben2020a}.

We have studied how the investment of classical resources for building an importance sampling function `pays off' in terms of statistical errors.
We believe that further studies extending our scaling analysis can help us to answer this conceptual question, but also to again push the limits of randomized measurements.

Finally, as an extension of our protocol, it would be interesting to consider an  adaptive measurement scheme, where the distribution $p_\mathrm{IS}$ is iteratively adapted based on prior measurements.

\begin{acknowledgments}
We thank A. Minguzzi, C. Branciard, M. Dalmonte for fruitful discussions, and comments on the manuscript.
AR is supported by Laboratoire d'excellence LANEF in Grenoble (ANR-10-LABX-51-01) and from the Grenoble Nanoscience Foundation.
BV acknowledges funding from the Austrian Science Fundation (FWF, P 32597 N), and the French National Research Agency (ANR-20-CE47-0005, JCJC project QRand). Work at Innsbruck is supported by the European Union program Horizon 2020 under Grants Agreement No.~817482 (PASQuanS) and No.~731473 (QuantERA via QTFLAG), the US Air Force Office of Scientific Research (AFOSR) via IOE Grant No.~FA9550-19-1-7044 LASCEM, by the Simons Collaboration on Ultra-Quantum Matter, which is a grant from the Simons Foundation (651440, PZ), and by the Institut f\"ur Quanteninformation. A.E.\ acknowledges funding by the German National Academy of Sciences Leopoldina under the grant number LPDS 2021-02. 
We used ML routines of TensorFlow-Keras, ITensor MPS algorithms~\cite{Fishman}, and the quantum toolbox QuTiP~\cite{Johansson2013}.
\end{acknowledgments}

%\bibliographystyle{apsrev4-1}
%\bibliography{OptimizationRM}


\clearpage

\maketitle
\section{Appendix A: Parametrizing local random unitaries}
In this section we discuss how the local unitaries can be parametrized in terms of two angles. The local random unitaries used in this protocol are distributed by the Haar measure and belong to the CUE. A single qubit random rotation $u_i$ with $i = 1,...,N$ can be defined as:
\begin{equation}
 u_i =\begin{bmatrix} 
\cos{\phi_i} \: e^{i\alpha_i} & \sin{\phi_i}\:e^{i\psi_i} \\
 -\sin{\phi_i} \: e^{-i\psi_i}& \cos{\phi_i}\: e^{-i\alpha_i} \\
\end{bmatrix}\label{eq:huw}
\quad
\end{equation}
where $\phi_i \in [0\; , \pi/2]$ ; $\alpha_i \;\&\; \psi_i \in [0\;,2\pi]$ with the Haar measure given as follows \cite{Diaconis2015}
\begin{equation}
    du_i = 2\cos{\phi_i}\sin{\phi_i} \:d\phi_i \:d\alpha_i\: d\psi_i
\end{equation}
The measure can be rewritten by defining $\sin^2{\phi_i} = \xi_i$ and leads to
\begin{equation}
    du_i = d\xi_i \:d\alpha_i\: d\psi_i
\end{equation}
where $\xi_i \in [0\; , 1]$. The same local random unitary $u_i$ can be experimentally realized by combining random rotations along $Y$ and $Z$ axes of the Bloch sphere and one equally writes \cite{Brydges2019}
\begin{equation}
    u_i = R_z(\gamma_i)R_y(\theta_i)R_z(\varphi_i)\label{eq:rot}
\end{equation}
where $R_{\beta}(\theta) = e^{-i\sigma^{\beta}\theta/2}$, $\sigma^{\beta}$ with $\beta = {y,\,z}$ are the Pauli matrices and $\theta \in [0\; , 2\pi]$ is the random rotation angle. Equating the matrix elements of Eq.~\eqref{eq:huw} and Eq.~\eqref{eq:rot} gives the relation between the parametrized unitary angles in function of the rotation angles and its corresponding distribution measures
\begin{align}
\begin{cases}  \xi_i = \sin^2{\theta_i/2}\\ \psi_i = (\varphi_i - \gamma_i)/2 \\ \alpha_i = -(\varphi_i + \gamma_i)/2 
\end{cases}
\hskip -1ex
\implies
\begin{cases}  d\xi_i = \sin({\frac{\theta_i}{2}})\cos({\frac{\theta_i}{2}})\,d\theta_i\\ d\psi_i = (d\varphi_i - d\gamma_i)/2 \\ d\alpha_i = -(d\varphi_i + d\gamma_i)/2 \label{eq:measure}
\end{cases}
\end{align}

As the measurement of each qubit is finally performed in the computational basis, the last $R_z$ rotation of $u_i$ can be dropped by taking $\gamma_i = 0$. From the Eq.~\eqref{eq:measure}, sampling $\varphi_i$ uniformly in $[0,\, 2\pi]$ leads to $\alpha_i$ and $\psi_i$ being distributed uniformly and conversely sampling $\xi_i$ uniformly in $[0,\, 1]$ leads to a uniform sampling of $\theta_i$. To realize unitaries sampled from the Haar measure, it is sufficient to randomly sample: $\xi_i$ relating to the $Y$ rotation $R_y(\theta_i)$ and $\varphi_i$ connecting to the local unitary angles $\psi_i$ and $\alpha_i$. Thus finally, each single qubit random unitary $u_i$ is parametrized by $R_y(\theta_i)R_z(\varphi_i)$ to sample from the Haar measure.
\section{Appendix B: Bounds on $X(u)$}
In this section, we derive the bounds
\begin{equation}
    \frac{1}{2^N}\le X(u) \le 2^N.
\end{equation}
Our starting point consists in rewriting the function $X(u)$ as
\begin{eqnarray}
X(u) &=& \frac{1}{2^N}
\mathrm{Tr}
\left(
\bigotimes_{i=1}^N
\left[
1\otimes 1+3\sigma^z_i \otimes \sigma^z_i
\right]
(u\rho u^\dag \otimes u \rho u^\dag)
\right)
\nonumber \\ 
&=&\frac{1}{2^N}
\sum_A
3^{|A|} 
\langle \sigma^z_A \rangle^2, 
\end{eqnarray}
with $\langle \sigma^z_A \rangle=\mathrm{Tr}(\sigma^z_A u \rho u^\dag)$. The summation involves all qubit partitions $A=(i_1,\dots, i_{|A|})$, and we defined the Pauli string $\sigma^z_A=\bigotimes_{i\in A} \sigma^z_i$. The first line of the above equation can be proven by expanding the trace operation in the computational basis $\ket{s}\otimes \ket{s'}$.

We can now infer bounds on $X(u)$, using the relation $0\le \langle \sigma^z_A \rangle^2\le 1$. First, we get a lower bound as
\begin{equation}
X(u) = \frac{1}{2^N}
\left(1+
\sum_{A\neq\emptyset}
3^{|A|} 
\langle \sigma^z_A \rangle^2
\right) 
\ge \frac{1}{2^N}.
\end{equation}
The upper bound is obtained using 
\begin{equation}
\sum_A
3^{|A|} = \sum_{k=0}^N \binom{N}{k} 3^k = 4^N, 
\end{equation}
which leads to 
\begin{equation}
X(u) \le  \frac{1}{2^N}
\sum_A
3^{|A|} = 2^N.
\end{equation}

Note that the lower bound is saturated by an identity density matrix $\rho=I/2^N$, while the upper bound is saturated by a product state $\rho = \ket{\psi_u}\bra{\psi_u}$, with $\ket{\psi_u}=u^\dag \ket{0\otimes \dots \otimes 0}$.

\section{Appendix C: Machine Learning}
In this section we elaborate the details of the training procedure followed to obtain trained neural network importance sampler models $X_\mathrm{IS}(u)$ used in step (i) of our protocol. Training of neural network (NN) models consists in fine tuning different parameters in the available optimized Machine Learning (ML) algorithms to make accurate predictions using multivariate regression analysis. We elaborate here on the parameters that were adjusted to achieve well trained models for the different quantum states that we tested. To reduce the variability during training with all the different available parameters, we choose to fix some parameters beforehand for all the training. The regression loss function was taken to be the Mean Absolute Error (MAE) with the activation function on each neuron fixed to `Relu'. The `Epoch' which represents the number of times the entire dataset passes through the NN architecture was taken to be 500 with a fixed learning rate of 0.001.

The important parameters that were adjusted during the course of the training for various states of different system sizes $N$ can be summarized by (i) the number of hidden layers $N_\mathrm{layers}$, (ii) number of neurons in each layer $N_\mathrm{neurons}$ and (iii) the number of training samples $N_\mathrm{samples}$. Each training sample consisted of $2N$ inputs ($\xi_i$ and $\varphi_i$ for each qubit) and one output (corresponding $X$ function defined by the state of interest $\rho$). During the training phase, the $N_\mathrm{samples}$ were split into 2 separate sets: a training set on which the model learned the features of the target function and the test set, to characterize the ability of the trained model to generalize. The goal of the training was to obtain an efficient compressed model capable of fitting a function $X_\mathrm{IS}(u)$ optimally to the target function $X(u)$ for any given set of $2N$ angles of $u$. 

The product states and the GHZ states were trained by Deep Neural Network (DNN) models using $N_\mathrm{layers}$ = 3 with decreasing $N_\mathrm{neurons}$ per layers (complexity of the NN decreasing from the input layer to the output layer) that varied for the different states. The last highly entangled quantum simulation state was trained using a Convolutional Neural Network (CNN) highlighting some quasi-translational invariant features with a $N_\mathrm{layers} = 2$ configuration. The right combination of $N_\mathrm{neurons}$ for each hidden layer was fine tuned to achieve the lowest possible fitting error $\mathcal{E}$ (MAE) associated with the target function $X(u)$ and the obtained fit for the corresponding output $X_\mathrm{IS}(u)$. This step is crucial as there exists a right combination of $N_\mathrm{neurons}$ that provides the adequate NN model without overfitting the data (model losing the ability to generalize to samples outside the training set).
\begin{figure}[t]
\begin{minipage}[b]{0.5\linewidth}
\centering
\includegraphics[width=\textwidth]{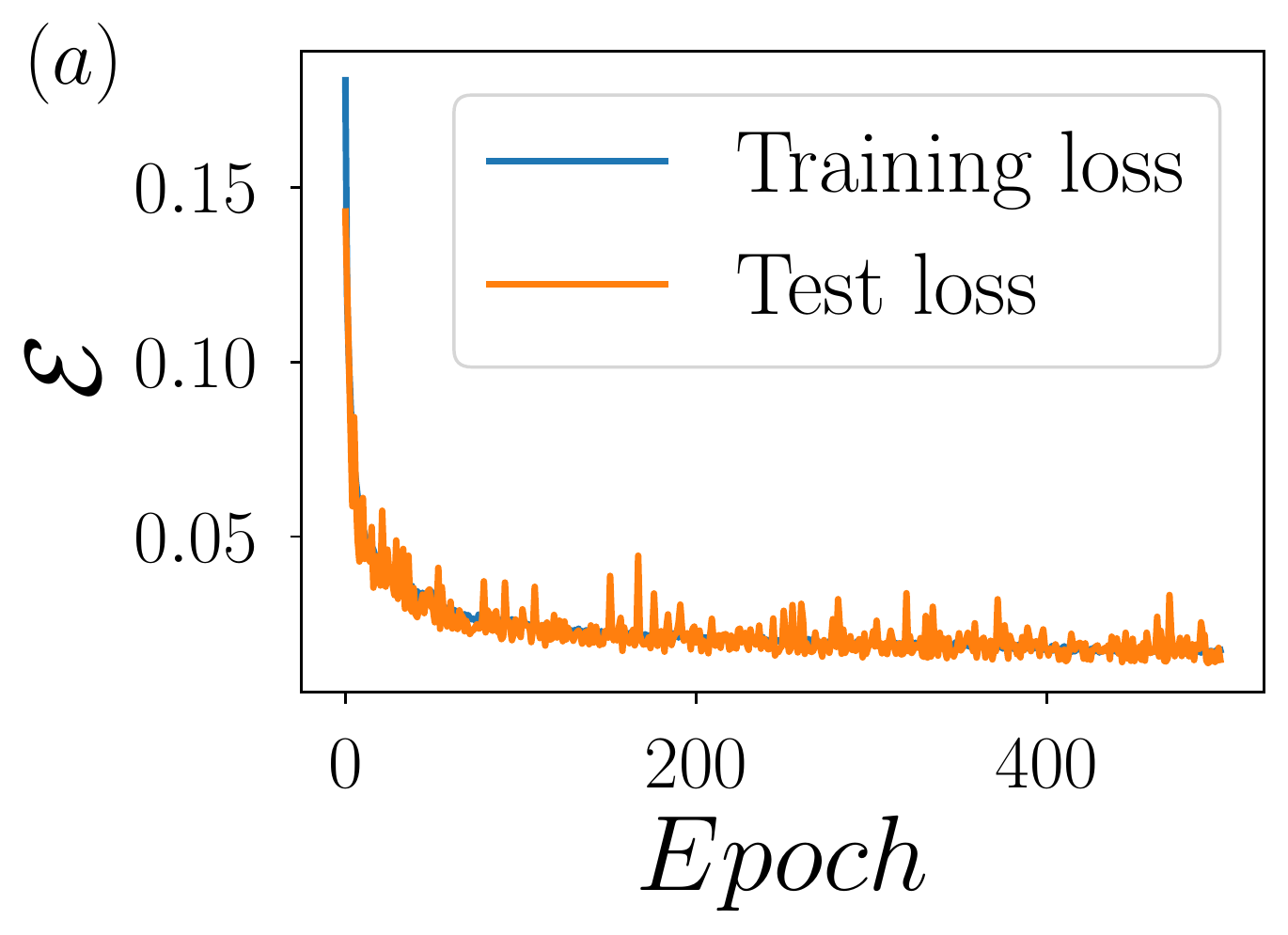}
\end{minipage}
\hskip -1ex
\begin{minipage}[b]{0.5\linewidth}
\centering
\includegraphics[width=\textwidth]{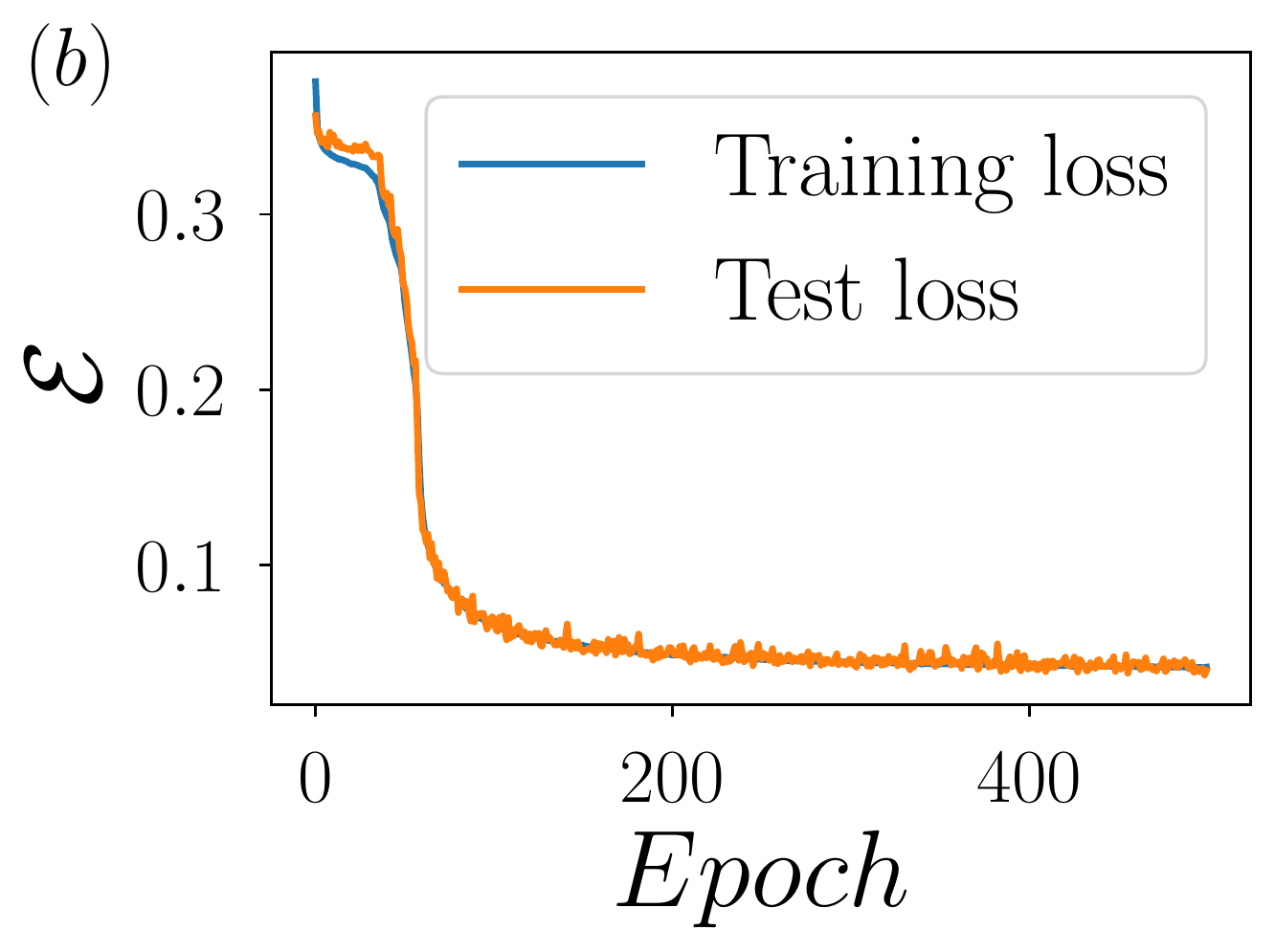}
\end{minipage}
\caption{{\it Training history for the DNN models.} An illustration of the fitting error ($\mathcal{E}$) vs Epoch for two different target functions, panel (a) for a 10-qubit product state and panel (b) for 5-qubit GHZ state. The blue curve indicates the fitting error for the training set while the orange curve shows for the test set. A superposition of the two curves highlights a good training without overfitting.\label{fig:val_loss}}
\end{figure}

For the product states, we targeted and achieved $ \mathcal{E} \leq 5\%$ for all the system sizes that we considered (an example in Fig.~\ref{fig:val_loss}a) by taking $N_\mathrm{samples} = 100000$ with $N_1 = 200$, $N_2 = 100$ and $N_3 = 10$ with $N_i$ being the $N_\mathrm{neurons}$ for each of the hidden layers of the DNN. On the other hand for GHZ states, the training cost of the model increased with the system size in order to learn the inherent non-trivial correlations of the target function. For example in Fig.~\ref{fig:val_loss}b, the 5 qubit GHZ state was trained to achieve $ \mathcal{E} \backsim 5\%$ by taking $N_\mathrm{samples} = 200000$ with $N_1 = 256$, $N_2 = 64$ and $N_3 = 16$. The training of the quantum simulation state was done on the pure MPS representation $\ket{\psi_D}$ of the same with a bond dimension $D = 32$. It was trained with $N_\mathrm{samples} = 500000$ and consisted a 2D-convolutional layer of $N_1 = 250$ with a kernel size of 2 by 2, followed by 2 dense hidden layers with $N_2 = 50$ and $N_3 = 5$ respectively. The training resulted in a fitting error $ \mathcal{E} \sim 10\%$.

It is important to note that the error $\mathcal{E}$ associated to the training does not lead to systematic errors in the estimations of the concerned quantities ($p_2$ and $S_2$) in the estimation phase of the protocol. We only take advantage here of the trained model to provide us with an approximate importance sampler in the form of $X_\mathrm{IS}(u)$. In general, all the training parameters and the overall NN architecture can be further improved to obtain efficient training and higher accurate fits to the target function. The classical cost of the training varies with an increase of the input parameters (in function of the number of qubits) and the state dependent target function $X(u)$.
\section{Appendix D: Metropolis sampling}
Once we obtain a faithful importance sampler $X_\mathrm{IS}(u)$, we perform an importance sampling of the unitaries $u$ using the Metropolis Algorithm (MA)~\cite{NRecipes2007}. The algorithm has the target distribution function $X_\mathrm{IS}(u)$ with a uniform proposal distribution to pick each candidate unitary $u^{(r)}$ defined by its $2N$ angles. For our algorithm, we had an acceptance rate $\alpha \sim 53\%$ for product states and $\alpha \sim 65\%$ for GHZ states with a burn-in period of typically 50 samples which could be adjusted depending on the concerned system size $N$. Through this acceptance-rejection method, the algorithm collects a total of $N_s$ samples of unitaries $u$ in which we have $N_u$ distinct ones. The final estimation of $p_2$ is made by the importance sampling Monte Carlo integration where we also take into account the number of occurrence $n^{(r)}$ of each unitary $u^{(r)}$ in our $N_u$ distinct samples. This expression is given as
\begin{equation}
    [p_2]_\mathrm{IS} = \frac{ 1}{N_s}
    \sum_{r = 1}^{N_u} \frac{n^{(r)} X_e(u^{(r)}) }{p_\mathrm{IS}(u^{(r)})}, 
\end{equation}
with $N_s=\sum_r n^{(r)}$.
In general, there is scope to further improve the Metropolis algorithm by implementing a more advanced Hamiltonian Monte Carlo (HMC) Metropolis algorithm.
\section{Appendix E: Analytics of importance sampling}
Statistical errors in our protocols for purity estimation arise from a finite number $N_u$ of local random unitaries $u=\bigotimes_{i=1}^N u_i$ and a finite number $N_M$ of single-shot measurements per unitary.  In this appendix, we derive analytical expressions for the variances of the  estimator $\hat{X}(u)$ for finite $N_u$ and $N_M$. 
We consider two limiting cases, (i) a uniform sampler $X_{\text{IS}}(u)=1$ where the single qubit unitaries $u_i$ a sampled uniformly and independently from the Haar measure $\text{d}u_i$ ($i=1,\dots, N$) and (ii) a perfect sampler where the importance sampling distribution $X_{\text{IS}}(u)=X(u)$ is given by the function $X(u)$ whose integral $p_2=\int \du X(u)$  we aim to estimate (here $\text{d}u=\prod_i\text{d}u_i $). 
We note that in the latter case (ii) statistical errors still arise from shot-noise due to a finite number $N_M$ of single-shot measurements.

We first recall the data taking procedure and estimators. We assume that an approximation $X_{\text{IS}}(u)$ of the function $X(u)$  whose integral $p_2=\int \du X(u)$ we aim to estimate has been obtained. In the case without prior knowledge, we simply take $X_{\text{IS}}(u)\equiv 1$ for all $u$ (uniform sampler). This defines the importance sampling distribution $p_{\text{IS}}(u)=X_{\text{IS}}(u)/ \int \du X_{\text{IS}}(u)$. From this distribution, we sample independently $N_u$ local random unitaries $u^{(r)}$ ($r=1,\dots,N_u$), apply them to the quantum state of interest $\rho$,  and perform $N_M$ computational basis measurements per unitary. An  unbiased estimator $\hat{X}$ of the purity $ p_2=\trAE{\rho^2}$ is then constructed in two steps: 

First, using the $N_M$ observed bitstrings $s^{(r)}_m$  ($m=1,\dots,N_M$) after the application of unitary $u^{(r)}$,  we define 
\begin{align}
\hat{X}(u^{(r)}) &=\binom{N_M}{2}^{-1} \sum_{\substack{m,m'=1 \\ m> m'}}^{N_M} (-2) ^{-D[s_m^{(r)},s^{(r)}_{m'}]}  \\ 
&=\binom{N_M}{2}^{-1}   \sum_{\substack{m,m'=1 \\ m> m'}}^{N_M}  \trAE{ A \; \hat{s}_m^{(r)}\otimes \hat{s}_{m'}^{(r)}}\;
\end{align}
with $\hat{s}_m^{(r)}=\ketbra{s_m^{(r)}}{s_m^{(r)}}$ denoting the projector to the computational basis state  corresponding to the bitstring $s_m^{(r)}$ and the $2$-copy operator $A$ is defined  as
\begin{align}
    A = 2^N \sum_{s,s'} (-2)^{-D[{s},{s}']} \ketbra{s}{s} \otimes \ketbra{s'}{s'} .
\end{align}
We note that $\hat{X}(u^{(r)})$ is the precisely  U-statistic~\cite{Hoeffding1992} for the  two-copy expectation value $\trAE{A (u^{(r)}\rho u^{(r)\dagger})^{\otimes 2}}$. Averaging over many computational basis measurements, we find thus for a fixed unitary $u^{(r)}$
\begin{align}
    \expect[\text{QM}]{ \hat{X}(u^{(r)})} &=\trAE{A (u^{(r)}\rho u^{(r)\dagger})^{\otimes 2}} = X({u^{(r)}}) \; .
\end{align}

Secondly, taking the outcomes for all unitaries $r=1,\dots,N_u$ together, we define  the estimator 
\begin{align}
\hat{X} = \frac{1}{N_u} \sum_{r=1}^{N_u} \frac{ \hat{X}(u^{(r)}) }{p_{\text{IS}}(u^{(r)})} .
\end{align}
As shown in Refs.~\cite{Elben2018,Elben2018a,Brydges2019},  $\hat{X}$ is an unbiased estimator of the purity $\trAE{\rho^2}$, i.e.\ 
\begin{align}
\expect{ \hat{X} }   &\equiv \expect[u]{ \expect[\text{QM}]{ \hat{X}}} \nonumber = \expect[u]{\frac{ X(u)} {p_{\text{IS}}(u)}} \nonumber \\ &= \int \du {X}(u) =    \trAE{\rho^2}.
\end{align}
Here, $\expect[u]{f(u)}\equiv \int \du p_{\text{IS}}(u) f(u)$ denotes the average over the importance sampling distribution $p_{\text{IS}}(u) \text{d}u$, with $\text{d}u=\prod_i \text{d}u_i$ and $\text{d}u_i$ the Haar measure  on the unitary group $U(2)$. 
In the case of $p_{\text{IS}}(u)=1$ (uniform sampler), $\mathbb{E}_u$ is thus simply the uniform average over  local random unitaries of the form $u=\bigotimes u_i$.

Our aim is to calculate the variance of $\hat{X}$ for finite $N_u$ and $N_M$, governing the statistical errors in our protocol. Summarizing our findings, we find:
\newtheorem{prop}{Proposition}
\begin{prop}
The variance of $\hat{X}$ is given by
\begin{align}
\var \left[\hat{X}\right]=& \frac{1}{N_u} \left(\frac{(N_M-3)(N_M-2)}{N_M (N_M-1)} \; \Gamma_4  + \frac{4 (N_M-2)}{N_M (N_M-1)}\; \Gamma_3    \right. \nonumber \\ &\qquad \quad +\left. \frac{2}{N_M (N_M-1)}\; \Gamma_2 - \trAE{\rho^2}^2 \right) \;. \label{varr}
\end{align}
Here, the coefficients $\Gamma_k$ are given by
\begin{align}
    \Gamma_k &= \mathbb{E}_u\left[ \frac{\trAE{ A_k  (u\rho u^\dag)^{\otimes k}}} {X_{\textnormal{IS}}^2(u)} \right]
\end{align}
where the $k$-copy operators $A_k$  are defined as
\begin{align}
A_4&= A \otimes A \label{eq:A4}\\
A_3& = (\mathbb{1}\otimes A)( {A}\otimes \mathbb{1}) \label{eq:A3} \\ 
A_2 &= A^2 \label{eq:A2}
\end{align}
with $A = 2^N \sum_{s,s'} (-2)^{-D[{s},{s}']} \ketbra{s}{s} \otimes \ketbra{s'}{s'}$. 
\end{prop}
\begin{proof}
We first note that $\var \left[\hat{X}\right] = 1/N_u\var\left[ \hat{X}(u^{(r)})\right] $ for any $r=1,\dots N_u$ due to independence and identical distribution of the sampled local random unitaries and of the separate single shot quantum measurements  (Born's rule). For simplicity of notation, we suppress thus in the following the index $r$.   Secondly, we have that $\var \hat{X}(u)= \mathbb{E} \left[  \hat{X}^2(u) \right] - \mathbb{E} \left[  \hat{X}(u) \right]^2$.  As stated above, we have  $\mathbb{E} \left[  \hat{X}(u) \right]=\trAE{\rho^2}$ and thus we concentrate in the following on the non-trivial first term. We find
\begin{align}
&\bbe{\frac{\hat{X}^2(u)}{X_{\textnormal{IS}}^2(u)}}\binom{N_M}{2}^{2}  \label{eq:profprop1_1} \\ &=  \! \! \! \sum_{\substack{m> n \\ m' > n'}} \! \! \bbe{  \trAE{ A^{\otimes 2}  \; \ps{m}\otimes \ps{n} \otimes  \ps{m'}\otimes \ps{n'}}X_{\textnormal{IS}}^{-2}(u) } \nonumber \\
&= \! \! \! \sum_{\substack{m> n \\ m' > n'}} \! \!  \expect[u]{ \expect[\text{QM}]{ \trAE{ A^{\otimes 2}  \; \ps{m}\otimes \ps{n} \otimes  \ps{m'}\otimes \ps{n'}}}X_{\textnormal{IS}}^{-2}(u) }\!. \nonumber
\end{align}
Magnitude and type of the different terms in this sum depend on how many indices in the expression $\expect[\text{QM}]{ \trAE{ A^{\otimes 2}  \; \ps{m}\otimes \ps{n} \otimes  \ps{m'}\otimes \ps{n'}}}$ coincide. We distinguish the following possibilities:
\begin{itemize}
    \item All indices are pairwise distinct, i.e.\ $m\neq m'$ and $n\neq n'$.  In this case, the expectation value $\mathbb{E}_{\text{QM}}$ completely factorizes, yielding a contribution
    \begin{align*}
    &\expect[\text{QM}]{ \trAE{ A^{\otimes 2}  \; \ps{m}\otimes \ps{n} \otimes  \ps{m'}\otimes \ps{n'}}} \\
    &\quad=\trAE{ A^{\otimes 2}  (u\rho u^\dag)^{\otimes 4}}.
    \end{align*}
    which is fourth order in the density matrix $\rho$.
    In total, there are $\binom{N_M}{4} \binom{4}{0} \binom{4}{2}$ such terms.
    \item Exactly two indices coincide, e.g.\ $m=m'$ and $n\neq n'$.  In this case,  we obtain a third order contribution
    \begin{align*}
    &\expect[\text{QM}]{ \trAE{ A^{\otimes 2}  \; \ps{m}\otimes \ps{n} \otimes  \ps{m}\otimes \ps{n'}}} \\
    &\quad=\trAE{ A_3  (u\rho u^\dag)^{\otimes 3}}
    \end{align*}
    with $A_3=(\mathbb{1}\otimes A)( {A}\otimes \mathbb{1})$.
    In total, there are $\binom{N_M}{3} \binom{3}{1} \binom{2}{1}$ such terms.
    \item Two pairs  of indices coincide,  e.q.\ $m= m'$ and $n= n'$.  In this case,  we obtain a second  order contribution
    \begin{align*}
    &\expect[\text{QM}]{ \trAE{ A^{\otimes 2}  \; \ps{m}\otimes \ps{n} \otimes  \ps{m}\otimes \ps{n}}} \\
    &\quad=  \trAE{ A^2  (u\rho u^\dag)^{\otimes 2}}\; .
    \end{align*}
    In total, there are $\binom{N_M}{2} \binom{2}{2} \binom{0}{0}$ such terms.
\end{itemize}
Inserting into Eq.~\eqref{eq:profprop1_1} and summing all terms up, we obtain 
\begin{align*}
&\bbe{\frac{\hat{X}^2(u)}{X_{\textnormal{IS}}^2(u)}}=  \frac{(N_M-3)(N_M-2)}{N_M (N_M-1)} \; \Gamma_4  + \frac{4 (N_M-2)}{N_M (N_M-1)}\; \Gamma_3     \nonumber\\ &\qquad \quad + \frac{2}{N_M (N_M-1)}\; \Gamma_2 \;. 
\end{align*}
Using $\var \left[\hat{X} \right] = 1/N_u\var\left[ \hat{X}(u)\right] $ yields the claim.
\end{proof}

\newtheorem{lem}{Lemma}

\begin{lem}[Uniform sampling]
Suppose $\rho$ is pure product state  of $N$ qubits \footnote{Using Weingarten calculus and techniques presented in Ref.~\cite{Zhou2020}, we can generalize the variance for uniform sampling to pure product states of $N$ qudits with arbitrary local dimension $d$ (i.e.\ including the case of global random unitaries). We find
$$
\Gamma_4=\left(\frac{d^2+9 d+2}{d^2+5 d+6} \right)^N,
\Gamma_3=\left(\frac{3d}{2+d}\right)^N,
\Gamma_2=\left(2d-1\right)^N .
$$
} and assume that the local random unitaries $u_i$ are uniformly sampled from the Haar measure, i.e. $X_{\text{IS}}(u)=1$.  Then, we find 
\begin{align}
    \Gamma_4 &=\left(\frac{6}{5}\right)^N
    \\
    \Gamma_3  &=\left(\frac{3}{2}\right)^N
     \\
  \Gamma_2 &=3^N \; .
\end{align} 
\end{lem}
\begin{proof}
We first note that in the case of qubits
\begin{align}
    A &= \frac{1}{2^N}\bigotimes_i(1+3\sigma^i_z \otimes \sigma^i_z) \; .
\end{align}
Assuming without loss of generality that $\rho = \ketbra{0}{0}^{\otimes N}$, it holds thus that
\begin{align}
    X(u) = \trAE{A (u\rho u^\dagger)^{\otimes 2}} = \frac{1}{2^N} \prod_i (1+3Z(u_i)^2) 
\end{align}
with the expectation value $Z(u_i)=\braket{0|u^\dagger_i\sigma_z u_i|0}$.
Secondly, we  define $\Gamma_k(u)=\trAE{ A_k  (u\rho u^\dagger)^{\otimes k}}$. Using the definition of $A_k$ [Eqs.~\eqref{eq:A4}-\eqref{eq:A2}], we find
\begin{align}
    \Gamma_2(u) &=\frac{1}{4^N} \prod_i (10+6 \, Z(u_i)^2)\\
        \Gamma_3(u) &=\frac{1}{4^N} \prod_i (1+15 \, Z(u_i)^2)\\
        \Gamma_4(u) &=\frac{1}{4^N} \prod_i (1+3 \,Z(u_i)^2)^2 \; .
\end{align}
 Employing the decompositions \eqref{eq:huw} and \eqref{eq:measure}, we can rewrite $Z(u_i)=1-2 \xi_i$ where $\xi_i$ is uniformly distributed in $[0,1]$. With the substitution $z_i=1-2 \xi_i$, we  find
\begin{align*}
\Gamma_2  &= \int \du \Gamma_2(u) = \left[\frac{1}{8}  \int_{-1}^1 dz (10+6 \, z^2) \right]^N = 3^N \\
\Gamma_3  &= \int \du \Gamma_3(u) = \left[\frac{1}{8}  \int_{-1}^1 dz (1+15 \, z^2) \right]^N= \left(\frac{3}{2}\right)^N \\
\Gamma_4 &= \int \du \Gamma_4(u) = \left[\frac{1}{8}  \int_{-1}^1 dz (1+3 \,z^2)^2 \right]^N=  \left(\frac{6}{5}\right)^N \;.
\end{align*}
Here, we used that $X_{\text{IS}}(u)=1$ (uniform sampler) and that the integral factorizes due to the independence of the local random unitaries.
\end{proof}

\begin{lem}[Perfect sampler]
Suppose $\rho$ is pure product state  of $N$ qubits and assume that  $X_{\text{IS}}(u)=X(u)$ defines a perfect importance sampler.  Then, we find
\begin{align}
    \Gamma_4 &=1
    \\
    \Gamma_3  &=\alpha^N
     \\
  \Gamma_2 &=\beta^N  
\end{align}
with $\alpha=\frac{5}{2}-\frac{2 \pi }{3 \sqrt{3}}\approx 1.29$, and $\beta=1+\frac{4 \pi }{3 \sqrt{3}}\approx 3.42$. 
\end{lem}
\begin{proof}
With the definitions and notation of the proof of Lemma 1, we have
\begin{align*}
\Gamma_2  &= \int \du \frac{\Gamma_2(u)}{X(u)} = \left[\frac{1}{4}  \int_{-1}^1 dz \frac{10+6 \, z^2}{1+3z^2} \right]^N = \left(1+\frac{4 \pi }{3 \sqrt{3}}\right)^N \\
\Gamma_3  &= \int \du \frac{\Gamma_3(u)}{X(u)} =  \left[\frac{1}{4}  \int_{-1}^1 dz \frac{ 1+15 \, z^2}{1+3z^2} \right]^N= \left(\frac{5}{2}-\frac{2 \pi }{3 \sqrt{3}} \right)^N \\
\Gamma_4 &= \int \du \frac{\Gamma_4(u)}{X(u)} = \left[\frac{1}{4}  \int_{-1}^1 dz  (1+3z^2) \right]^N=  1 \;.
\end{align*}
\end{proof}
\begin{figure}[h]
\includegraphics[width=0.78\columnwidth]{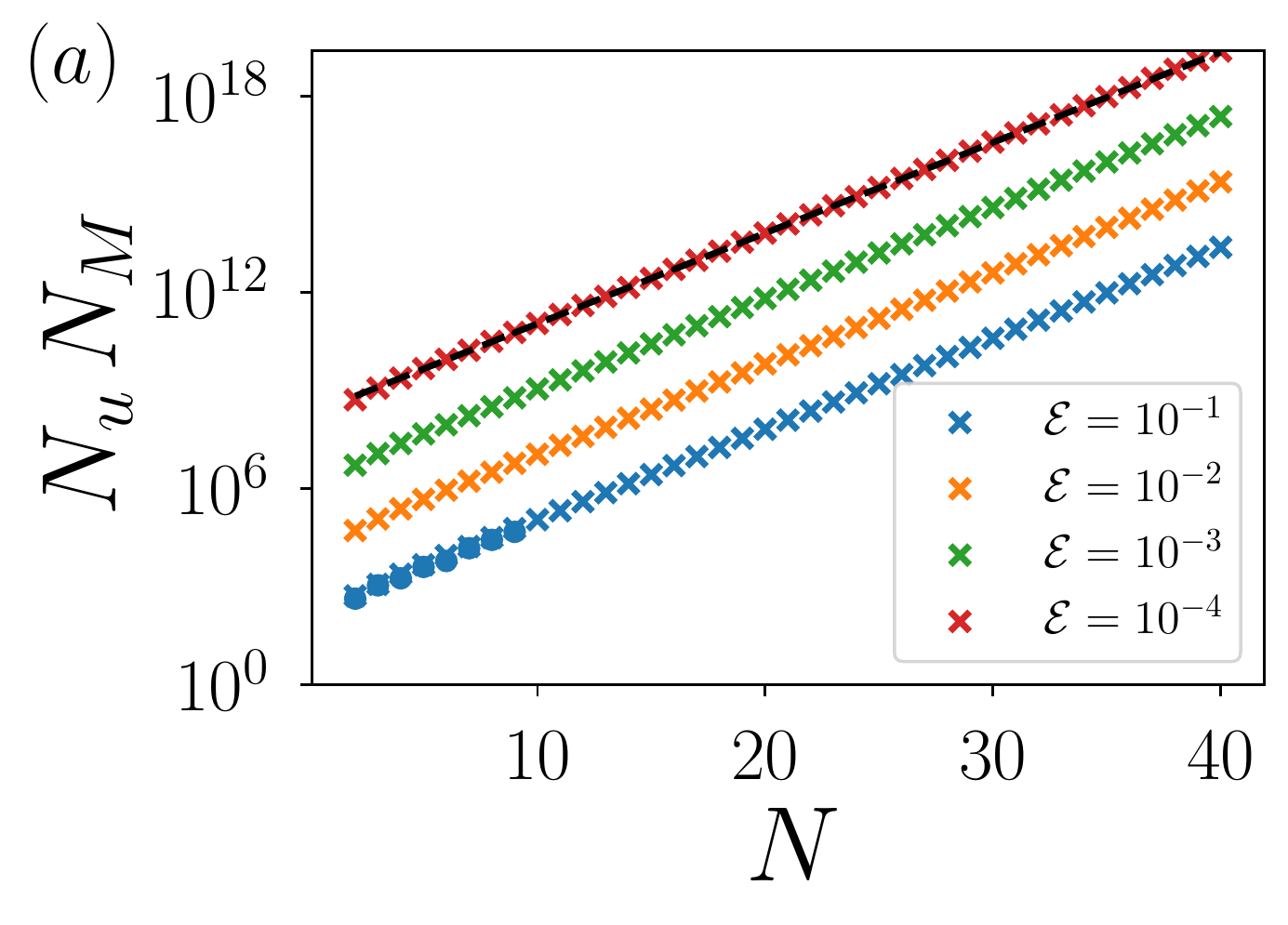}
\includegraphics[width=0.78\columnwidth]{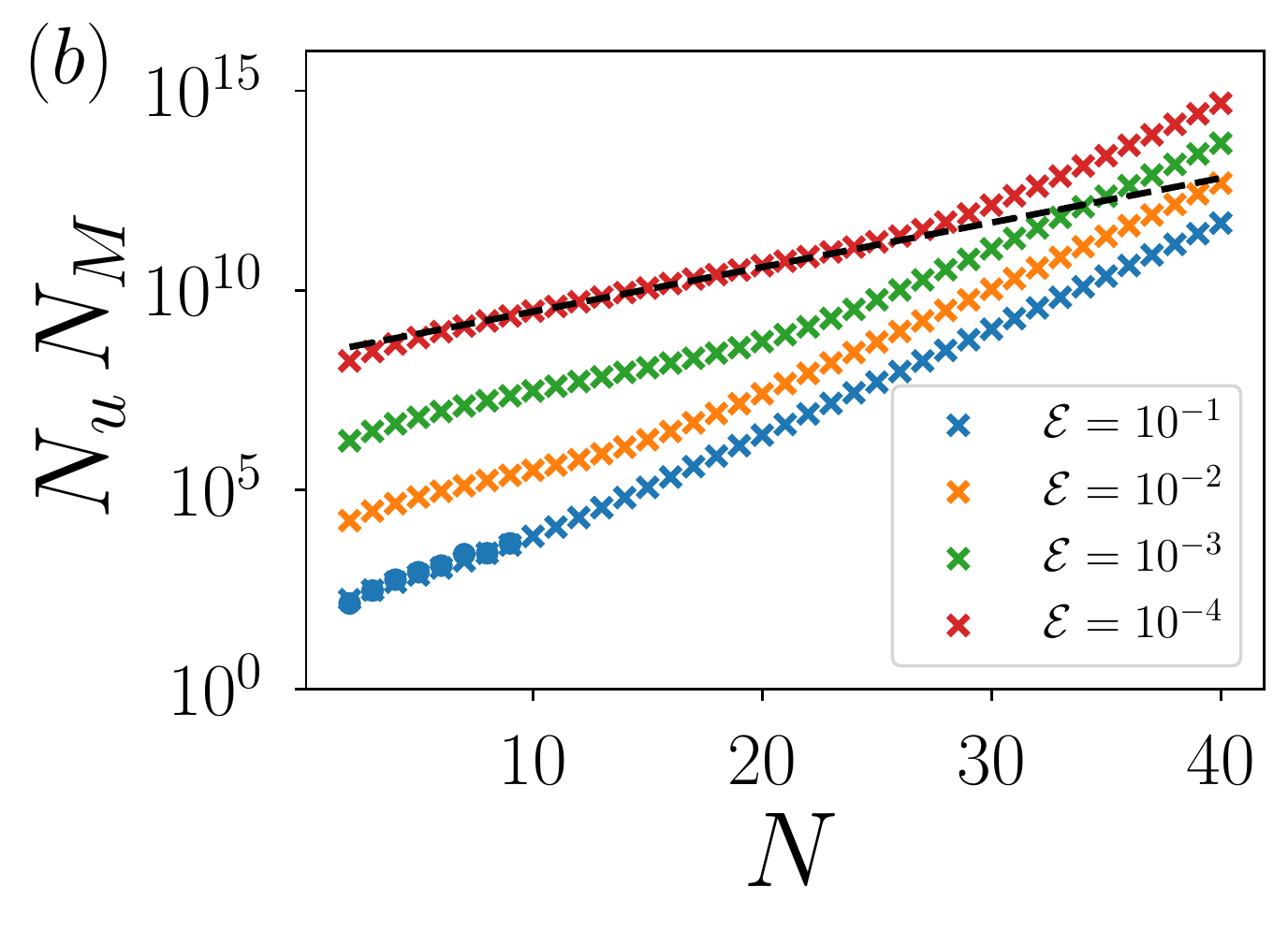}
\caption{{\it Scaling of the total number of measurements in different error regimes $\mathcal{E}$ for a product state.} Panels (a) uniform sampling and (b) importance sampling present the scaling of the required total number of measurements $N_uN_M$ to achieve different levels of accuracy (or error regimes) as a function of the system size $N$. The points in cross are analytical results and circles are the corresponding values obtained by numerical simulations of our protocol. The black dashed line ($\propto 2^{aN}$) is a guide for the asymptotic scaling in the high accuracy limit ($\mathcal{E} \to 0$) with $a = 0.92$ for uniform sampling and $a = 0.37$ for importance sampling.
\label{fig:SM_scaling_analytics}}
\end{figure}
In order to relate our numerics to the analytical expression of Eq.~\eqref{varr} obtained earlier, we could assume a normal distribution of the values $\hat{X}_e$ obtained by simulating numerically the experimental protocol for different set of values of $N_u$ and $N_M$ in the case of uniform sampling ($X_{\mathrm{IS}}(u) = 1$) and importance sampling from a perfect sampler ($X_{\mathrm{IS}}(u) = X(u)$).
In that case, $\mathrm{Std}[\hat{X}] = \sqrt{\frac{\pi}{2}}\,\cal{E}$ where $\mathcal{E} = \overline{|\hat{X}_e - \hat{X}|}$ is the average statistical error in estimating the purity. Note that we have checked numerically that are our samples of purity estimations are indeed approximately normally distributed. 
For a given value of $\cal{E}$, we could then extract analytically the optimal required number of measurements $N_uN_M$ by using Eq.~\eqref{varr} in the case of uniform and importance sampling.
This is illustrated in Fig. 2(c,d) of the main text. 

From the analytical expressions of the variance derived above in Lemma 2, we observe in Fig.~\ref{fig:SM_scaling_analytics}  for importance sampling two different regimes of scaling. For $N\lesssim N_c$ smaller than a certain $N_c \sim \log (1/\mathcal{E})$, the scaling exponent $a$ for importance sampling is strongly reduced $a=0.37$, and increases to $a=0.88$ for $N>N_c$. In contrast, uniform sampling does not present this feature and has a nearly constant scaling of $a = 0.92$. While being advantageous in all displayed cases, importance sampling is thus most powerful in the high accuracy regime (small errors $\mathcal{E}$).

% From the analytical expressions of the variance derived above in Lemma 1 and 2, we observe in the limit of high accuracy ($\mathcal{E} \to 0 $) that the scaling exponent $a$ for importance sampling continues to decrease and converges to an asymptotic value of 0.37 for $N < N_c$ while uniform sampling doesn't present this feature and has a constant scaling of $a = 0.92$ in all accuracy regimes. These features are more clearly highlighted in Fig.~\ref{fig:SM_scaling_analytics} for different values of error $\mathcal{E}$.

\section{Appendix F: Further Results}
\subsection{Scaling analysis}
In this section, firstly, we provide different scaling analysis for the product state, and the GHZ state for a fixed value of $N_u$. We represent the average statistical error $\mathcal{E} = \overline{|p_2 - {p_2}_{e}|}$ computed over 100 experimental runs, as a function of the rescaled units $N_M/2^{aN}$ where $a$ is adjusted so that points for different system sizes $N$ collapse into one curve as shown in Fig.~\ref{fig:SM_scaling}(a-d). Here, we performed importance sampling from trained neural network ansatz. We can identify an error scaling where we observe that importance sampling scaling is approximately divided by 2. This suggests that the required number of measurements $N_M$ to reach a given accuracy is exponentially reduced by using importance sampling as we pick unitaries that reduce the effect of shot noise.

Secondly, we show additionally the scaling for a pure random state taken from the Haar measure with the importance sampling was done from a perfect sampler ($X_{\mathrm{IS}}(u) = X(u)$). Fig.~\ref{fig:SM_scaling}(e-f) show the scaling for uniform and importance sampling where we plot the optimal required number of measurements $N_uN_M$ for a given average statistical error $\cal{E}$ as a function of the system size $N$ for uniform and importance sampling. 

% Fig.~\ref{fig:ML}a-b) illustrates the characteristic error scaling $1/\sqrt{N_u}$ that is expected in the limit $N_M\to \infty$, with a prefactor $\approx  5$ times smaller for importance sampling. More importantly, the effect of shot-noise (finite $N_M$) is also greatly reduced when using importance sampling. We represent in Fig.~\ref{fig:ML}c-d), the errors $\mathcal{E}$ for one value of $N_u$, as a function of $N_M/2^{bN}$. To identify an error scaling, the value of $b$ is adjusted so the data obtained for different values of $N$ collapse onto a single curve. Note that with uniform sampling, estimations can also be realized via classical shadows, with improved performances for GHZ states, but with the same scalings with $N$~\cite{Huang2020,Elben2020}. Remarkably, we observe that $b$ is approximately divided by $2$, when we start using importance sampling, which means that the required numbers $N_M$ to achieve a given accuracy is exponentially reduced. In practical terms, our small-scale analysis suggests that, for a given measurement budget, and for these two particular states, we can double the system size where the purity can be measured.
\begin{figure}[h]
\includegraphics[width=0.48\columnwidth]{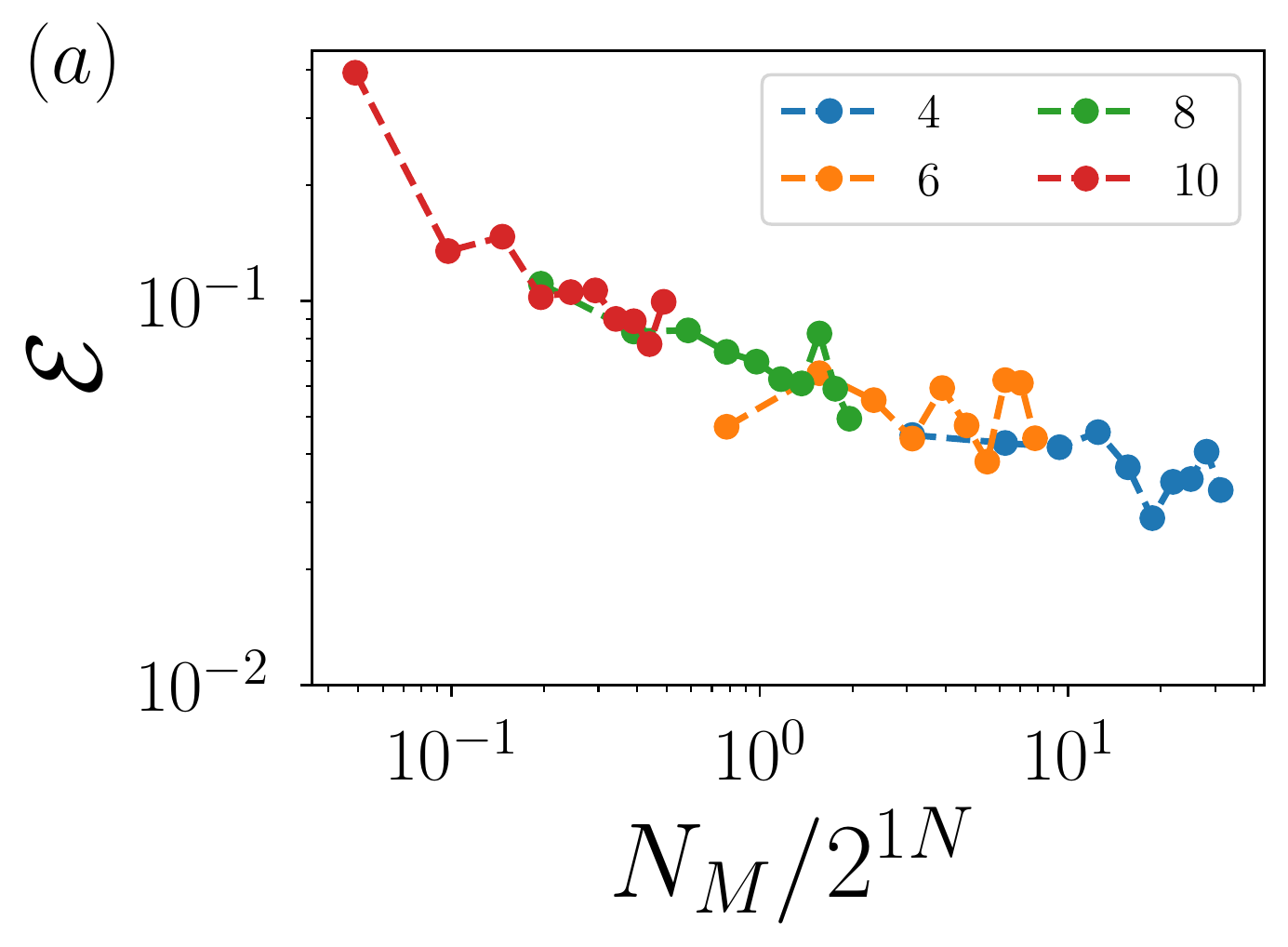}
\includegraphics[width=0.48\columnwidth]{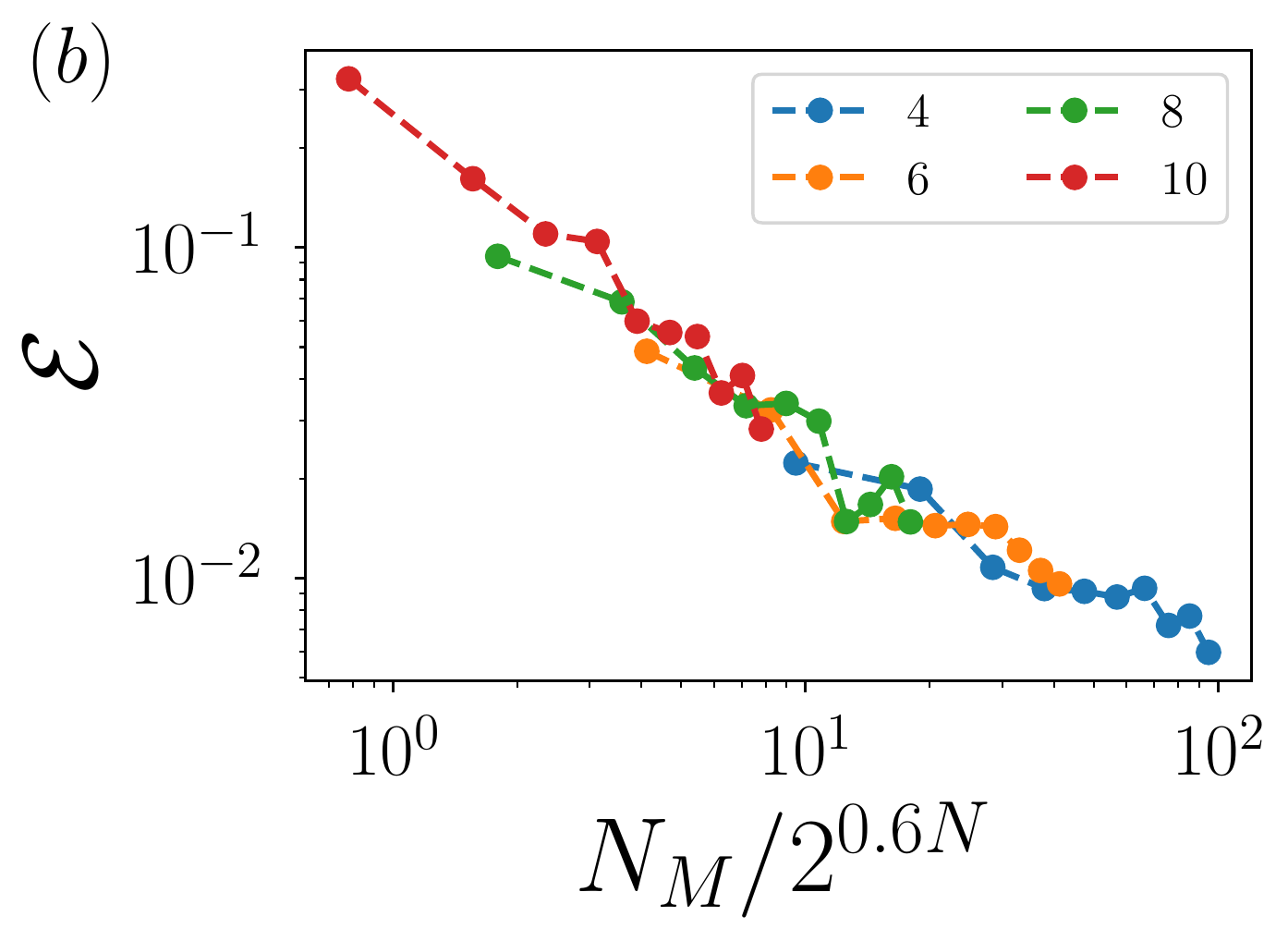}
\includegraphics[width=0.48\columnwidth]{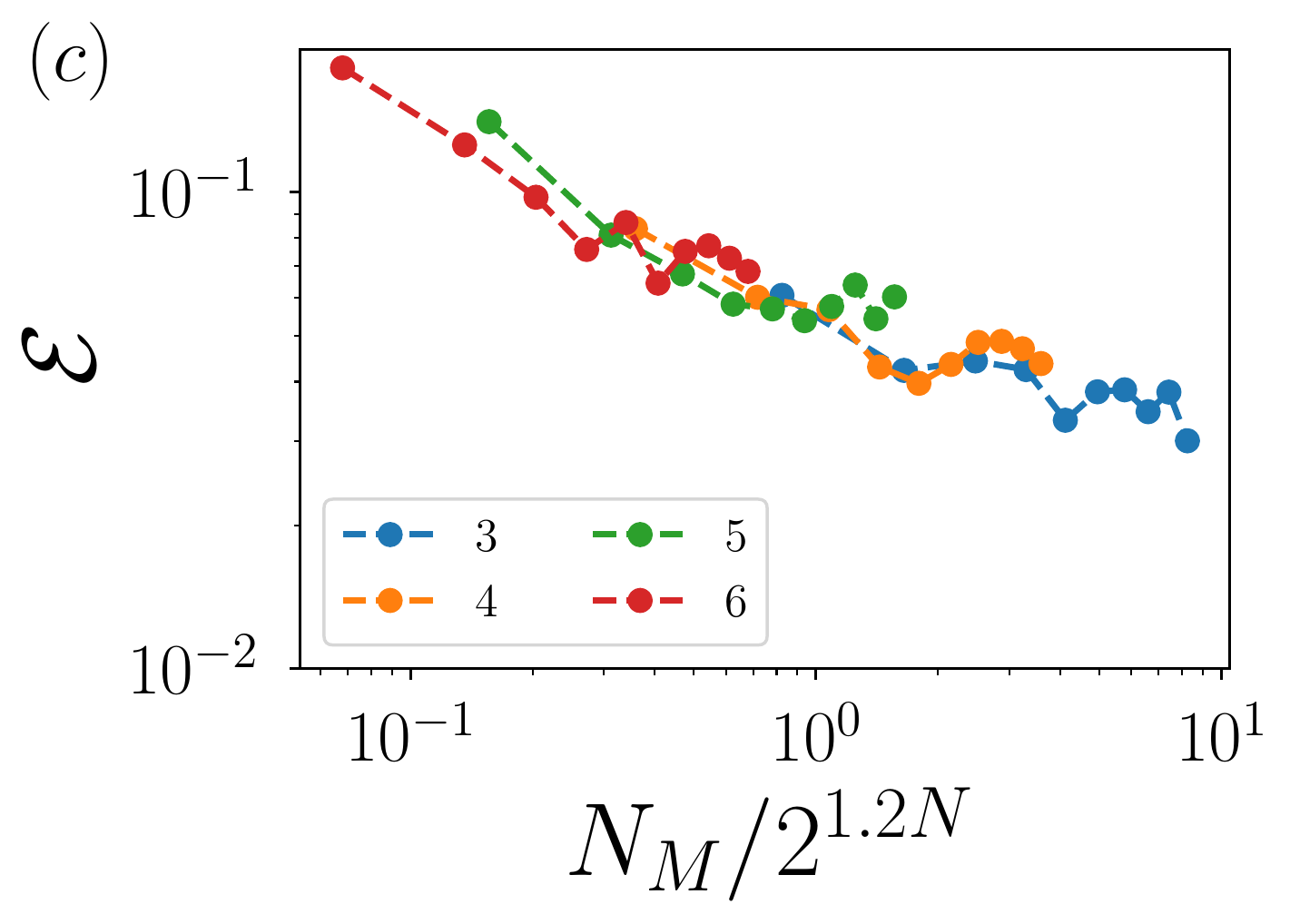}
\includegraphics[width=0.48\columnwidth]{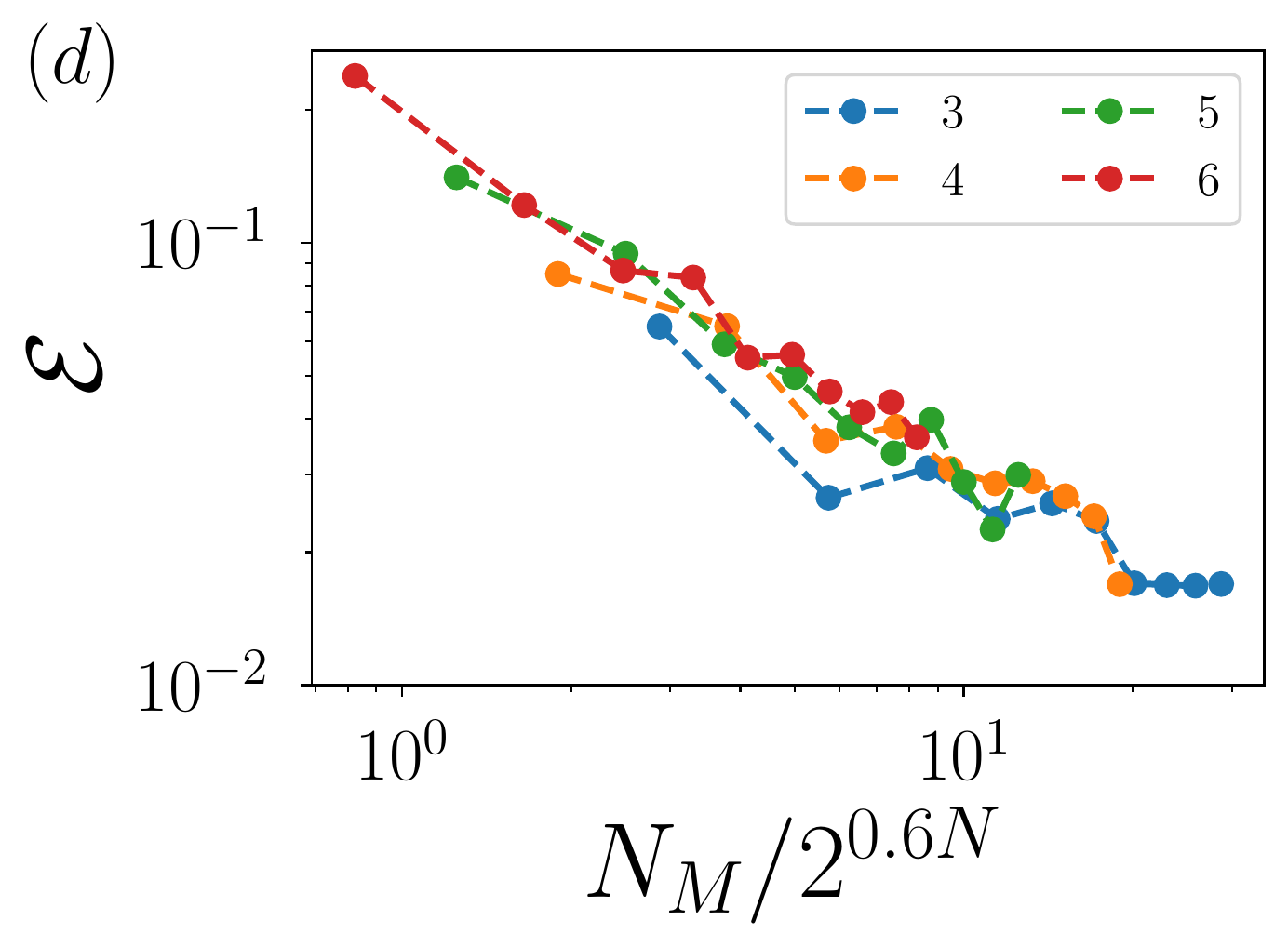}
\includegraphics[width=0.48\columnwidth]{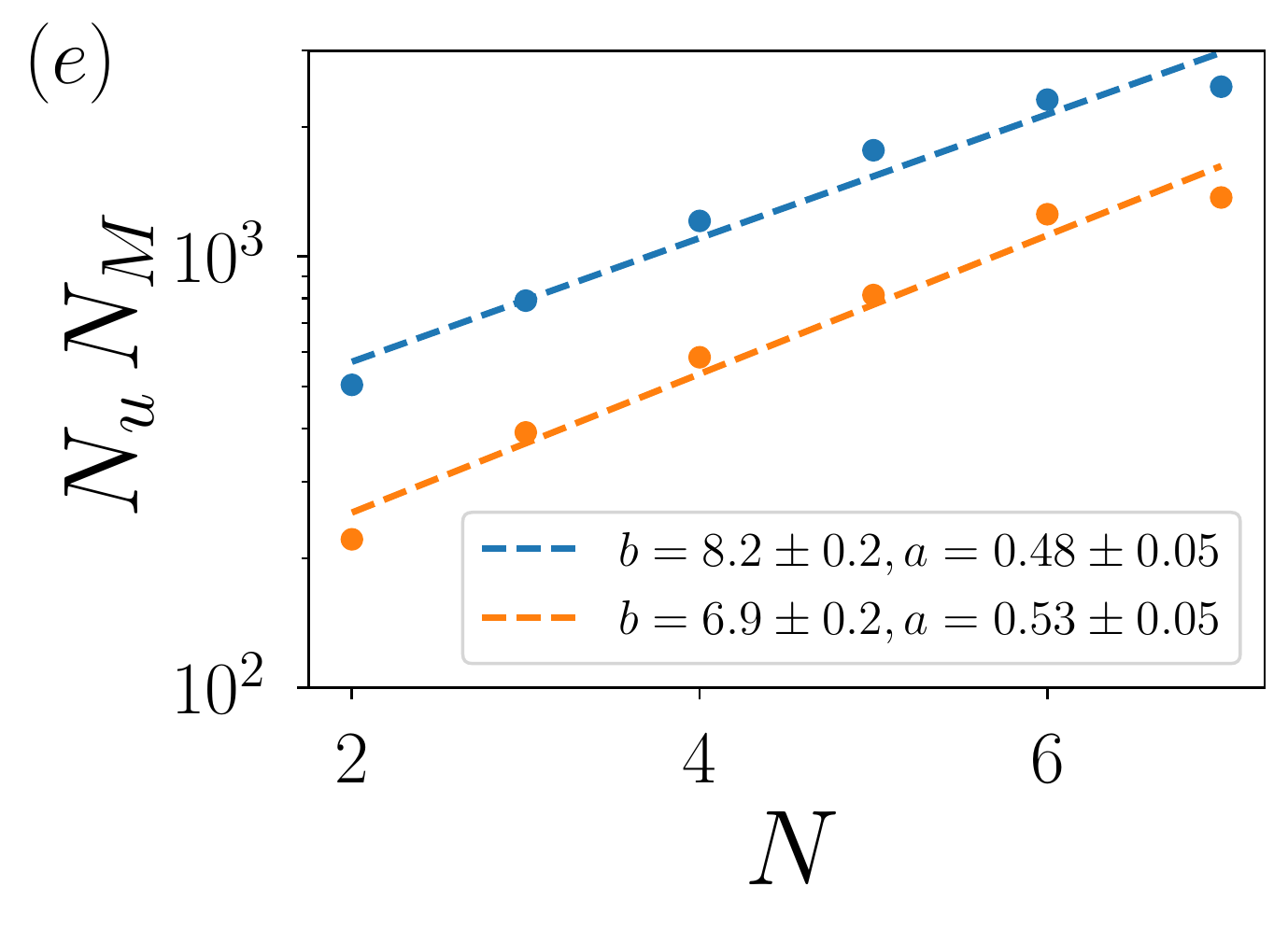}
\includegraphics[width=0.48\columnwidth]{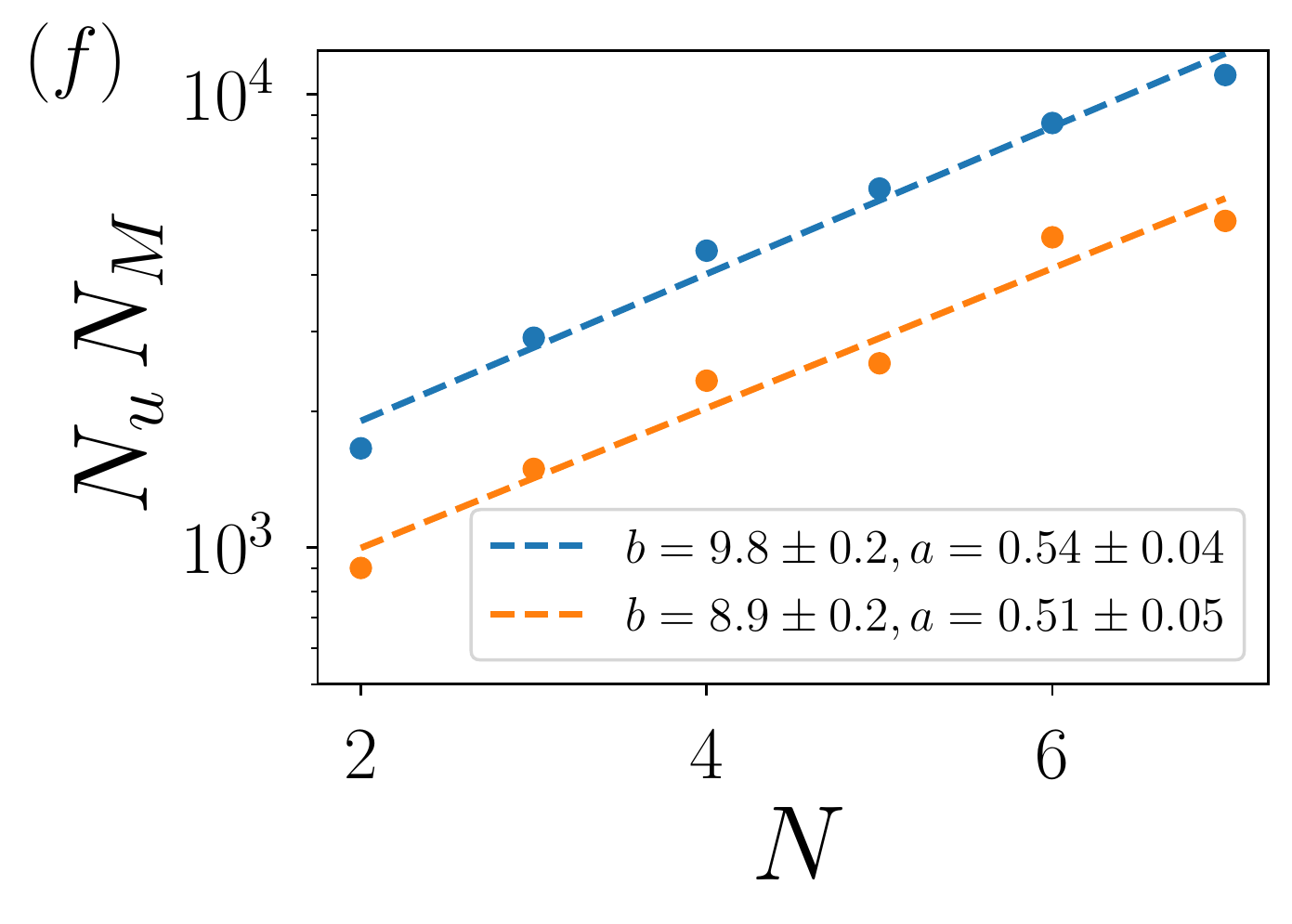}
\caption{{\it Scaling of statistical errors for product, GHZ and random states.} Panels (a) and (b) (product state) and (c) and (d) (GHZ state) show the error in function of the rescaled unit $N_M/2^{aN}$ for different system sizes for a fixed value of $N_u$ ($N_u = 500$ for product state and $N_u = 200$ for GHZ state). The unitaries in the left panels (a) and (c) being sampled from a uniform distribution while the right panel (b) and (d) are sampled from the trained neural network. Panel (e) and (f) highlight the scaling  of the required total number of measurements $N_uN_M$ as a function of $N$ for uniform and importance sampling for a pure random state, to obtain a statistical error of $\mathcal{E}=0.1$ (e), and $\mathcal{E}=0.05$ (f).
\label{fig:SM_scaling}}
\end{figure}
\subsection{Mixed state sub-system optimisation}
We further analyze the performance of importance sampling for a reduced state $\rho_A = \mathrm{Tr_B}(\rho_{AB})$ where $A$ is the half partition of the highly entangled 10-qubit quantum simulation state denoted by $\rho_{AB}$ studied earlier in the main text. The half-partition purity was found to be $\mathrm{Tr}(\rho_{A}^2) = 0.16$. 
The role of importance sampling continues to be relevant for probing such partitions of states that are highly mixed because the bitstring probabilities obtained in the experiments take values in a much reduced interval ~\cite{Brydges2019}. Therefore, it is important to sample the adequate unitaries through importance sampling to provide better probability signals for our concerned estimates.

\begin{figure}
\vspace{1cm}
\begin{minipage}[b]{0.5\linewidth}
\centering
\includegraphics[width=\textwidth]{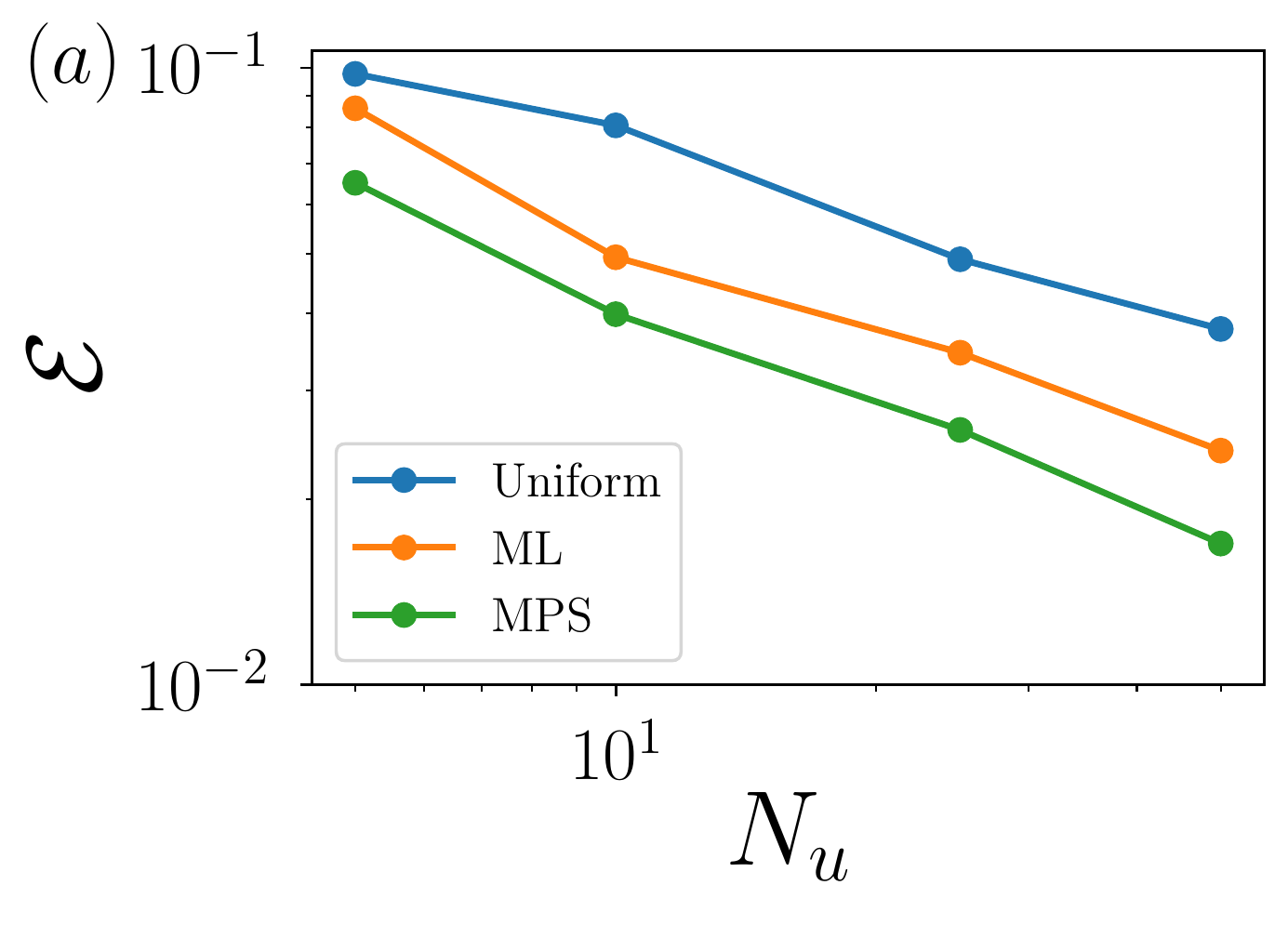}
\end{minipage}
\hskip -1ex
\begin{minipage}[b]{0.5\linewidth}
\centering
\includegraphics[width=\textwidth]{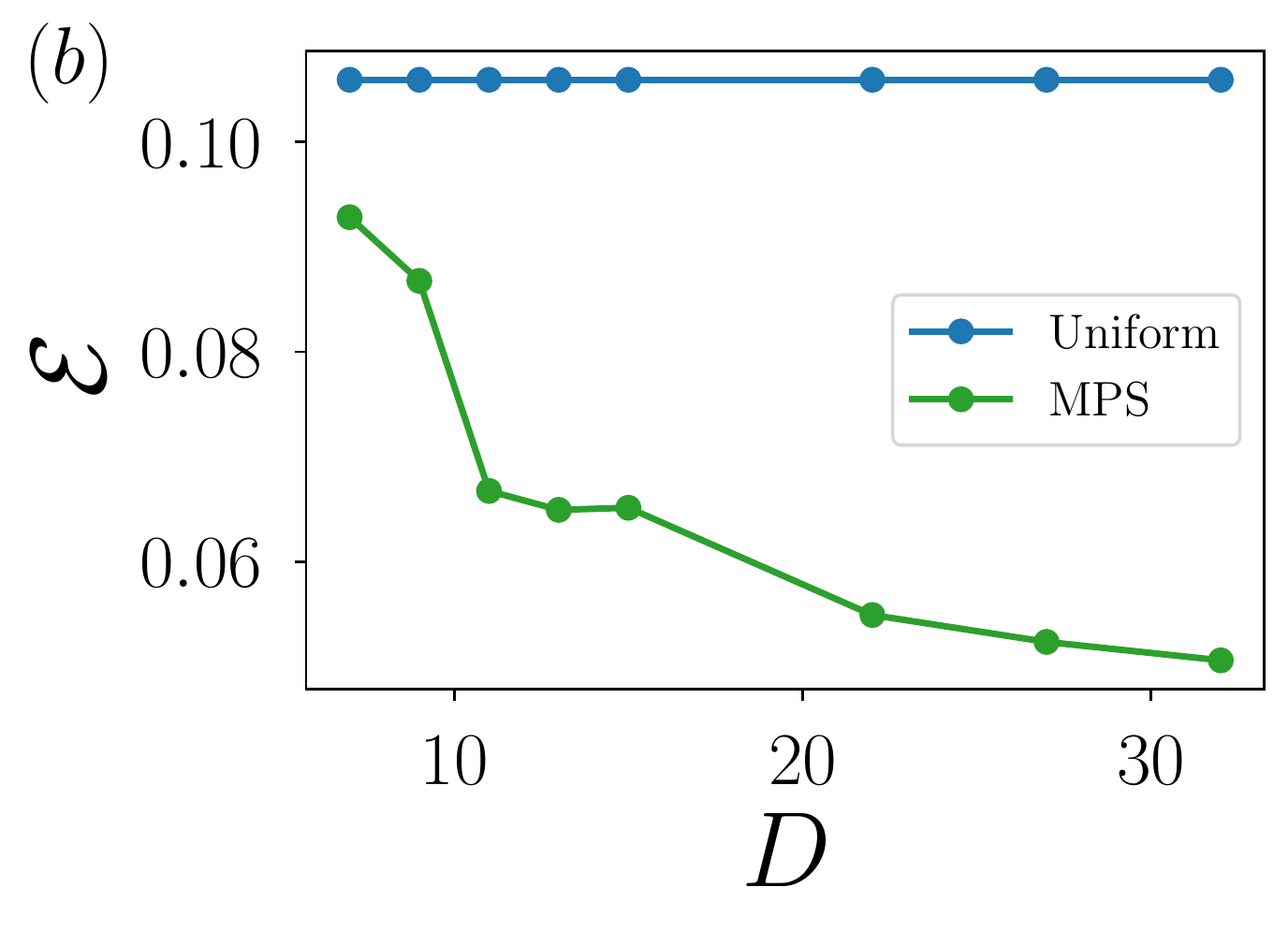}
\end{minipage}
\hskip -1ex
\begin{minipage}[b]{0.75\linewidth}
\centering
\includegraphics[width=\textwidth]{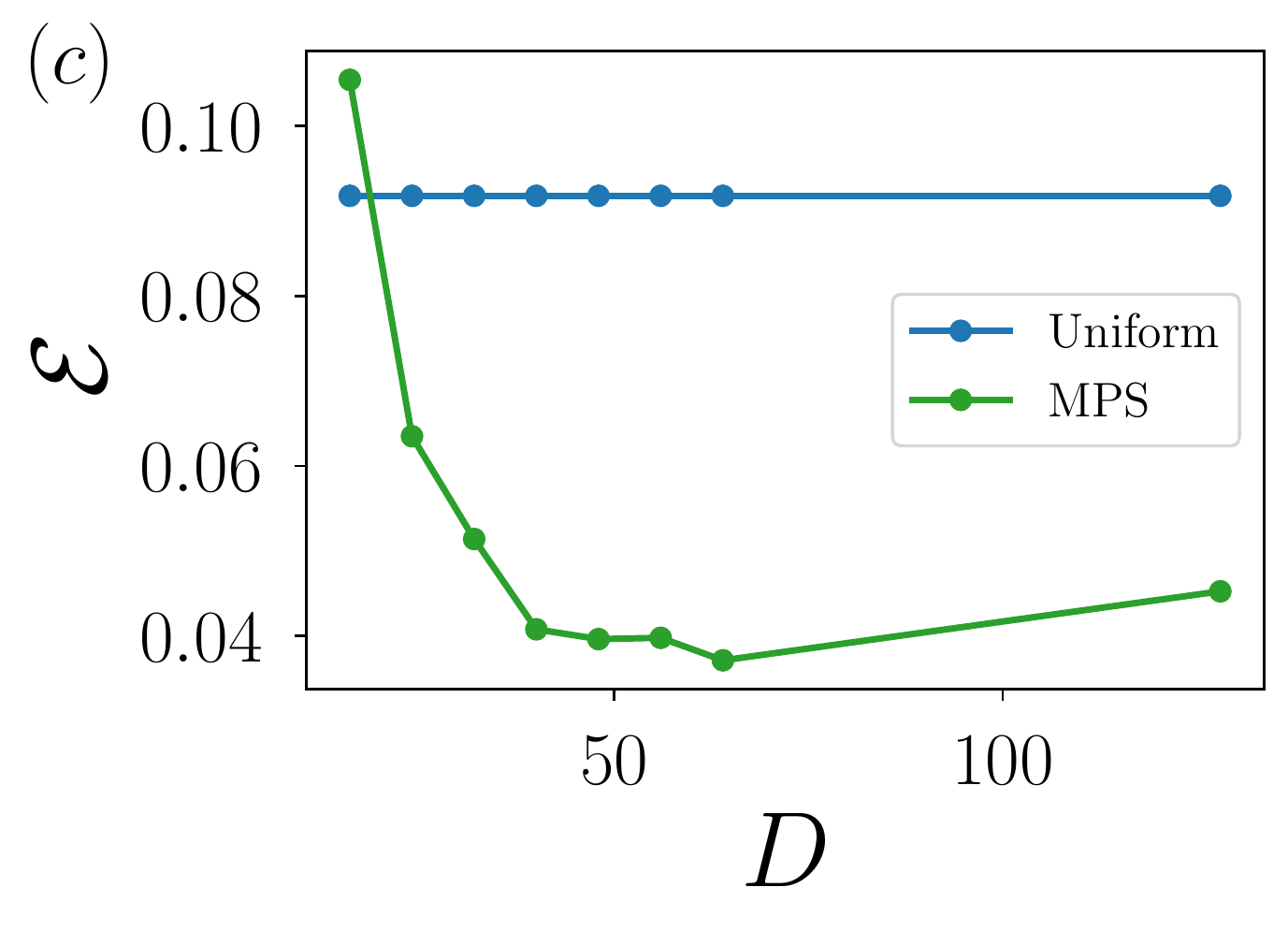}
\end{minipage}
\caption{{\it Purity estimation of half-partition mixed state of the highly entangled 10 and 20 qubit states with ML and MPS samplers.} Panel (a) shows the average statistical error $\mathcal{E}$ of the estimated purity in function of $N_u$ with $N_M = 7500$ for a uniform sampling and importance sampling done from a NN and a MPS representation of the corresponding reduced state. Panel (b) illustrates the scaling of the error $\mathcal{E}$ w.r.t different bond dimensions $D$ used for the MPS representation of the reduced state for $N_u = 5$ and $N_M = 7500$. Panel (c) shows the scaling performance for a 10-qubit reduced state from a 20-qubit system with $N_u = 5$ and $N_M = 10^5$. \label{fig:mixed}}
\end{figure}
We can construct the importance sampler $X_\mathrm{IS}(u)$ for $\rho_A$ in two ways: first, by training a NN on the reduced pure theory state $\rho_A$ to build a fit function and secondly using an approximate MPS representation $\ket{\psi_D}$ of bond dimension $D$ of the reduced state $\rho_\mathrm{red} = \mathrm{Tr_B}(\ket{\psi_D}\bra{\psi_D})$. Fig.\ref{fig:mixed}a already shows a reduction of statistical error compared to an uniform sampling by using the best trained importance sampler CNN of the mixed reduced state. The statistical errors can be further reduced by using the MPS representation of the reduced state of $D = 15$ with a fidelity overlap to the true state being $\mathcal{F}(\rho_A,\rho_\mathrm{red}) = 0.83$. In Fig.\ref{fig:mixed}b we show the reduction of statistical errors as a function of the bond dimension $D$ of the reduced MPS representation. 

We equally highlight in Fig.~\ref{fig:mixed}c the performance of importance sampling on a mixed reduced 10-qubit state $\rho_A$ taken from an entangled $20-$qubit system with the reduced state purity $\tr(\rho_A^2) = 0.103$. For the parameters of Ref.~\cite{Brydges2019}, the evolution time with the $XY$ model was set to $t=7.5$ ms, and we considered here a pure state approximation $\rho=\ket{\psi}\bra{\psi}$.
The importance sampler state was taken to be the MPS approximation $\rho_{\mathrm{red}} = \tr_B(\ket{\psi_D}\bra{\psi_D})$ of varying bond dimension $D$ of the 20 qubit state $\ket{\psi}$. The MPS representation of $D = 24$ which has a fidelity $\mathcal{F}(\rho_A,\rho_\mathrm{red}) = 0.87$ already outperforms the uniform sampler (second point in Fig.~\ref{fig:mixed}c). We observe a similar reduction of statistical errors for moderate values of $D$, which emphasizes the strength of importance sampling with MPS models. 
%This idea can be generalized and implemented to all the $2^{10}$ bi-partitions of the 10 ion quantum simulator and is equally relevant to probe half partitions or other sub-systems of larger systems.

%\input{Optimisation_RM_v2.bbl}

\begin{thebibliography}{50}%
\makeatletter
\providecommand \@ifxundefined [1]{%
 \@ifx{#1\undefined}
}%
\providecommand \@ifnum [1]{%
 \ifnum #1\expandafter \@firstoftwo
 \else \expandafter \@secondoftwo
 \fi
}%
\providecommand \@ifx [1]{%
 \ifx #1\expandafter \@firstoftwo
 \else \expandafter \@secondoftwo
 \fi
}%
\providecommand \natexlab [1]{#1}%
\providecommand \enquote  [1]{``#1''}%
\providecommand \bibnamefont  [1]{#1}%
\providecommand \bibfnamefont [1]{#1}%
\providecommand \citenamefont [1]{#1}%
\providecommand \href@noop [0]{\@secondoftwo}%
\providecommand \href [0]{\begingroup \@sanitize@url \@href}%
\providecommand \@href[1]{\@@startlink{#1}\@@href}%
\providecommand \@@href[1]{\endgroup#1\@@endlink}%
\providecommand \@sanitize@url [0]{\catcode `\\12\catcode `\$12\catcode
  `\&12\catcode `\#12\catcode `\^12\catcode `\_12\catcode `\%12\relax}%
\providecommand \@@startlink[1]{}%
\providecommand \@@endlink[0]{}%
\providecommand \url  [0]{\begingroup\@sanitize@url \@url }%
\providecommand \@url [1]{\endgroup\@href {#1}{\urlprefix }}%
\providecommand \urlprefix  [0]{URL }%
\providecommand \Eprint [0]{\href }%
\providecommand \doibase [0]{http://dx.doi.org/}%
\providecommand \selectlanguage [0]{\@gobble}%
\providecommand \bibinfo  [0]{\@secondoftwo}%
\providecommand \bibfield  [0]{\@secondoftwo}%
\providecommand \translation [1]{[#1]}%
\providecommand \BibitemOpen [0]{}%
\providecommand \bibitemStop [0]{}%
\providecommand \bibitemNoStop [0]{.\EOS\space}%
\providecommand \EOS [0]{\spacefactor3000\relax}%
\providecommand \BibitemShut  [1]{\csname bibitem#1\endcsname}%
\let\auto@bib@innerbib\@empty
%</preamble>
\bibitem [{\citenamefont {Arute}\ \emph {et~al.}(2019)\citenamefont {Arute},
  \citenamefont {Arya}, \citenamefont {Babbush}, \citenamefont {Bacon},
  \citenamefont {Bardin}, \citenamefont {Barends}, \citenamefont {Biswas},
  \citenamefont {Boixo}, \citenamefont {Brandao}, \citenamefont {Buell} \emph
  {et~al.}}]{Arute2019}%
  \BibitemOpen
  \bibfield  {author} {\bibinfo {author} {\bibfnamefont {F.}~\bibnamefont
  {Arute}}, \bibinfo {author} {\bibfnamefont {K.}~\bibnamefont {Arya}},
  \bibinfo {author} {\bibfnamefont {R.}~\bibnamefont {Babbush}}, \bibinfo
  {author} {\bibfnamefont {D.}~\bibnamefont {Bacon}}, \bibinfo {author}
  {\bibfnamefont {J.~C.}\ \bibnamefont {Bardin}}, \bibinfo {author}
  {\bibfnamefont {R.}~\bibnamefont {Barends}}, \bibinfo {author} {\bibfnamefont
  {R.}~\bibnamefont {Biswas}}, \bibinfo {author} {\bibfnamefont
  {S.}~\bibnamefont {Boixo}}, \bibinfo {author} {\bibfnamefont {F.~G. S.~L.}\
  \bibnamefont {Brandao}}, \bibinfo {author} {\bibfnamefont {D.~A.}\
  \bibnamefont {Buell}},  \emph {et~al.},\ }\href {\doibase
  10.1038/s41586-019-1666-5} {\bibfield  {journal} {\bibinfo  {journal}
  {Nature}\ }\textbf {\bibinfo {volume} {574}},\ \bibinfo {pages} {505}
  (\bibinfo {year} {2019})}\BibitemShut {NoStop}%
\bibitem [{\citenamefont {Ebadi}\ \emph {et~al.}(2021)\citenamefont {Ebadi},
  \citenamefont {Wang}, \citenamefont {Levine}, \citenamefont {Keesling},
  \citenamefont {Semeghini}, \citenamefont {Omran}, \citenamefont {Bluvstein},
  \citenamefont {Samajdar}, \citenamefont {Pichler}, \citenamefont {Ho},\ and\
  \citenamefont {et~al.}}]{Ebadi2021}%
  \BibitemOpen
  \bibfield  {author} {\bibinfo {author} {\bibfnamefont {S.}~\bibnamefont
  {Ebadi}}, \bibinfo {author} {\bibfnamefont {T.~T.}\ \bibnamefont {Wang}},
  \bibinfo {author} {\bibfnamefont {H.}~\bibnamefont {Levine}}, \bibinfo
  {author} {\bibfnamefont {A.}~\bibnamefont {Keesling}}, \bibinfo {author}
  {\bibfnamefont {G.}~\bibnamefont {Semeghini}}, \bibinfo {author}
  {\bibfnamefont {A.}~\bibnamefont {Omran}}, \bibinfo {author} {\bibfnamefont
  {D.}~\bibnamefont {Bluvstein}}, \bibinfo {author} {\bibfnamefont
  {R.}~\bibnamefont {Samajdar}}, \bibinfo {author} {\bibfnamefont
  {H.}~\bibnamefont {Pichler}}, \bibinfo {author} {\bibfnamefont {W.~W.}\
  \bibnamefont {Ho}}, \ and\ \bibinfo {author} {\bibnamefont {et~al.}},\ }\href
  {\doibase 10.1038/s41586-021-03582-4} {\bibfield  {journal} {\bibinfo
  {journal} {Nature}\ }\textbf {\bibinfo {volume} {595}},\ \bibinfo {pages}
  {227–232} (\bibinfo {year} {2021})}\BibitemShut {NoStop}%
\bibitem [{\citenamefont {Scholl}\ \emph {et~al.}(2021)\citenamefont {Scholl},
  \citenamefont {Schuler}, \citenamefont {Williams}, \citenamefont
  {Eberharter}, \citenamefont {Barredo}, \citenamefont {Schymik}, \citenamefont
  {Lienhard}, \citenamefont {Henry}, \citenamefont {Lang}, \citenamefont
  {Lahaye},\ and\ \citenamefont {et~al.}}]{Scholl2021}%
  \BibitemOpen
  \bibfield  {author} {\bibinfo {author} {\bibfnamefont {P.}~\bibnamefont
  {Scholl}}, \bibinfo {author} {\bibfnamefont {M.}~\bibnamefont {Schuler}},
  \bibinfo {author} {\bibfnamefont {H.~J.}\ \bibnamefont {Williams}}, \bibinfo
  {author} {\bibfnamefont {A.~A.}\ \bibnamefont {Eberharter}}, \bibinfo
  {author} {\bibfnamefont {D.}~\bibnamefont {Barredo}}, \bibinfo {author}
  {\bibfnamefont {K.-N.}\ \bibnamefont {Schymik}}, \bibinfo {author}
  {\bibfnamefont {V.}~\bibnamefont {Lienhard}}, \bibinfo {author}
  {\bibfnamefont {L.-P.}\ \bibnamefont {Henry}}, \bibinfo {author}
  {\bibfnamefont {T.~C.}\ \bibnamefont {Lang}}, \bibinfo {author}
  {\bibfnamefont {T.}~\bibnamefont {Lahaye}}, \ and\ \bibinfo {author}
  {\bibnamefont {et~al.}},\ }\href {\doibase 10.1038/s41586-021-03585-1}
  {\bibfield  {journal} {\bibinfo  {journal} {Nature}\ }\textbf {\bibinfo
  {volume} {595}},\ \bibinfo {pages} {233–238} (\bibinfo {year}
  {2021})}\BibitemShut {NoStop}%
\bibitem [{\citenamefont {{Van Enk}}\ and\ \citenamefont
  {Beenakker}(2012)}]{VanEnk2012}%
  \BibitemOpen
  \bibfield  {author} {\bibinfo {author} {\bibfnamefont {S.~J.}\ \bibnamefont
  {{Van Enk}}}\ and\ \bibinfo {author} {\bibfnamefont {C.~W.}\ \bibnamefont
  {Beenakker}},\ }\href {\doibase 10.1103/PhysRevLett.108.110503} {\bibfield
  {journal} {\bibinfo  {journal} {Phys. Rev. Lett.}\ }\textbf {\bibinfo
  {volume} {108}},\ \bibinfo {pages} {110503} (\bibinfo {year}
  {2012})}\BibitemShut {NoStop}%
\bibitem [{\citenamefont {Tran}\ \emph {et~al.}(2016)\citenamefont {Tran},
  \citenamefont {Daki\ifmmode~\acute{c}\else \'{c}\fi{}}, \citenamefont
  {Laskowski},\ and\ \citenamefont {Paterek}}]{Cong2016}%
  \BibitemOpen
  \bibfield  {author} {\bibinfo {author} {\bibfnamefont {M.~C.}\ \bibnamefont
  {Tran}}, \bibinfo {author} {\bibfnamefont {B.}~\bibnamefont
  {Daki\ifmmode~\acute{c}\else \'{c}\fi{}}}, \bibinfo {author} {\bibfnamefont
  {W.}~\bibnamefont {Laskowski}}, \ and\ \bibinfo {author} {\bibfnamefont
  {T.}~\bibnamefont {Paterek}},\ }\href {\doibase 10.1103/PhysRevA.94.042302}
  {\bibfield  {journal} {\bibinfo  {journal} {Phys. Rev. A}\ }\textbf {\bibinfo
  {volume} {94}},\ \bibinfo {pages} {042302} (\bibinfo {year}
  {2016})}\BibitemShut {NoStop}%
\bibitem [{\citenamefont {Elben}\ \emph {et~al.}(2018)\citenamefont {Elben},
  \citenamefont {Vermersch}, \citenamefont {Dalmonte}, \citenamefont {Cirac},\
  and\ \citenamefont {Zoller}}]{Elben2018}%
  \BibitemOpen
  \bibfield  {author} {\bibinfo {author} {\bibfnamefont {A.}~\bibnamefont
  {Elben}}, \bibinfo {author} {\bibfnamefont {B.}~\bibnamefont {Vermersch}},
  \bibinfo {author} {\bibfnamefont {M.}~\bibnamefont {Dalmonte}}, \bibinfo
  {author} {\bibfnamefont {J.~I.}\ \bibnamefont {Cirac}}, \ and\ \bibinfo
  {author} {\bibfnamefont {P.}~\bibnamefont {Zoller}},\ }\href {\doibase
  10.1103/PhysRevLett.120.050406} {\bibfield  {journal} {\bibinfo  {journal}
  {Phys. Rev. Lett.}\ }\textbf {\bibinfo {volume} {120}},\ \bibinfo {pages}
  {050406} (\bibinfo {year} {2018})}\BibitemShut {NoStop}%
\bibitem [{\citenamefont {Vermersch}\ \emph {et~al.}(2018)\citenamefont
  {Vermersch}, \citenamefont {Elben}, \citenamefont {Dalmonte}, \citenamefont
  {Cirac},\ and\ \citenamefont {Zoller}}]{Vermersch2018}%
  \BibitemOpen
  \bibfield  {author} {\bibinfo {author} {\bibfnamefont {B.}~\bibnamefont
  {Vermersch}}, \bibinfo {author} {\bibfnamefont {A.}~\bibnamefont {Elben}},
  \bibinfo {author} {\bibfnamefont {M.}~\bibnamefont {Dalmonte}}, \bibinfo
  {author} {\bibfnamefont {J.~I.}\ \bibnamefont {Cirac}}, \ and\ \bibinfo
  {author} {\bibfnamefont {P.}~\bibnamefont {Zoller}},\ }\href {\doibase
  10.1103/PhysRevA.97.023604} {\bibfield  {journal} {\bibinfo  {journal} {Phys.
  Rev. A}\ }\textbf {\bibinfo {volume} {97}},\ \bibinfo {pages} {023604}
  (\bibinfo {year} {2018})}\BibitemShut {NoStop}%
\bibitem [{\citenamefont {Elben}\ \emph {et~al.}(2019)\citenamefont {Elben},
  \citenamefont {Vermersch}, \citenamefont {Roos},\ and\ \citenamefont
  {Zoller}}]{Elben2018a}%
  \BibitemOpen
  \bibfield  {author} {\bibinfo {author} {\bibfnamefont {A.}~\bibnamefont
  {Elben}}, \bibinfo {author} {\bibfnamefont {B.}~\bibnamefont {Vermersch}},
  \bibinfo {author} {\bibfnamefont {C.~F.}\ \bibnamefont {Roos}}, \ and\
  \bibinfo {author} {\bibfnamefont {P.}~\bibnamefont {Zoller}},\ }\href
  {\doibase 10.1103/PhysRevA.99.052323} {\bibfield  {journal} {\bibinfo
  {journal} {Phys. Rev. A}\ }\textbf {\bibinfo {volume} {99}},\ \bibinfo
  {pages} {052323} (\bibinfo {year} {2019})}\BibitemShut {NoStop}%
\bibitem [{\citenamefont {Knips}\ \emph {et~al.}(2020)\citenamefont {Knips},
  \citenamefont {Dziewior}, \citenamefont {K{\l}obus}, \citenamefont
  {Laskowski}, \citenamefont {Paterek}, \citenamefont {Shadbolt}, \citenamefont
  {Weinfurter},\ and\ \citenamefont {Meinecke}}]{Knips2019}%
  \BibitemOpen
  \bibfield  {author} {\bibinfo {author} {\bibfnamefont {L.}~\bibnamefont
  {Knips}}, \bibinfo {author} {\bibfnamefont {J.}~\bibnamefont {Dziewior}},
  \bibinfo {author} {\bibfnamefont {W.}~\bibnamefont {K{\l}obus}}, \bibinfo
  {author} {\bibfnamefont {W.}~\bibnamefont {Laskowski}}, \bibinfo {author}
  {\bibfnamefont {T.}~\bibnamefont {Paterek}}, \bibinfo {author} {\bibfnamefont
  {P.~J.}\ \bibnamefont {Shadbolt}}, \bibinfo {author} {\bibfnamefont
  {H.}~\bibnamefont {Weinfurter}}, \ and\ \bibinfo {author} {\bibfnamefont
  {J.~D.~A.}\ \bibnamefont {Meinecke}},\ }\href {\doibase
  10.1038/s41534-020-0281-5} {\bibfield  {journal} {\bibinfo  {journal} {npj
  Quant. Inf.}\ }\textbf {\bibinfo {volume} {6}},\ \bibinfo {pages} {51}
  (\bibinfo {year} {2020})}\BibitemShut {NoStop}%
\bibitem [{\citenamefont {Ketterer}\ \emph {et~al.}(2019)\citenamefont
  {Ketterer}, \citenamefont {Wyderka},\ and\ \citenamefont
  {G{\"{u}}hne}}]{Ketterer2019}%
  \BibitemOpen
  \bibfield  {author} {\bibinfo {author} {\bibfnamefont {A.}~\bibnamefont
  {Ketterer}}, \bibinfo {author} {\bibfnamefont {N.}~\bibnamefont {Wyderka}}, \
  and\ \bibinfo {author} {\bibfnamefont {O.}~\bibnamefont {G{\"{u}}hne}},\
  }\href {\doibase 10.1103/PhysRevLett.122.120505} {\bibfield  {journal}
  {\bibinfo  {journal} {Phys. Rev. Lett.}\ }\textbf {\bibinfo {volume} {122}},\
  \bibinfo {pages} {120505} (\bibinfo {year} {2019})}\BibitemShut {NoStop}%
\bibitem [{\citenamefont {Huang}\ \emph {et~al.}(2020)\citenamefont {Huang},
  \citenamefont {Kueng},\ and\ \citenamefont {Preskill}}]{Huang2020}%
  \BibitemOpen
  \bibfield  {author} {\bibinfo {author} {\bibfnamefont {H.-Y.}\ \bibnamefont
  {Huang}}, \bibinfo {author} {\bibfnamefont {R.}~\bibnamefont {Kueng}}, \ and\
  \bibinfo {author} {\bibfnamefont {J.}~\bibnamefont {Preskill}},\ }\href
  {\doibase 10.1038/s41567-020-0932-7} {\bibfield  {journal} {\bibinfo
  {journal} {Nat. Phys.}\ }\textbf {\bibinfo {volume} {16}},\ \bibinfo {pages}
  {1050} (\bibinfo {year} {2020})}\BibitemShut {NoStop}%
\bibitem [{\citenamefont {Elben}\ \emph
  {et~al.}(2020{\natexlab{a}})\citenamefont {Elben}, \citenamefont {Kueng},
  \citenamefont {Huang}, \citenamefont {van Bijnen}, \citenamefont {Kokail},
  \citenamefont {Dalmonte}, \citenamefont {Calabrese}, \citenamefont {Kraus},
  \citenamefont {Preskill}, \citenamefont {Zoller},\ and\ \citenamefont
  {Vermersch}}]{Elben2020}%
  \BibitemOpen
  \bibfield  {author} {\bibinfo {author} {\bibfnamefont {A.}~\bibnamefont
  {Elben}}, \bibinfo {author} {\bibfnamefont {R.}~\bibnamefont {Kueng}},
  \bibinfo {author} {\bibfnamefont {H.-Y.~R.}\ \bibnamefont {Huang}}, \bibinfo
  {author} {\bibfnamefont {R.}~\bibnamefont {van Bijnen}}, \bibinfo {author}
  {\bibfnamefont {C.}~\bibnamefont {Kokail}}, \bibinfo {author} {\bibfnamefont
  {M.}~\bibnamefont {Dalmonte}}, \bibinfo {author} {\bibfnamefont
  {P.}~\bibnamefont {Calabrese}}, \bibinfo {author} {\bibfnamefont
  {B.}~\bibnamefont {Kraus}}, \bibinfo {author} {\bibfnamefont
  {J.}~\bibnamefont {Preskill}}, \bibinfo {author} {\bibfnamefont
  {P.}~\bibnamefont {Zoller}}, \ and\ \bibinfo {author} {\bibfnamefont
  {B.}~\bibnamefont {Vermersch}},\ }\href {\doibase
  10.1103/PhysRevLett.125.200501} {\bibfield  {journal} {\bibinfo  {journal}
  {Phys. Rev. Lett.}\ }\textbf {\bibinfo {volume} {125}},\ \bibinfo {pages}
  {200501} (\bibinfo {year} {2020}{\natexlab{a}})}\BibitemShut {NoStop}%
\bibitem [{\citenamefont {Zhou}\ \emph {et~al.}(2020)\citenamefont {Zhou},
  \citenamefont {Zeng},\ and\ \citenamefont {Liu}}]{Zhou2020}%
  \BibitemOpen
  \bibfield  {author} {\bibinfo {author} {\bibfnamefont {Y.}~\bibnamefont
  {Zhou}}, \bibinfo {author} {\bibfnamefont {P.}~\bibnamefont {Zeng}}, \ and\
  \bibinfo {author} {\bibfnamefont {Z.}~\bibnamefont {Liu}},\ }\href {\doibase
  10.1103/PhysRevLett.125.200502} {\bibfield  {journal} {\bibinfo  {journal}
  {Phys. Rev. Lett.}\ }\textbf {\bibinfo {volume} {125}},\ \bibinfo {pages}
  {200502} (\bibinfo {year} {2020})}\BibitemShut {NoStop}%
\bibitem [{\citenamefont {Ketterer}\ \emph {et~al.}(2020)\citenamefont
  {Ketterer}, \citenamefont {Wyderka},\ and\ \citenamefont
  {G{\"{u}}hne}}]{Ketterer2020}%
  \BibitemOpen
  \bibfield  {author} {\bibinfo {author} {\bibfnamefont {A.}~\bibnamefont
  {Ketterer}}, \bibinfo {author} {\bibfnamefont {N.}~\bibnamefont {Wyderka}}, \
  and\ \bibinfo {author} {\bibfnamefont {O.}~\bibnamefont {G{\"{u}}hne}},\
  }\href {\doibase 10.22331/q-2020-09-16-325} {\bibfield  {journal} {\bibinfo
  {journal} {Quantum}\ }\textbf {\bibinfo {volume} {4}},\ \bibinfo {pages}
  {325} (\bibinfo {year} {2020})}\BibitemShut {NoStop}%
\bibitem [{\citenamefont {Ketterer}\ \emph {et~al.}()\citenamefont {Ketterer},
  \citenamefont {Imai}, \citenamefont {Wyderka},\ and\ \citenamefont
  {G{\"{u}}hne}}]{Ketterer2020a}%
  \BibitemOpen
  \bibfield  {author} {\bibinfo {author} {\bibfnamefont {A.}~\bibnamefont
  {Ketterer}}, \bibinfo {author} {\bibfnamefont {S.}~\bibnamefont {Imai}},
  \bibinfo {author} {\bibfnamefont {N.}~\bibnamefont {Wyderka}}, \ and\
  \bibinfo {author} {\bibfnamefont {O.}~\bibnamefont {G{\"{u}}hne}},\
  }\href@noop {} {\ }\Eprint {http://arxiv.org/abs/2012.12176}
  {arXiv:2012.12176} \BibitemShut {NoStop}%
\bibitem [{\citenamefont {Vitale}\ \emph {et~al.}()\citenamefont {Vitale},
  \citenamefont {Elben}, \citenamefont {Kueng}, \citenamefont {Neven},
  \citenamefont {Carrasco}, \citenamefont {Kraus}, \citenamefont {Zoller},
  \citenamefont {Calabrese}, \citenamefont {Vermersch},\ and\ \citenamefont
  {Dalmonte}}]{Vitale2021}%
  \BibitemOpen
  \bibfield  {author} {\bibinfo {author} {\bibfnamefont {V.}~\bibnamefont
  {Vitale}}, \bibinfo {author} {\bibfnamefont {A.}~\bibnamefont {Elben}},
  \bibinfo {author} {\bibfnamefont {R.}~\bibnamefont {Kueng}}, \bibinfo
  {author} {\bibfnamefont {A.}~\bibnamefont {Neven}}, \bibinfo {author}
  {\bibfnamefont {J.}~\bibnamefont {Carrasco}}, \bibinfo {author}
  {\bibfnamefont {B.}~\bibnamefont {Kraus}}, \bibinfo {author} {\bibfnamefont
  {P.}~\bibnamefont {Zoller}}, \bibinfo {author} {\bibfnamefont
  {P.}~\bibnamefont {Calabrese}}, \bibinfo {author} {\bibfnamefont
  {B.}~\bibnamefont {Vermersch}}, \ and\ \bibinfo {author} {\bibfnamefont
  {M.}~\bibnamefont {Dalmonte}},\ }\href {http://arxiv.org/abs/2101.07814} {\
  }\Eprint {http://arxiv.org/abs/2101.07814} {arXiv:2101.07814} \BibitemShut
  {NoStop}%
\bibitem [{\citenamefont {Imai}\ \emph {et~al.}(2021)\citenamefont {Imai},
  \citenamefont {Wyderka}, \citenamefont {Ketterer},\ and\ \citenamefont
  {G\"uhne}}]{Satoya2021}%
  \BibitemOpen
  \bibfield  {author} {\bibinfo {author} {\bibfnamefont {S.}~\bibnamefont
  {Imai}}, \bibinfo {author} {\bibfnamefont {N.}~\bibnamefont {Wyderka}},
  \bibinfo {author} {\bibfnamefont {A.}~\bibnamefont {Ketterer}}, \ and\
  \bibinfo {author} {\bibfnamefont {O.}~\bibnamefont {G\"uhne}},\ }\href
  {\doibase 10.1103/PhysRevLett.126.150501} {\bibfield  {journal} {\bibinfo
  {journal} {Phys. Rev. Lett.}\ }\textbf {\bibinfo {volume} {126}},\ \bibinfo
  {pages} {150501} (\bibinfo {year} {2021})}\BibitemShut {NoStop}%
\bibitem [{\citenamefont {Rath}\ \emph {et~al.}()\citenamefont {Rath},
  \citenamefont {Branciard}, \citenamefont {Minguzzi},\ and\ \citenamefont
  {Vermersch}}]{Rath2021b}%
  \BibitemOpen
  \bibfield  {author} {\bibinfo {author} {\bibfnamefont {A.}~\bibnamefont
  {Rath}}, \bibinfo {author} {\bibfnamefont {C.}~\bibnamefont {Branciard}},
  \bibinfo {author} {\bibfnamefont {A.}~\bibnamefont {Minguzzi}}, \ and\
  \bibinfo {author} {\bibfnamefont {B.}~\bibnamefont {Vermersch}},\ }\href
  {http://arxiv.org/abs/2105.13164} {\ }\Eprint
  {http://arxiv.org/abs/2105.13164} {arXiv:2105.13164} \BibitemShut {NoStop}%
\bibitem [{\citenamefont {Vermersch}\ \emph {et~al.}(2019)\citenamefont
  {Vermersch}, \citenamefont {Elben}, \citenamefont {Sieberer}, \citenamefont
  {Yao},\ and\ \citenamefont {Zoller}}]{Vermersch2019}%
  \BibitemOpen
  \bibfield  {author} {\bibinfo {author} {\bibfnamefont {B.}~\bibnamefont
  {Vermersch}}, \bibinfo {author} {\bibfnamefont {A.}~\bibnamefont {Elben}},
  \bibinfo {author} {\bibfnamefont {L.~M.}\ \bibnamefont {Sieberer}}, \bibinfo
  {author} {\bibfnamefont {N.~Y.}\ \bibnamefont {Yao}}, \ and\ \bibinfo
  {author} {\bibfnamefont {P.}~\bibnamefont {Zoller}},\ }\href {\doibase
  10.1103/PhysRevX.9.021061} {\bibfield  {journal} {\bibinfo  {journal} {Phys.
  Rev. X}\ }\textbf {\bibinfo {volume} {9}},\ \bibinfo {pages} {021061}
  (\bibinfo {year} {2019})}\BibitemShut {NoStop}%
\bibitem [{\citenamefont {Qi}\ \emph {et~al.}()\citenamefont {Qi},
  \citenamefont {Davis}, \citenamefont {Periwal},\ and\ \citenamefont
  {Schleier-Smith}}]{Qi2019}%
  \BibitemOpen
  \bibfield  {author} {\bibinfo {author} {\bibfnamefont {X.-L.}\ \bibnamefont
  {Qi}}, \bibinfo {author} {\bibfnamefont {E.~J.}\ \bibnamefont {Davis}},
  \bibinfo {author} {\bibfnamefont {A.}~\bibnamefont {Periwal}}, \ and\
  \bibinfo {author} {\bibfnamefont {M.}~\bibnamefont {Schleier-Smith}},\ }\href
  {http://arxiv.org/abs/1906.00524} {\ }\Eprint
  {http://arxiv.org/abs/1906.00524} {arXiv:1906.00524} \BibitemShut {NoStop}%
\bibitem [{\citenamefont {Garcia}\ \emph {et~al.}(2021)\citenamefont {Garcia},
  \citenamefont {Zhou},\ and\ \citenamefont {Jaffe}}]{Garcia2021}%
  \BibitemOpen
  \bibfield  {author} {\bibinfo {author} {\bibfnamefont {R.~J.}\ \bibnamefont
  {Garcia}}, \bibinfo {author} {\bibfnamefont {Y.}~\bibnamefont {Zhou}}, \ and\
  \bibinfo {author} {\bibfnamefont {A.}~\bibnamefont {Jaffe}},\ }\href
  {\doibase 10.1103/PhysRevResearch.3.033155} {\bibfield  {journal} {\bibinfo
  {journal} {Phys. Rev. Research}\ }\textbf {\bibinfo {volume} {3}},\ \bibinfo
  {pages} {033155} (\bibinfo {year} {2021})}\BibitemShut {NoStop}%
\bibitem [{\citenamefont {Elben}\ \emph
  {et~al.}(2020{\natexlab{b}})\citenamefont {Elben}, \citenamefont {Yu},
  \citenamefont {Zhu}, \citenamefont {Hafezi}, \citenamefont {Pollmann},
  \citenamefont {Zoller},\ and\ \citenamefont {Vermersch}}]{Elben2019}%
  \BibitemOpen
  \bibfield  {author} {\bibinfo {author} {\bibfnamefont {A.}~\bibnamefont
  {Elben}}, \bibinfo {author} {\bibfnamefont {J.}~\bibnamefont {Yu}}, \bibinfo
  {author} {\bibfnamefont {G.}~\bibnamefont {Zhu}}, \bibinfo {author}
  {\bibfnamefont {M.}~\bibnamefont {Hafezi}}, \bibinfo {author} {\bibfnamefont
  {F.}~\bibnamefont {Pollmann}}, \bibinfo {author} {\bibfnamefont
  {P.}~\bibnamefont {Zoller}}, \ and\ \bibinfo {author} {\bibfnamefont
  {B.}~\bibnamefont {Vermersch}},\ }\href {\doibase 10.1126/sciadv.aaz3666}
  {\bibfield  {journal} {\bibinfo  {journal} {Sci. Adv.}\ }\textbf {\bibinfo
  {volume} {6}},\ \bibinfo {pages} {eaaz3666} (\bibinfo {year}
  {2020}{\natexlab{b}})}\BibitemShut {NoStop}%
\bibitem [{\citenamefont {Cian}\ \emph {et~al.}(2021)\citenamefont {Cian},
  \citenamefont {Dehghani}, \citenamefont {Elben}, \citenamefont {Vermersch},
  \citenamefont {Zhu}, \citenamefont {Barkeshli}, \citenamefont {Zoller},\ and\
  \citenamefont {Hafezi}}]{Cian2020}%
  \BibitemOpen
  \bibfield  {author} {\bibinfo {author} {\bibfnamefont {Z.-P.}\ \bibnamefont
  {Cian}}, \bibinfo {author} {\bibfnamefont {H.}~\bibnamefont {Dehghani}},
  \bibinfo {author} {\bibfnamefont {A.}~\bibnamefont {Elben}}, \bibinfo
  {author} {\bibfnamefont {B.}~\bibnamefont {Vermersch}}, \bibinfo {author}
  {\bibfnamefont {G.}~\bibnamefont {Zhu}}, \bibinfo {author} {\bibfnamefont
  {M.}~\bibnamefont {Barkeshli}}, \bibinfo {author} {\bibfnamefont
  {P.}~\bibnamefont {Zoller}}, \ and\ \bibinfo {author} {\bibfnamefont
  {M.}~\bibnamefont {Hafezi}},\ }\href {\doibase
  10.1103/PhysRevLett.126.050501} {\bibfield  {journal} {\bibinfo  {journal}
  {Phys. Rev. Lett.}\ }\textbf {\bibinfo {volume} {126}},\ \bibinfo {pages}
  {050501} (\bibinfo {year} {2021})}\BibitemShut {NoStop}%
\bibitem [{\citenamefont {Elben}\ \emph
  {et~al.}(2020{\natexlab{c}})\citenamefont {Elben}, \citenamefont {Vermersch},
  \citenamefont {van Bijnen}, \citenamefont {Kokail}, \citenamefont {Brydges},
  \citenamefont {Maier}, \citenamefont {Joshi}, \citenamefont {Blatt},
  \citenamefont {Roos},\ and\ \citenamefont {Zoller}}]{Elben2020a}%
  \BibitemOpen
  \bibfield  {author} {\bibinfo {author} {\bibfnamefont {A.}~\bibnamefont
  {Elben}}, \bibinfo {author} {\bibfnamefont {B.}~\bibnamefont {Vermersch}},
  \bibinfo {author} {\bibfnamefont {R.}~\bibnamefont {van Bijnen}}, \bibinfo
  {author} {\bibfnamefont {C.}~\bibnamefont {Kokail}}, \bibinfo {author}
  {\bibfnamefont {T.}~\bibnamefont {Brydges}}, \bibinfo {author} {\bibfnamefont
  {C.}~\bibnamefont {Maier}}, \bibinfo {author} {\bibfnamefont {M.~K.}\
  \bibnamefont {Joshi}}, \bibinfo {author} {\bibfnamefont {R.}~\bibnamefont
  {Blatt}}, \bibinfo {author} {\bibfnamefont {C.~F.}\ \bibnamefont {Roos}}, \
  and\ \bibinfo {author} {\bibfnamefont {P.}~\bibnamefont {Zoller}},\ }\href
  {\doibase 10.1103/PhysRevLett.124.010504} {\bibfield  {journal} {\bibinfo
  {journal} {Phys. Rev. Lett.}\ }\textbf {\bibinfo {volume} {124}},\ \bibinfo
  {pages} {010504} (\bibinfo {year} {2020}{\natexlab{c}})}\BibitemShut
  {NoStop}%
\bibitem [{\citenamefont {Gross}\ \emph {et~al.}(2010)\citenamefont {Gross},
  \citenamefont {Liu}, \citenamefont {Flammia}, \citenamefont {Becker},\ and\
  \citenamefont {Eisert}}]{Gross2010}%
  \BibitemOpen
  \bibfield  {author} {\bibinfo {author} {\bibfnamefont {D.}~\bibnamefont
  {Gross}}, \bibinfo {author} {\bibfnamefont {Y.~K.}\ \bibnamefont {Liu}},
  \bibinfo {author} {\bibfnamefont {S.~T.}\ \bibnamefont {Flammia}}, \bibinfo
  {author} {\bibfnamefont {S.}~\bibnamefont {Becker}}, \ and\ \bibinfo {author}
  {\bibfnamefont {J.}~\bibnamefont {Eisert}},\ }\href {\doibase
  10.1103/PhysRevLett.105.150401} {\bibfield  {journal} {\bibinfo  {journal}
  {Phys. Rev. Lett.}\ }\textbf {\bibinfo {volume} {105}},\ \bibinfo {pages}
  {150401} (\bibinfo {year} {2010})}\BibitemShut {NoStop}%
\bibitem [{\citenamefont {Brydges}\ \emph {et~al.}(2019)\citenamefont
  {Brydges}, \citenamefont {Elben}, \citenamefont {Jurcevic}, \citenamefont
  {Vermersch}, \citenamefont {Maier}, \citenamefont {Lanyon}, \citenamefont
  {Zoller}, \citenamefont {Blatt},\ and\ \citenamefont {Roos}}]{Brydges2019}%
  \BibitemOpen
  \bibfield  {author} {\bibinfo {author} {\bibfnamefont {T.}~\bibnamefont
  {Brydges}}, \bibinfo {author} {\bibfnamefont {A.}~\bibnamefont {Elben}},
  \bibinfo {author} {\bibfnamefont {P.}~\bibnamefont {Jurcevic}}, \bibinfo
  {author} {\bibfnamefont {B.}~\bibnamefont {Vermersch}}, \bibinfo {author}
  {\bibfnamefont {C.}~\bibnamefont {Maier}}, \bibinfo {author} {\bibfnamefont
  {B.~P.}\ \bibnamefont {Lanyon}}, \bibinfo {author} {\bibfnamefont
  {P.}~\bibnamefont {Zoller}}, \bibinfo {author} {\bibfnamefont
  {R.}~\bibnamefont {Blatt}}, \ and\ \bibinfo {author} {\bibfnamefont {C.~F.}\
  \bibnamefont {Roos}},\ }\href {\doibase 10.1126/science.aau4963} {\bibfield
  {journal} {\bibinfo  {journal} {Science}\ }\textbf {\bibinfo {volume}
  {364}},\ \bibinfo {pages} {260} (\bibinfo {year} {2019})}\BibitemShut
  {NoStop}%
\bibitem [{\citenamefont {Joshi}\ \emph {et~al.}(2020)\citenamefont {Joshi},
  \citenamefont {Elben}, \citenamefont {Vermersch}, \citenamefont {Brydges},
  \citenamefont {Maier}, \citenamefont {Zoller}, \citenamefont {Blatt},\ and\
  \citenamefont {Roos}}]{Joshi2020}%
  \BibitemOpen
  \bibfield  {author} {\bibinfo {author} {\bibfnamefont {M.~K.}\ \bibnamefont
  {Joshi}}, \bibinfo {author} {\bibfnamefont {A.}~\bibnamefont {Elben}},
  \bibinfo {author} {\bibfnamefont {B.}~\bibnamefont {Vermersch}}, \bibinfo
  {author} {\bibfnamefont {T.}~\bibnamefont {Brydges}}, \bibinfo {author}
  {\bibfnamefont {C.}~\bibnamefont {Maier}}, \bibinfo {author} {\bibfnamefont
  {P.}~\bibnamefont {Zoller}}, \bibinfo {author} {\bibfnamefont
  {R.}~\bibnamefont {Blatt}}, \ and\ \bibinfo {author} {\bibfnamefont {C.~F.}\
  \bibnamefont {Roos}},\ }\href {\doibase 10.1103/PhysRevLett.124.240505}
  {\bibfield  {journal} {\bibinfo  {journal} {Phys. Rev. Lett.}\ }\textbf
  {\bibinfo {volume} {124}},\ \bibinfo {pages} {240505} (\bibinfo {year}
  {2020})}\BibitemShut {NoStop}%
\bibitem [{\citenamefont {Horodecki}\ and\ \citenamefont
  {Horodecki}(1996)}]{Horodecki1996}%
  \BibitemOpen
  \bibfield  {author} {\bibinfo {author} {\bibfnamefont {R.}~\bibnamefont
  {Horodecki}}\ and\ \bibinfo {author} {\bibfnamefont {M.}~\bibnamefont
  {Horodecki}},\ }\href {\doibase 10.1103/PhysRevA.54.1838} {\bibfield
  {journal} {\bibinfo  {journal} {Phys. Rev. A}\ }\textbf {\bibinfo {volume}
  {54}},\ \bibinfo {pages} {1838} (\bibinfo {year} {1996})}\BibitemShut
  {NoStop}%
\bibitem [{\citenamefont {Eisert}\ \emph {et~al.}(2010)\citenamefont {Eisert},
  \citenamefont {Cramer},\ and\ \citenamefont {Plenio}}]{Eisert2010}%
  \BibitemOpen
  \bibfield  {author} {\bibinfo {author} {\bibfnamefont {J.}~\bibnamefont
  {Eisert}}, \bibinfo {author} {\bibfnamefont {M.}~\bibnamefont {Cramer}}, \
  and\ \bibinfo {author} {\bibfnamefont {M.~B.}\ \bibnamefont {Plenio}},\
  }\href {\doibase 10.1103/RevModPhys.82.277} {\bibfield  {journal} {\bibinfo
  {journal} {Rev. Mod. Phys.}\ }\textbf {\bibinfo {volume} {82}},\ \bibinfo
  {pages} {277} (\bibinfo {year} {2010})}\BibitemShut {NoStop}%
\bibitem [{\citenamefont {Satzinger}\ \emph {et~al.}(2021)\citenamefont
  {Satzinger}, \citenamefont {Liu}, \citenamefont {Smith}, \citenamefont
  {Knapp}, \citenamefont {Newman}, \citenamefont {Jones}, \citenamefont {Chen},
  \citenamefont {Quintana}, \citenamefont {Mi}, \citenamefont {Dunsworth},
  \citenamefont {Gidney}, \citenamefont {Aleiner}, \citenamefont {Arute},
  \citenamefont {Arya}, \citenamefont {Atalaya}, \citenamefont {Babbush},
  \citenamefont {Bardin}, \citenamefont {Barends}, \citenamefont {Basso},
  \citenamefont {Bengtsson}, \citenamefont {Bilmes}, \citenamefont {Broughton},
  \citenamefont {Buckley}, \citenamefont {Buell}, \citenamefont {Burkett},
  \citenamefont {Bushnell}, \citenamefont {Chiaro}, \citenamefont {Collins},
  \citenamefont {Courtney}, \citenamefont {Demura}, \citenamefont {Derk},
  \citenamefont {Eppens}, \citenamefont {Erickson}, \citenamefont {Farhi},
  \citenamefont {Foaro}, \citenamefont {Fowler}, \citenamefont {Foxen},
  \citenamefont {Giustina}, \citenamefont {Greene}, \citenamefont {Gross},
  \citenamefont {Harrigan}, \citenamefont {Harrington}, \citenamefont {Hilton},
  \citenamefont {Hong}, \citenamefont {Huang}, \citenamefont {Huggins},
  \citenamefont {Ioffe}, \citenamefont {Isakov}, \citenamefont {Jeffrey},
  \citenamefont {Jiang}, \citenamefont {Kafri}, \citenamefont {Kechedzhi},
  \citenamefont {Khattar}, \citenamefont {Kim}, \citenamefont {Klimov},
  \citenamefont {Korotkov}, \citenamefont {Kostritsa}, \citenamefont
  {Landhuis}, \citenamefont {Laptev}, \citenamefont {Locharla}, \citenamefont
  {Lucero}, \citenamefont {Martin}, \citenamefont {McClean}, \citenamefont
  {McEwen}, \citenamefont {Miao}, \citenamefont {Mohseni}, \citenamefont
  {Montazeri}, \citenamefont {Mruczkiewicz}, \citenamefont {Mutus},
  \citenamefont {Naaman}, \citenamefont {Neeley}, \citenamefont {Neill},
  \citenamefont {Niu}, \citenamefont {O'Brien}, \citenamefont {Opremcak},
  \citenamefont {Pató}, \citenamefont {Petukhov}, \citenamefont {Rubin},
  \citenamefont {Sank}, \citenamefont {Shvarts}, \citenamefont {Strain},
  \citenamefont {Szalay}, \citenamefont {Villalonga}, \citenamefont {White},
  \citenamefont {Yao}, \citenamefont {Yeh}, \citenamefont {Yoo}, \citenamefont
  {Zalcman}, \citenamefont {Neven}, \citenamefont {Boixo}, \citenamefont
  {Megrant}, \citenamefont {Chen}, \citenamefont {Kelly}, \citenamefont
  {Smelyanskiy}, \citenamefont {Kitaev}, \citenamefont {Knap}, \citenamefont
  {Pollmann},\ and\ \citenamefont {Roushan}}]{satzinger2021realizing}%
  \BibitemOpen
  \bibfield  {author} {\bibinfo {author} {\bibfnamefont {K.~J.}\ \bibnamefont
  {Satzinger}}, \bibinfo {author} {\bibfnamefont {Y.}~\bibnamefont {Liu}},
  \bibinfo {author} {\bibfnamefont {A.}~\bibnamefont {Smith}}, \bibinfo
  {author} {\bibfnamefont {C.}~\bibnamefont {Knapp}}, \bibinfo {author}
  {\bibfnamefont {M.}~\bibnamefont {Newman}}, \bibinfo {author} {\bibfnamefont
  {C.}~\bibnamefont {Jones}}, \bibinfo {author} {\bibfnamefont
  {Z.}~\bibnamefont {Chen}}, \bibinfo {author} {\bibfnamefont {C.}~\bibnamefont
  {Quintana}}, \bibinfo {author} {\bibfnamefont {X.}~\bibnamefont {Mi}},
  \bibinfo {author} {\bibfnamefont {A.}~\bibnamefont {Dunsworth}}, \bibinfo
  {author} {\bibfnamefont {C.}~\bibnamefont {Gidney}}, \bibinfo {author}
  {\bibfnamefont {I.}~\bibnamefont {Aleiner}}, \bibinfo {author} {\bibfnamefont
  {F.}~\bibnamefont {Arute}}, \bibinfo {author} {\bibfnamefont
  {K.}~\bibnamefont {Arya}}, \bibinfo {author} {\bibfnamefont {J.}~\bibnamefont
  {Atalaya}}, \bibinfo {author} {\bibfnamefont {R.}~\bibnamefont {Babbush}},
  \bibinfo {author} {\bibfnamefont {J.~C.}\ \bibnamefont {Bardin}}, \bibinfo
  {author} {\bibfnamefont {R.}~\bibnamefont {Barends}}, \bibinfo {author}
  {\bibfnamefont {J.}~\bibnamefont {Basso}}, \bibinfo {author} {\bibfnamefont
  {A.}~\bibnamefont {Bengtsson}}, \bibinfo {author} {\bibfnamefont
  {A.}~\bibnamefont {Bilmes}}, \bibinfo {author} {\bibfnamefont
  {M.}~\bibnamefont {Broughton}}, \bibinfo {author} {\bibfnamefont {B.~B.}\
  \bibnamefont {Buckley}}, \bibinfo {author} {\bibfnamefont {D.~A.}\
  \bibnamefont {Buell}}, \bibinfo {author} {\bibfnamefont {B.}~\bibnamefont
  {Burkett}}, \bibinfo {author} {\bibfnamefont {N.}~\bibnamefont {Bushnell}},
  \bibinfo {author} {\bibfnamefont {B.}~\bibnamefont {Chiaro}}, \bibinfo
  {author} {\bibfnamefont {R.}~\bibnamefont {Collins}}, \bibinfo {author}
  {\bibfnamefont {W.}~\bibnamefont {Courtney}}, \bibinfo {author}
  {\bibfnamefont {S.}~\bibnamefont {Demura}}, \bibinfo {author} {\bibfnamefont
  {A.~R.}\ \bibnamefont {Derk}}, \bibinfo {author} {\bibfnamefont
  {D.}~\bibnamefont {Eppens}}, \bibinfo {author} {\bibfnamefont
  {C.}~\bibnamefont {Erickson}}, \bibinfo {author} {\bibfnamefont
  {E.}~\bibnamefont {Farhi}}, \bibinfo {author} {\bibfnamefont
  {L.}~\bibnamefont {Foaro}}, \bibinfo {author} {\bibfnamefont {A.~G.}\
  \bibnamefont {Fowler}}, \bibinfo {author} {\bibfnamefont {B.}~\bibnamefont
  {Foxen}}, \bibinfo {author} {\bibfnamefont {M.}~\bibnamefont {Giustina}},
  \bibinfo {author} {\bibfnamefont {A.}~\bibnamefont {Greene}}, \bibinfo
  {author} {\bibfnamefont {J.~A.}\ \bibnamefont {Gross}}, \bibinfo {author}
  {\bibfnamefont {M.~P.}\ \bibnamefont {Harrigan}}, \bibinfo {author}
  {\bibfnamefont {S.~D.}\ \bibnamefont {Harrington}}, \bibinfo {author}
  {\bibfnamefont {J.}~\bibnamefont {Hilton}}, \bibinfo {author} {\bibfnamefont
  {S.}~\bibnamefont {Hong}}, \bibinfo {author} {\bibfnamefont {T.}~\bibnamefont
  {Huang}}, \bibinfo {author} {\bibfnamefont {W.~J.}\ \bibnamefont {Huggins}},
  \bibinfo {author} {\bibfnamefont {L.~B.}\ \bibnamefont {Ioffe}}, \bibinfo
  {author} {\bibfnamefont {S.~V.}\ \bibnamefont {Isakov}}, \bibinfo {author}
  {\bibfnamefont {E.}~\bibnamefont {Jeffrey}}, \bibinfo {author} {\bibfnamefont
  {Z.}~\bibnamefont {Jiang}}, \bibinfo {author} {\bibfnamefont
  {D.}~\bibnamefont {Kafri}}, \bibinfo {author} {\bibfnamefont
  {K.}~\bibnamefont {Kechedzhi}}, \bibinfo {author} {\bibfnamefont
  {T.}~\bibnamefont {Khattar}}, \bibinfo {author} {\bibfnamefont
  {S.}~\bibnamefont {Kim}}, \bibinfo {author} {\bibfnamefont {P.~V.}\
  \bibnamefont {Klimov}}, \bibinfo {author} {\bibfnamefont {A.~N.}\
  \bibnamefont {Korotkov}}, \bibinfo {author} {\bibfnamefont {F.}~\bibnamefont
  {Kostritsa}}, \bibinfo {author} {\bibfnamefont {D.}~\bibnamefont {Landhuis}},
  \bibinfo {author} {\bibfnamefont {P.}~\bibnamefont {Laptev}}, \bibinfo
  {author} {\bibfnamefont {A.}~\bibnamefont {Locharla}}, \bibinfo {author}
  {\bibfnamefont {E.}~\bibnamefont {Lucero}}, \bibinfo {author} {\bibfnamefont
  {O.}~\bibnamefont {Martin}}, \bibinfo {author} {\bibfnamefont {J.~R.}\
  \bibnamefont {McClean}}, \bibinfo {author} {\bibfnamefont {M.}~\bibnamefont
  {McEwen}}, \bibinfo {author} {\bibfnamefont {K.~C.}\ \bibnamefont {Miao}},
  \bibinfo {author} {\bibfnamefont {M.}~\bibnamefont {Mohseni}}, \bibinfo
  {author} {\bibfnamefont {S.}~\bibnamefont {Montazeri}}, \bibinfo {author}
  {\bibfnamefont {W.}~\bibnamefont {Mruczkiewicz}}, \bibinfo {author}
  {\bibfnamefont {J.}~\bibnamefont {Mutus}}, \bibinfo {author} {\bibfnamefont
  {O.}~\bibnamefont {Naaman}}, \bibinfo {author} {\bibfnamefont
  {M.}~\bibnamefont {Neeley}}, \bibinfo {author} {\bibfnamefont
  {C.}~\bibnamefont {Neill}}, \bibinfo {author} {\bibfnamefont {M.~Y.}\
  \bibnamefont {Niu}}, \bibinfo {author} {\bibfnamefont {T.~E.}\ \bibnamefont
  {O'Brien}}, \bibinfo {author} {\bibfnamefont {A.}~\bibnamefont {Opremcak}},
  \bibinfo {author} {\bibfnamefont {B.}~\bibnamefont {Pató}}, \bibinfo
  {author} {\bibfnamefont {A.}~\bibnamefont {Petukhov}}, \bibinfo {author}
  {\bibfnamefont {N.~C.}\ \bibnamefont {Rubin}}, \bibinfo {author}
  {\bibfnamefont {D.}~\bibnamefont {Sank}}, \bibinfo {author} {\bibfnamefont
  {V.}~\bibnamefont {Shvarts}}, \bibinfo {author} {\bibfnamefont
  {D.}~\bibnamefont {Strain}}, \bibinfo {author} {\bibfnamefont
  {M.}~\bibnamefont {Szalay}}, \bibinfo {author} {\bibfnamefont
  {B.}~\bibnamefont {Villalonga}}, \bibinfo {author} {\bibfnamefont {T.~C.}\
  \bibnamefont {White}}, \bibinfo {author} {\bibfnamefont {Z.}~\bibnamefont
  {Yao}}, \bibinfo {author} {\bibfnamefont {P.}~\bibnamefont {Yeh}}, \bibinfo
  {author} {\bibfnamefont {J.}~\bibnamefont {Yoo}}, \bibinfo {author}
  {\bibfnamefont {A.}~\bibnamefont {Zalcman}}, \bibinfo {author} {\bibfnamefont
  {H.}~\bibnamefont {Neven}}, \bibinfo {author} {\bibfnamefont
  {S.}~\bibnamefont {Boixo}}, \bibinfo {author} {\bibfnamefont
  {A.}~\bibnamefont {Megrant}}, \bibinfo {author} {\bibfnamefont
  {Y.}~\bibnamefont {Chen}}, \bibinfo {author} {\bibfnamefont {J.}~\bibnamefont
  {Kelly}}, \bibinfo {author} {\bibfnamefont {V.}~\bibnamefont {Smelyanskiy}},
  \bibinfo {author} {\bibfnamefont {A.}~\bibnamefont {Kitaev}}, \bibinfo
  {author} {\bibfnamefont {M.}~\bibnamefont {Knap}}, \bibinfo {author}
  {\bibfnamefont {F.}~\bibnamefont {Pollmann}}, \ and\ \bibinfo {author}
  {\bibfnamefont {P.}~\bibnamefont {Roushan}},\ }\href@noop {} {\enquote
  {\bibinfo {title} {Realizing topologically ordered states on a quantum
  processor},}\ } (\bibinfo {year} {2021}),\ \Eprint
  {http://arxiv.org/abs/2104.01180} {arXiv:2104.01180 [quant-ph]} \BibitemShut
  {NoStop}%
\bibitem [{\citenamefont {Zhu}\ \emph {et~al.}(2021)\citenamefont {Zhu},
  \citenamefont {Cian}, \citenamefont {Noel}, \citenamefont {Risinger},
  \citenamefont {Biswas}, \citenamefont {Egan}, \citenamefont {Zhu},
  \citenamefont {Green}, \citenamefont {Alderete}, \citenamefont {Nguyen},
  \citenamefont {Wang}, \citenamefont {Maksymov}, \citenamefont {Nam},
  \citenamefont {Cetina}, \citenamefont {Linke}, \citenamefont {Hafezi},\ and\
  \citenamefont {Monroe}}]{zhu2021crossplatform}%
  \BibitemOpen
  \bibfield  {author} {\bibinfo {author} {\bibfnamefont {D.}~\bibnamefont
  {Zhu}}, \bibinfo {author} {\bibfnamefont {Z.-P.}\ \bibnamefont {Cian}},
  \bibinfo {author} {\bibfnamefont {C.}~\bibnamefont {Noel}}, \bibinfo {author}
  {\bibfnamefont {A.}~\bibnamefont {Risinger}}, \bibinfo {author}
  {\bibfnamefont {D.}~\bibnamefont {Biswas}}, \bibinfo {author} {\bibfnamefont
  {L.}~\bibnamefont {Egan}}, \bibinfo {author} {\bibfnamefont {Y.}~\bibnamefont
  {Zhu}}, \bibinfo {author} {\bibfnamefont {A.~M.}\ \bibnamefont {Green}},
  \bibinfo {author} {\bibfnamefont {C.~H.}\ \bibnamefont {Alderete}}, \bibinfo
  {author} {\bibfnamefont {N.~H.}\ \bibnamefont {Nguyen}}, \bibinfo {author}
  {\bibfnamefont {Q.}~\bibnamefont {Wang}}, \bibinfo {author} {\bibfnamefont
  {A.}~\bibnamefont {Maksymov}}, \bibinfo {author} {\bibfnamefont
  {Y.}~\bibnamefont {Nam}}, \bibinfo {author} {\bibfnamefont {M.}~\bibnamefont
  {Cetina}}, \bibinfo {author} {\bibfnamefont {N.~M.}\ \bibnamefont {Linke}},
  \bibinfo {author} {\bibfnamefont {M.}~\bibnamefont {Hafezi}}, \ and\ \bibinfo
  {author} {\bibfnamefont {C.}~\bibnamefont {Monroe}},\ }\href@noop {}
  {\enquote {\bibinfo {title} {Cross-platform comparison of arbitrary quantum
  computations},}\ } (\bibinfo {year} {2021}),\ \Eprint
  {http://arxiv.org/abs/2107.11387} {arXiv:2107.11387 [quant-ph]} \BibitemShut
  {NoStop}%
\bibitem [{SM()}]{SM}%
  \BibitemOpen
  \href@noop {} {}\bibinfo {note} {{See Supplemental Material, which includes
  Refs. 6, 8, 13, 25.}}\BibitemShut {Stop}%
\bibitem [{\citenamefont {Diaconis}\ and\ \citenamefont
  {Forrester}()}]{Diaconis2015}%
  \BibitemOpen
  \bibfield  {author} {\bibinfo {author} {\bibfnamefont {P.}~\bibnamefont
  {Diaconis}}\ and\ \bibinfo {author} {\bibfnamefont {P.~J.}\ \bibnamefont
  {Forrester}},\ }\href {http://arxiv.org/abs/1512.09229} {\ }\Eprint
  {http://arxiv.org/abs/1512.09229} {arXiv:1512.09229} \BibitemShut {NoStop}%
\bibitem [{\citenamefont {Planitz}\ \emph {et~al.}(1987)\citenamefont
  {Planitz}, \citenamefont {Press}, \citenamefont {Flannery}, \citenamefont
  {Teukolsky},\ and\ \citenamefont {Vetterling}}]{NRecipes2007}%
  \BibitemOpen
  \bibfield  {author} {\bibinfo {author} {\bibfnamefont {M.}~\bibnamefont
  {Planitz}}, \bibinfo {author} {\bibfnamefont {W.~H.}\ \bibnamefont {Press}},
  \bibinfo {author} {\bibfnamefont {B.~P.}\ \bibnamefont {Flannery}}, \bibinfo
  {author} {\bibfnamefont {S.~A.}\ \bibnamefont {Teukolsky}}, \ and\ \bibinfo
  {author} {\bibfnamefont {W.~T.}\ \bibnamefont {Vetterling}},\ }\href
  {\doibase 10.2307/3616786} {\emph {\bibinfo {title} {{Numerical Recipes: The
  Art of Scientific Computing}}}},\ \bibinfo {edition} {3rd}\ ed.,\
  Vol.~\bibinfo {volume} {71}\ (\bibinfo  {publisher} {Cambridge University
  Press},\ \bibinfo {address} {New York, NY, USA},\ \bibinfo {year}
  {1987})\BibitemShut {NoStop}%
\bibitem [{Note1()}]{Note1}%
  \BibitemOpen
  \bibinfo {note} {In the examples below, we have $X_\protect \mathrm
  {IS}\approx X(u)>0$.}\BibitemShut {Stop}%
\bibitem [{\citenamefont {Schollw{\"{o}}ck}(2011)}]{Schollwock2011}%
  \BibitemOpen
  \bibfield  {author} {\bibinfo {author} {\bibfnamefont {U.}~\bibnamefont
  {Schollw{\"{o}}ck}},\ }\href {\doibase 10.1016/j.aop.2010.09.012} {\bibfield
  {journal} {\bibinfo  {journal} {Annals of Physics}\ }\textbf {\bibinfo
  {volume} {326}},\ \bibinfo {pages} {96} (\bibinfo {year} {2011})}\BibitemShut
  {NoStop}%
\bibitem [{\citenamefont {Cramer}\ \emph {et~al.}(2010)\citenamefont {Cramer},
  \citenamefont {Plenio}, \citenamefont {Flammia}, \citenamefont {Somma},
  \citenamefont {Gross}, \citenamefont {Bartlett}, \citenamefont
  {Landon-Cardinal}, \citenamefont {Poulin},\ and\ \citenamefont
  {Liu}}]{Cramer2010}%
  \BibitemOpen
  \bibfield  {author} {\bibinfo {author} {\bibfnamefont {M.}~\bibnamefont
  {Cramer}}, \bibinfo {author} {\bibfnamefont {M.~B.}\ \bibnamefont {Plenio}},
  \bibinfo {author} {\bibfnamefont {S.~T.}\ \bibnamefont {Flammia}}, \bibinfo
  {author} {\bibfnamefont {R.}~\bibnamefont {Somma}}, \bibinfo {author}
  {\bibfnamefont {D.}~\bibnamefont {Gross}}, \bibinfo {author} {\bibfnamefont
  {S.~D.}\ \bibnamefont {Bartlett}}, \bibinfo {author} {\bibfnamefont
  {O.}~\bibnamefont {Landon-Cardinal}}, \bibinfo {author} {\bibfnamefont
  {D.}~\bibnamefont {Poulin}}, \ and\ \bibinfo {author} {\bibfnamefont {Y.-K.}\
  \bibnamefont {Liu}},\ }\href {\doibase 10.1038/ncomms1147} {\bibfield
  {journal} {\bibinfo  {journal} {Nat. Comm.}\ }\textbf {\bibinfo {volume}
  {1}},\ \bibinfo {pages} {149} (\bibinfo {year} {2010})}\BibitemShut {NoStop}%
\bibitem [{\citenamefont {Torlai}\ \emph {et~al.}(2018)\citenamefont {Torlai},
  \citenamefont {Mazzola}, \citenamefont {Carrasquilla}, \citenamefont
  {Troyer}, \citenamefont {Melko},\ and\ \citenamefont {Carleo}}]{Torlai2018}%
  \BibitemOpen
  \bibfield  {author} {\bibinfo {author} {\bibfnamefont {G.}~\bibnamefont
  {Torlai}}, \bibinfo {author} {\bibfnamefont {G.}~\bibnamefont {Mazzola}},
  \bibinfo {author} {\bibfnamefont {J.}~\bibnamefont {Carrasquilla}}, \bibinfo
  {author} {\bibfnamefont {M.}~\bibnamefont {Troyer}}, \bibinfo {author}
  {\bibfnamefont {R.}~\bibnamefont {Melko}}, \ and\ \bibinfo {author}
  {\bibfnamefont {G.}~\bibnamefont {Carleo}},\ }\href {\doibase
  10.1038/s41567-018-0048-5} {\bibfield  {journal} {\bibinfo  {journal} {Nat.
  Phys.}\ }\textbf {\bibinfo {volume} {14}},\ \bibinfo {pages} {447} (\bibinfo
  {year} {2018})}\BibitemShut {NoStop}%
\bibitem [{\citenamefont {Torlai}\ \emph {et~al.}(2019)\citenamefont {Torlai},
  \citenamefont {Timar}, \citenamefont {van Nieuwenburg}, \citenamefont
  {Levine}, \citenamefont {Omran}, \citenamefont {Keesling}, \citenamefont
  {Bernien}, \citenamefont {Greiner}, \citenamefont {{Vuleti\ifmmode
  \acutec\else {\'{c}}\fi}}, \citenamefont {Lukin}, \citenamefont {Melko},\
  and\ \citenamefont {Endres}}]{Torlai2019}%
  \BibitemOpen
  \bibfield  {author} {\bibinfo {author} {\bibfnamefont {G.}~\bibnamefont
  {Torlai}}, \bibinfo {author} {\bibfnamefont {B.}~\bibnamefont {Timar}},
  \bibinfo {author} {\bibfnamefont {E.~P.~L.}\ \bibnamefont {van Nieuwenburg}},
  \bibinfo {author} {\bibfnamefont {H.}~\bibnamefont {Levine}}, \bibinfo
  {author} {\bibfnamefont {A.}~\bibnamefont {Omran}}, \bibinfo {author}
  {\bibfnamefont {A.}~\bibnamefont {Keesling}}, \bibinfo {author}
  {\bibfnamefont {H.}~\bibnamefont {Bernien}}, \bibinfo {author} {\bibfnamefont
  {M.}~\bibnamefont {Greiner}}, \bibinfo {author} {\bibfnamefont
  {V.}~\bibnamefont {{Vuleti\ifmmode \acutec\else {\'{c}}\fi}}}, \bibinfo
  {author} {\bibfnamefont {M.~D.}\ \bibnamefont {Lukin}}, \bibinfo {author}
  {\bibfnamefont {R.~G.}\ \bibnamefont {Melko}}, \ and\ \bibinfo {author}
  {\bibfnamefont {M.}~\bibnamefont {Endres}},\ }\href {\doibase
  10.1103/PhysRevLett.123.230504} {\bibfield  {journal} {\bibinfo  {journal}
  {Phys. Rev. Lett.}\ }\textbf {\bibinfo {volume} {123}},\ \bibinfo {pages}
  {230504} (\bibinfo {year} {2019})}\BibitemShut {NoStop}%
\bibitem [{\citenamefont {Kokail}\ \emph {et~al.}(2021)\citenamefont {Kokail},
  \citenamefont {van Bijnen}, \citenamefont {Elben}, \citenamefont
  {Vermersch},\ and\ \citenamefont {Zoller}}]{Kokail2021}%
  \BibitemOpen
  \bibfield  {author} {\bibinfo {author} {\bibfnamefont {C.}~\bibnamefont
  {Kokail}}, \bibinfo {author} {\bibfnamefont {R.}~\bibnamefont {van Bijnen}},
  \bibinfo {author} {\bibfnamefont {A.}~\bibnamefont {Elben}}, \bibinfo
  {author} {\bibfnamefont {B.}~\bibnamefont {Vermersch}}, \ and\ \bibinfo
  {author} {\bibfnamefont {P.}~\bibnamefont {Zoller}},\ }\href {\doibase
  10.1038/s41567-021-01260-w} {\bibfield  {journal} {\bibinfo  {journal}
  {Nature Physics}\ }\textbf {\bibinfo {volume} {17}},\ \bibinfo {pages}
  {936–942} (\bibinfo {year} {2021})}\BibitemShut {NoStop}%
\bibitem [{\citenamefont {Carleo}\ and\ \citenamefont
  {Troyer}(2017)}]{Carleo2017}%
  \BibitemOpen
  \bibfield  {author} {\bibinfo {author} {\bibfnamefont {G.}~\bibnamefont
  {Carleo}}\ and\ \bibinfo {author} {\bibfnamefont {M.}~\bibnamefont
  {Troyer}},\ }\href {\doibase 10.1126/science.aag2302} {\bibfield  {journal}
  {\bibinfo  {journal} {Science}\ }\textbf {\bibinfo {volume} {355}},\ \bibinfo
  {pages} {602} (\bibinfo {year} {2017})}\BibitemShut {NoStop}%
\bibitem [{\citenamefont {Flammia}\ and\ \citenamefont
  {Liu}(2011)}]{Flammia2011}%
  \BibitemOpen
  \bibfield  {author} {\bibinfo {author} {\bibfnamefont {S.~T.}\ \bibnamefont
  {Flammia}}\ and\ \bibinfo {author} {\bibfnamefont {Y.-K.}\ \bibnamefont
  {Liu}},\ }\href {\doibase 10.1103/PhysRevLett.106.230501} {\bibfield
  {journal} {\bibinfo  {journal} {Phys. Rev. Lett.}\ }\textbf {\bibinfo
  {volume} {106}},\ \bibinfo {pages} {230501} (\bibinfo {year}
  {2011})}\BibitemShut {NoStop}%
\bibitem [{\citenamefont {da~Silva}\ \emph {et~al.}(2011)\citenamefont
  {da~Silva}, \citenamefont {Landon-Cardinal},\ and\ \citenamefont
  {Poulin}}]{DaSilva2011}%
  \BibitemOpen
  \bibfield  {author} {\bibinfo {author} {\bibfnamefont {M.~P.}\ \bibnamefont
  {da~Silva}}, \bibinfo {author} {\bibfnamefont {O.}~\bibnamefont
  {Landon-Cardinal}}, \ and\ \bibinfo {author} {\bibfnamefont {D.}~\bibnamefont
  {Poulin}},\ }\href {\doibase 10.1103/PhysRevLett.107.210404} {\bibfield
  {journal} {\bibinfo  {journal} {Phys. Rev. Lett.}\ }\textbf {\bibinfo
  {volume} {107}},\ \bibinfo {pages} {210404} (\bibinfo {year}
  {2011})}\BibitemShut {NoStop}%
\bibitem [{Note2()}]{Note2}%
  \BibitemOpen
  \bibinfo {note} {The precise master equation is given in Ref.~\cite
  {Brydges2019}, the propagation time is $t=5$ ms}\BibitemShut {NoStop}%
\bibitem [{\citenamefont {Han}\ \emph {et~al.}(2018)\citenamefont {Han},
  \citenamefont {Wang}, \citenamefont {Fan}, \citenamefont {Wang},\ and\
  \citenamefont {Zhang}}]{Han2018}%
  \BibitemOpen
  \bibfield  {author} {\bibinfo {author} {\bibfnamefont {Z.-Y.}\ \bibnamefont
  {Han}}, \bibinfo {author} {\bibfnamefont {J.}~\bibnamefont {Wang}}, \bibinfo
  {author} {\bibfnamefont {H.}~\bibnamefont {Fan}}, \bibinfo {author}
  {\bibfnamefont {L.}~\bibnamefont {Wang}}, \ and\ \bibinfo {author}
  {\bibfnamefont {P.}~\bibnamefont {Zhang}},\ }\href {\doibase
  10.1103/PhysRevX.8.031012} {\bibfield  {journal} {\bibinfo  {journal} {Phys.
  Rev. X}\ }\textbf {\bibinfo {volume} {8}},\ \bibinfo {pages} {31012}
  (\bibinfo {year} {2018})}\BibitemShut {NoStop}%
\bibitem [{\citenamefont {Zaletel}\ \emph {et~al.}(2015)\citenamefont
  {Zaletel}, \citenamefont {Mong}, \citenamefont {Karrasch}, \citenamefont
  {Moore},\ and\ \citenamefont {Pollmann}}]{Zaletel2015}%
  \BibitemOpen
  \bibfield  {author} {\bibinfo {author} {\bibfnamefont {M.~P.}\ \bibnamefont
  {Zaletel}}, \bibinfo {author} {\bibfnamefont {R.~S.}\ \bibnamefont {Mong}},
  \bibinfo {author} {\bibfnamefont {C.}~\bibnamefont {Karrasch}}, \bibinfo
  {author} {\bibfnamefont {J.~E.}\ \bibnamefont {Moore}}, \ and\ \bibinfo
  {author} {\bibfnamefont {F.}~\bibnamefont {Pollmann}},\ }\href {\doibase
  10.1103/PhysRevB.91.165112} {\bibfield  {journal} {\bibinfo  {journal} {Phys.
  Rev. B}\ }\textbf {\bibinfo {volume} {91}},\ \bibinfo {pages} {165112}
  (\bibinfo {year} {2015})},\ \Eprint {http://arxiv.org/abs/1407.1832}
  {1407.1832} \BibitemShut {NoStop}%
\bibitem [{\citenamefont {Fishman}\ \emph {et~al.}()\citenamefont {Fishman},
  \citenamefont {White},\ and\ \citenamefont {Stoudenmire}}]{Fishman}%
  \BibitemOpen
  \bibfield  {author} {\bibinfo {author} {\bibfnamefont {M.}~\bibnamefont
  {Fishman}}, \bibinfo {author} {\bibfnamefont {S.~R.}\ \bibnamefont {White}},
  \ and\ \bibinfo {author} {\bibfnamefont {E.~M.}\ \bibnamefont
  {Stoudenmire}},\ }\href {http://arxiv.org/abs/2007.14822} {\ }\Eprint
  {http://arxiv.org/abs/2007.14822} {arXiv:2007.14822} \BibitemShut {NoStop}%
\bibitem [{\citenamefont {Johansson}\ \emph {et~al.}(2013)\citenamefont
  {Johansson}, \citenamefont {Nation},\ and\ \citenamefont
  {Nori}}]{Johansson2013}%
  \BibitemOpen
  \bibfield  {author} {\bibinfo {author} {\bibfnamefont {J.}~\bibnamefont
  {Johansson}}, \bibinfo {author} {\bibfnamefont {P.}~\bibnamefont {Nation}}, \
  and\ \bibinfo {author} {\bibfnamefont {F.}~\bibnamefont {Nori}},\ }\href
  {\doibase 10.1016/j.cpc.2012.11.019} {\bibfield  {journal} {\bibinfo
  {journal} {Comput. Phys. Commun.}\ }\textbf {\bibinfo {volume} {184}},\
  \bibinfo {pages} {1234} (\bibinfo {year} {2013})}\BibitemShut {NoStop}%
\bibitem [{\citenamefont {Hoeffding}(1992)}]{Hoeffding1992}%
  \BibitemOpen
  \bibfield  {author} {\bibinfo {author} {\bibfnamefont {W.}~\bibnamefont
  {Hoeffding}},\ }in\ \href@noop {} {\emph {\bibinfo {booktitle} {Breakthroughs
  in Statistics}}}\ (\bibinfo  {publisher} {Springer},\ \bibinfo {year}
  {1992})\ pp.\ \bibinfo {pages} {308--334}\BibitemShut {NoStop}%
\bibitem [{Note3()}]{Note3}%
  \BibitemOpen
  \bibinfo {note} {Using Weingarten calculus and techniques presented in
  Ref.~\cite {Zhou2020}, we can generalize the variance for uniform sampling to
  pure product states of $N$ qudits with arbitrary local dimension $d$ (i.e.\
  including the case of global random unitaries). We find $$ \Gamma _4=\left
  (\protect \frac {d^2+9 d+2}{d^2+5 d+6} \right )^N, \Gamma _3=\left (\protect
  \frac {3d}{2+d}\right )^N, \Gamma _2=\left (2d-1\right )^N . $$}\BibitemShut
  {NoStop}%
\end{thebibliography}
\end{document}